\documentclass[11pt,notitlepage,tightenlines,nofootinbib,superscriptaddress]{revtex4-2}
\pdfoutput=1

\pdfoutput=1
\usepackage[utf8]{inputenc}
\usepackage{amsthm}
\usepackage{amssymb}
\usepackage{amsmath}

\usepackage[colorlinks=true, linkcolor=MidnightBlue, urlcolor=MidnightBlue, citecolor=MidnightBlue, linktoc=section, pdftitle={No second law of entanglement manipulation after all}, pdfauthor={Ludovico Lami and Bartosz Regula}]{hyperref}

\usepackage{newpxtext,newpxmath}

\usepackage{bbold}
\usepackage{bbm}
\usepackage{nonfloat}
\usepackage{braket}
\usepackage{dsfont}
\usepackage{mathdots}
\usepackage{mathtools}
\usepackage{enumerate}
\usepackage{csquotes}
\usepackage{stmaryrd}
\usepackage[cal=boondox]{mathalfa}
\usepackage{graphicx}
\usepackage{stackengine}
\usepackage{scalerel}
\usepackage[dvipsnames]{xcolor}
\usepackage{array,xpatch}
\usepackage{makecell}
\newcolumntype{x}[1]{>{\centering\arraybackslash}p{#1}}
\usepackage{tikz}
\usepackage{pgfplots}
\usetikzlibrary{shapes.geometric, shapes.misc, positioning, arrows, decorations.pathreplacing, decorations.text, angles, quotes, backgrounds, calc}
\usepackage{booktabs}
\usepackage{xfrac}
\usepackage{siunitx}
\usepackage{centernot}
\usepackage{comment}
\usepackage{chngcntr}

\newtheorem{thm}{Theorem}
\newtheorem*{thm*}{Theorem}
\newtheorem{prop}[thm]{Proposition}
\newtheorem*{prop*}{Proposition}
\newtheorem{lemma}[thm]{Lemma}
\newtheorem*{lemma*}{Lemma}

\newtheorem*{cor*}{Corollary}
\newtheorem{cj}[thm]{Conjecture}
\newtheorem*{cj*}{Conjecture}
\newtheorem{Def}[thm]{Definition}
\newtheorem*{Def*}{Definition}

\makeatletter
\def\thmhead@plain#1#2#3{%
  \thmname{#1}\thmnumber{\@ifnotempty{#1}{ }\@upn{#2}}%
  \thmnote{ {\the\thm@notefont#3}}}
\let\thmhead\thmhead@plain
\makeatother

\theoremstyle{definition}
\newtheorem{rem}[thm]{Remark}
\newtheorem*{rem*}{Remark}
\newtheorem*{note}{Note}

\newcommand{\bb}{\begin{equation}\begin{aligned}}
\newcommand{\bbb}{\begin{equation*}\begin{aligned}}
\newcommand{\ee}{\end{aligned}\end{equation}}
\newcommand{\eee}{\end{aligned}\end{equation*}}
\newcommand\floor[1]{\left\lfloor#1\right\rfloor}
\newcommand\ceil[1]{\left\lceil#1\right\rceil}

\let\textleq\relax
\let\textgeq\relax
\let\texteq\relax

\newcommand{\texteq}[1]{\stackrel{\mathclap{\scriptsize \mbox{#1}}}{=}}
\newcommand{\textleq}[1]{\stackrel{\mathclap{\scriptsize \mbox{#1}}}{\leq}}

\newcommand{\textgeq}[1]{\stackrel{\mathclap{\scriptsize \mbox{#1}}}{\geq}}
\newcommand{\ketbra}[1]{\ket{#1}\!\!\bra{#1}}
\newcommand{\ketbraa}[2]{\ket{#1}\!\!\bra{#2}}
\newcommand{\sumno}{\sum\nolimits}

\newcommand{\id}{\mathds{1}}
\newcommand{\R}{\mathds{R}}
\newcommand{\N}{\mathds{N}}
\newcommand{\C}{\mathds{C}}

\newcommand{\cptp}{\mathrm{CPTP}}
\newcommand{\ptp}{\mathrm{PTP}}

\newcommand{\locc}{\mathrm{LOCC}}

\newcommand{\ppt}{\mathrm{PPT}}

\DeclareMathOperator{\Tr}{Tr}
\DeclareMathOperator{\rk}{rk}
\DeclareMathOperator{\cl}{cl}
\DeclareMathOperator{\co}{conv}
\DeclareMathOperator{\cone}{cone}

\DeclareMathAlphabet{\pazocal}{OMS}{zplm}{m}{n}

\DeclareMathOperator{\supp}{supp}

\newcommand{\HH}{\pazocal{H}}
\newcommand{\T}{\pazocal{T}}
\newcommand{\B}{\pazocal{B}}

\newcommand{\Tp}{\pazocal{T}_+}
\newcommand{\D}{\pazocal{D}}
\newcommand{\K}{\pazocal{K}}

\newcommand{\Bsa}{\pazocal{B}_{\mathrm{sa}}}

\newcommand{\lsmatrix}{\left(\begin{smallmatrix}}
\newcommand{\rsmatrix}{\end{smallmatrix}\right)}

\stackMath
\newcommand\xxrightarrow[2][]{\mathrel{%
  \setbox2=\hbox{\stackon{\scriptstyle #1}{\scriptstyle#2}}%
  \stackunder[4pt]{%
    \xrightarrow{\makebox[\dimexpr\wd2\relax]{$\scriptstyle#2$}}%
  }{%
   \scriptstyle#1\,%
  }%
}}

\newcommand{\tends}[2]{\xxrightarrow[#2]{\mathrm{#1}}}
\newcommand{\tendsn}[1]{\xxrightarrow[\! n\rightarrow \infty\!]{#1}}

\renewcommand{\tends}[2]{\xxrightarrow[#2]{\vphantom{.}\smash{{\raisebox{-1.8pt}{\scriptsize{#1}}}}}}

\stackMath

\makeatletter
\newcommand*\rel@kern[1]{\kern#1\dimexpr\macc@kerna}
\newcommand*\widebar[1]{%
  \begingroup
  \def\mathaccent##1##2{%
    \rel@kern{0.8}%
    \overline{\rel@kern{-0.8}\macc@nucleus\rel@kern{0.2}}%
    \rel@kern{-0.2}%
  }%
  \macc@depth\@ne
  \let\math@bgroup\@empty \let\math@egroup\macc@set@skewchar
  \mathsurround\z@ \frozen@everymath{\mathgroup\macc@group\relax}%
  \macc@set@skewchar\relax
  \let\mathaccentV\macc@nested@a
  \macc@nested@a\relax111{#1}%
  \endgroup
}

\definecolor{Blues5seq1}{RGB}{239,243,255}
\definecolor{Blues5seq2}{RGB}{189,215,231}
\definecolor{Blues5seq3}{RGB}{107,174,214}
\definecolor{Blues5seq4}{RGB}{49,130,189}
\definecolor{Blues5seq5}{RGB}{8,81,156}

\definecolor{Greens5seq1}{RGB}{237,248,233}
\definecolor{Greens5seq2}{RGB}{186,228,179}
\definecolor{Greens5seq3}{RGB}{116,196,118}
\definecolor{Greens5seq4}{RGB}{49,163,84}
\definecolor{Greens5seq5}{RGB}{0,109,44}

\definecolor{Reds5seq1}{RGB}{254,229,217}
\definecolor{Reds5seq2}{RGB}{252,174,145}
\definecolor{Reds5seq3}{RGB}{251,106,74}
\definecolor{Reds5seq4}{RGB}{222,45,38}
\definecolor{Reds5seq5}{RGB}{165,15,21}

\definecolor{Yellow}{RGB}{233, 233, 150}
\definecolor{Orange}{RGB}{255, 230, 179}
\definecolor{Blueish}{RGB}{192, 242, 217}

\definecolor{goldenyellow}{rgb}{1.0, 0.87, 0.0}

\usetikzlibrary{decorations.text,arrows.meta,bending,shapes.misc}

\tikzset{cross/.style={cross out, draw=black, minimum size=2*(#1-\pgflinewidth), inner sep=0pt, outer sep=0pt},
cross/.default={1pt}}

\pgfdeclarearrow{
  name = mytriangle,
  parameters = { \the\pgfarrowlength, \the\pgfarrowwidth, \the\pgfarrowlinewidth, \ifpgfarrowroundcap c\fi },
  defaults = { length = 0pt .85, width = 0pt 1.3, line width = 0pt .1, round = true },
  setup code = {
    \pgfarrowssettipend{1\pgfarrowlength}
    \pgfarrowssetlineend{.2\pgfarrowlength}
    \pgfarrowssetvisualbackend{0\pgfarrowlength}
    \pgfarrowssetbackend{0\pgfarrowwidth}
    \pgfarrowshullpoint{1\pgfarrowlength}{0\pgfarrowwidth}
    \pgfarrowsupperhullpoint{0\pgfarrowlength}{.5\pgfarrowwidth}
    \pgfarrowssavethe\pgfarrowwidth
    \pgfarrowssavethe\pgfarrowlength
    \pgfarrowssavethe\pgfarrowlinewidth
  },
  drawing code = {
    \ifpgfarrowroundcap\pgfsetroundjoin\fi
    \pgfsetlinewidth{\pgfarrowlinewidth}
    \pgfpathmoveto{\pgfqpoint{1\pgfarrowlength}{0\pgfarrowwidth}}
    \pgfpathlineto{\pgfpoint{0\pgfarrowlength}{.5\pgfarrowwidth}}
    \pgfpathlineto{\pgfqpoint{0\pgfarrowlength}{-.5\pgfarrowwidth}}
    \pgfpathclose
    \pgfusepathqfillstroke
  },
}

\makeatletter
\def\l@subsection#1#2{}
\def\l@subsubsection#1#2{}
\makeatother

\makeatletter
\patchcmd{\@ssect@ltx}
    {\addcontentsline{toc}{#1}{\protect\numberline{}#8}}
    {}
    {}
    {}
\makeatother

\makeatletter

\def\@sect@ltx#1#2#3#4#5#6[#7]#8{%
    \@ifnum{#2>\c@secnumdepth}{%
        \def\H@svsec{\phantomsection}%
        \let\@svsec\@empty
    }{%
        \H@refstepcounter{#1}%
        \def\H@svsec{%
            \phantomsection
        }%
        \protected@edef\@svsec{{#1}}%
        \@ifundefined{@#1cntformat}{%
            \prepdef\@svsec\@seccntformat
        }{%
            \expandafter\prepdef
            \expandafter\@svsec
            \csname @#1cntformat\endcsname
        }%
    }%
    \@tempskipa #5\relax
    \@ifdim{\@tempskipa>\z@}{%
        \begingroup
        \interlinepenalty \@M
        #6{%
            \@ifundefined{@hangfrom@#1}{\@hang@from}{\csname @hangfrom@#1\endcsname}%
            {\hskip#3\relax\H@svsec}{\@svsec}{#8}%
        }%
        \@@par
        \endgroup
        \@ifundefined{#1mark}{\@gobble}{\csname #1mark\endcsname}{#7}%
        \addcontentsline{toc}{#1}{%
            \@ifnum{#2>\c@secnumdepth}{%
                \protect\numberline{}%
            }{%
                \protect\numberline{\csname the#1\endcsname}%
            }%
            #7}
    }{%
        \def\@svsechd{%
            #6{%
                \@ifundefined{@runin@to@#1}{\@runin@to}{\csname @runin@to@#1\endcsname}%
                {\hskip#3\relax\H@svsec}{\@svsec}{#8}%
            }%
            \@ifundefined{#1mark}{\@gobble}{\csname #1mark\endcsname}{#7}%
            \addcontentsline{toc}{#1}{%
                \@ifnum{#2>\c@secnumdepth}{%
                    \protect\numberline{}%
                }{%
                    \protect\numberline{\csname the#1\endcsname}%
                }%
                #8}%
        }%
    }%
    \@xsect{#5}}%

\makeatother

\pgfplotsset{compat=1.17}

\newcommand{\SEP}{\pazocal{S}}
\newcommand{\PPT}{\pazocal{P\!P\!T}}
\newcommand{\sepp}{\mathrm{NE}}
\newcommand{\pptp}{\mathrm{PPTP}}
\newcommand{\kp}{\mathrm{KP}}

\newcommand{\ten}{E^\tau_N}
\newcommand{\tl}{L^\tau}
\newcommand{\tsr}{R^\tau}
\newcommand{\tn}{N_\tau}

\newcommand{\wt}{\widetilde}

\newcommand{\lset}{\left\{}
\newcommand{\rset}{\right\}}
\renewcommand{\bar}{\,:\,}
\newcommand{\idc}{\mathrm{id}}

\let\epsilon\varepsilon

\DeclareMathAlphabet\mathbfcal{OMS}{cmsy}{b}{n}

\begin{document}

\title{
No second law of entanglement manipulation after all
}

\author{Ludovico Lami}
\email{ludovico.lami@gmail.com}
\affiliation{Institut f\"{u}r Theoretische Physik und IQST, Universit\"{a}t Ulm, Albert-Einstein-Allee 11, D-89069 Ulm, Germany}
\affiliation{QuSoft, Korteweg--de Vries Institute for Mathematics, and Institute for Theoretical Physics, University of Amsterdam, the Netherlands}

\author{Bartosz Regula}
\email{bartosz.regula@gmail.com}
\affiliation{Department of Physics, Graduate School of Science, The University of Tokyo, Bunkyo-ku, Tokyo 113-0033, Japan}
\affiliation{School of Physical and Mathematical Sciences, Nanyang Technological University, 637371, Singapore}

\begin{abstract}
Many fruitful analogies have emerged between the theories of quantum entanglement and thermodynamics, motivating the pursuit of an axiomatic description of entanglement akin to the laws of thermodynamics. A long-standing open problem has been to establish a true second law of entanglement, and in particular a unique function which governs all transformations between entangled systems, mirroring the role of entropy in thermodynamics. Contrary to previous promising evidence, here we show that this is impossible, and no direct counterpart to the second law of thermodynamics can be established. This is accomplished by demonstrating the irreversibility of entanglement theory from first principles --- assuming only the most general microscopic physical constraints of entanglement manipulation, we show that entanglement theory is irreversible under all non-entangling transformations. We furthermore rule out reversibility without significant entanglement expenditure, showing that reversible entanglement transformations require the generation of macroscopically large amounts of entanglement according to certain measures. Our results not only reveal fundamental differences between quantum entanglement transformations and thermodynamic processes, but also showcase a unique property of entanglement which distinguishes it from other known quantum resources.
\end{abstract}

\maketitle

Thermodynamics is perhaps the only physical theory that has withstood the several revolutions that have overturned the scientific paradigm since its inception. It started as a phenomenological theory of heat engines, in which context the first and second laws were first formulated, and it has since evolved to encompass general relativity and quantum mechanics. Arguably, its special status stems from its meta-theoretic character: at its root, thermodynamics is a framework to decide which transformations a closed system can or cannot undergo, independently of the underlying physics.
In accordance with this view, axiomatic approaches have played an important role in the development of thermodynamics, from Clausius's~\cite{Clausius} and Kelvin's~\cite{Kelvin} formulations of the second law to the groundbreaking work of Carath\'eodory~\cite{Caratheodory}, Giles~\cite{GILES}, and recently Lieb and Yngvason~\cite{Lieb-Yngvason}. A key feature and strength of these approaches is that only generic assumptions are made on the physical laws that govern the systems under consideration. Early statements of the second law posit the absolute physical impossibility of realising certain transformations, and already in the minds of Carnot, Clausius, and Kelvin were intended to hold equally well e.g.\ for mechanical and electromagnetic processes~\cite{CARNOT,Kelvin}.
A remarkable success of the axiomatic approach is to arrive at an abstract construction of the entropy as the unique function $S$ that encodes all transformations between comparable equilibrium states: given two such states $X$ and $Y$, $X$ can be transformed into $Y$ adiabatically if and only if $S(X)\leq S(Y)$~\cite{GILES, Lieb-Yngvason}. A logical consequence is that comparable states with the same entropy must be connected by reversible transformations, e.g.\ Carnot cycles~\cite{CARNOT}.

With the advent of quantum information science, the phenomenon of quantum entanglement emerged as a physical resource in its own right~\cite{Horodecki-review}, enabling significant advantages in tasks such as communication~\cite{dense-coding,teleportation,Bennett-error-correction}, computation~\cite{rauss}, and cryptography~\cite{Ekert91}. 
The parallel with thermodynamics prompted a debate concerning the axiomatisation of entanglement theory~\cite{Popescu1997, Vedral1998, Vidal2000,Horodecki2002} and the possible emergence of a single entanglement measure, akin to entropy, which would govern all entanglement transformations and establish the reversibility of this resource~\cite{Popescu1997, Horodecki2000, Horodecki2002, vedral_2002}. 
Although later results suggested that entanglement may often be quite different from thermodynamics, even exhibiting irreversibility in some of the most practically relevant settings~\cite{Vidal-irreversibility,bennett_2000}, hope persisted for an axiomatic framework for entanglement manipulation that would exactly mirror thermodynamic properties. Notably, identifying a unique entropic measure of entanglement was long known to be possible for the special case of 
pure states~\cite{Bennett-distillation,Popescu1997}, and several proposals for general reversible frameworks have been formulated~\cite{Horodecki2002, Martin-exact-PPT, BrandaoPlenio1}. The seminal work of Brand\~{a}o and Plenio~\cite{BrandaoPlenio1, BrandaoPlenio2} then provided 
further evidence in this direction by showing that reversible manipulation may~\cite{berta_gap} be possible when the physical restrictions governing entanglement transformations are suitably relaxed. These findings strengthened the belief that a fully reversible and physically consistent theory of entanglement could be established.

Here, however, we prove a general no-go result which shows that entanglement theory is fundamentally irreversible. Equivalently, we show from first principles that entanglement transformations cannot be governed by a single measure, and that an axiomatic second law of entanglement manipulation cannot be established.

Our sole assumption is that entanglement manipulation by separated parties should be accomplished by means of operations that make the theory fully consistent, namely, that \emph{never transform an unentangled system into an entangled one.}
This can be thought of as the analogue in the entanglement setting of the Kelvin--Planck statement of the second law, which in classical thermodynamics forbids the creation of resources (work) from objects which are not resourceful themselves (a single heat bath)~\cite{Kelvin, PLANCK}.
By imposing only this requirement, we dispense with the need to make any assumptions about the structure of the considered processes: for example, we do not even posit that all intermediate transformations obey the laws of standard quantum mechanics, as previous works implicitly did. Instead, we only look at the initial and final states of the system, and demand that no resource, in this case entanglement, is generated in the overall transformation.
This philosophy, hereafter termed \emph{axiomatic}, is analogous to that followed by the pioneers of thermodynamics --- and more recently by Lieb and Yngvason~\cite{Lieb-Yngvason} --- to establish truly universal versions of the second law. Such a general approach allows us to preclude the reversibility of entanglement under all physically-motivated manipulation protocols.

Importantly, however, our conclusions remain unaffected even when the above assumptions are substantially relaxed. 
It is intuitive to ask whether irreversibility could be avoided with just a small amount of generated entanglement, restoring the hope for reversible transformations in practice. We disprove such a possibility by strengthening our result to show that, with a suitable choice of an entanglement measure such as the entanglement negativity~\cite{negativity}, it is necessary to generate macroscopically large quantities of entanglement in the process --- any smaller amount cannot break the fundamental irreversibility revealed in our work. In particular, as we argue below, macroscopic entanglement generation is the price one 
would have to pay in Brand\~{a}o and Plenio's framework~\cite{BrandaoPlenio1, BrandaoPlenio2} to restore reversibility.

The most surprising aspect of our findings is not only the stark contrast with thermodynamics, but also the fact that 
several other quantum phenomena --- including quantum coherence and purity --- have been shown to be reversible in analogous axiomatic settings~\cite{RT-review}, and no quantum resource has ever been found to be irreversible under similar assumptions.
Our result is thus a first of its kind: it highlights a fundamental difference between entanglement on one side, and thermodynamics and all other quantum resource theories known to date on the other.

The generality of our approach allows for an extension of the results beyond the theory of entanglement of quantum \emph{states}, to the manipulation of quantum \emph{operations}~\cite{no_second_law_channels}. This corresponds to the setting of quantum communication, where the resource in consideration is the ability to reliably transmit quantum systems.
Importantly, thermodynamics allows for reversible manipulation of operations~\cite{Faist2019} as well, so an irreversibility of communication theory is, once again, in heavy contrast with thermodynamics.

\section*{Entanglement manipulation}

The framework of entanglement theory features two separated parties, conventionally named Alice and Bob, who share a large number of identical copies of a bipartite quantum state, and wish to transform them into as many copies as possible of some target state, all while making a vanishingly small error in the asymptotic limit. We introduce this setting in Fig.~\ref{distillation_fig}.
\begin{figure}[h]
\includegraphics{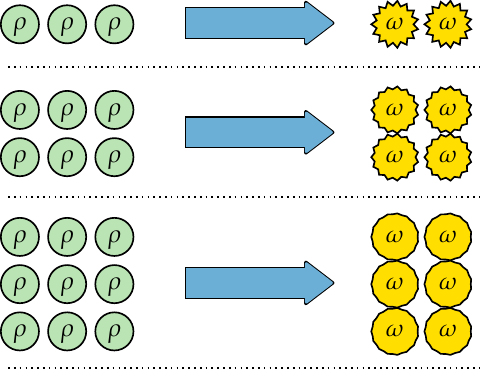}
\caption{\textbf{Asymptotic state conversion.} Here, an entanglement transformation protocol allows us to obtain two copies of a target state $\omega$ for every three copies of an initial state $\rho$, with the transformation error improving as more copies of $\rho$ are provided.\newline
More generally, the initial global state is represented by an $n$-fold tensor product $\rho_{AB}^{\otimes n}$, where $\rho_{AB}$ is a density operator on some tensor product separable Hilbert space $\HH_A\otimes \HH_B$. In contrast with previous works, we do not assume that such Hilbert space is finite-dimensional. By means of some quantum operation $\Lambda: A^n B^n\to A'^m B'^m$ that acts on $n$ copies of $AB$ and outputs $m$ copies of a (different) bipartite system $A'B'$, the initial state will be transformed into $\Lambda \!\left( \rho_{AB}^{\otimes n} \right)$. 
Given a desired target state $\omega_{A'B'}^{\otimes m}$, we thus require that the output state of the protocol be almost indistinguishable from this target state operationally, in the sense that any attempt of discriminating them by means of a quantum measurement should incur an error akin to that of a random guess. By the Helstrom--Holevo theorem~\cite{HELSTROM,Holevo1976}, this property can be captured mathematically by imposing that the distance between the output of the transformation and the 
target state, as quantified by the trace norm $\|\cdot\|_1$, has to vanish. Therefore, by requiring that $\displaystyle \lim_{n\to\infty} \left\| \Lambda \left( \rho_{AB}^{\otimes n} \right) - \omega_{A'B'}^{\otimes m} \right\|_1 = 0$,
we guarantee that the conversion of $\rho_{AB}^{\otimes n}$ into $\omega_{A'B'}^{\otimes m}$ will get increasingly better with more copies of the state $\rho_{AB}$ available, culminating in an asymptotically perfect transformation.
}
\label{distillation_fig}
\end{figure}

The figure of merit in transforming the input quantum state $\rho$ into a target state $\omega$ is the transformation rate $R(\rho\to \omega)$, defined as the maximum ratio $m/n$ that can be achieved in the limit $n\to\infty$ under the condition that $n$ copies of $\rho$ are transformed into $m$ copies of $\omega$ with asymptotically vanishing error. 
Such a rate depends crucially on the set of allowed operations. In keeping with our axiomatic approach, we consider the largest physically consistent class of transformations: namely, those which are incapable of generating entanglement, and can only manipulate entanglement already present in the system.

To formalise this, we introduce the set of separable (or unentangled) states on a bipartite system $AB$, composed of all those states $\sigma_{AB}$ that admit a decomposition of the form~\cite{Werner, Holevo2005}
\begin{equation}
    \sigma_{AB} = \int \ketbra{\psi}_A \otimes \ketbra{\phi}_B\, \mathrm{d}\mu(\psi,\phi)\, ,
    \label{separable_main}
\end{equation}
where $\mu$ is an appropriate probability measure on the set of pairs of local pure states. Our assumption is that any allowed operation $\Lambda$ should transform quantum states on $AB$ into valid quantum states on some output system $A'B'$, in such a way that $\Lambda(\sigma_{AB})$ is separable for all separable states $\sigma_{AB}$. We refer to such operations as \emph{non-entangling} (NE); they are also 
known as \emph{separability preserving}. 
Hereafter, all transformation rates are understood to be with respect to this family of protocols.

We say that two states $\rho,\omega$ can be interconverted reversibly if $R(\rho\to\omega) R(\omega\to\rho)=1$, as visualised in Fig.~\ref{reversibility_fig}. However, to demonstrate or disprove reversibility of entanglement theory as a whole, it is not necessary to check all possible pairs $\rho,\omega$. Instead, we can fix one of the two states, say the second, to be the standard unit of entanglement, the two-qubit maximally entangled 
state (`entanglement bit') $\Phi_2\coloneqq \frac12 \sum_{i,j=1}^2 \ketbraa{ii}{jj}$~\cite{Horodecki-review}. The two quantities $E_d(\rho)\coloneqq R(\rho\to \Phi_2)$ and $E_c(\rho)\coloneqq R(\Phi_2\to \rho)^{-1}$ are referred to as the distillable entanglement and the entanglement cost of $\rho$, respectively. Entanglement theory is then reversible if $E_d(\rho)=E_c(\rho)$ for all states $\rho$.

\begin{figure}[h!]
\includegraphics{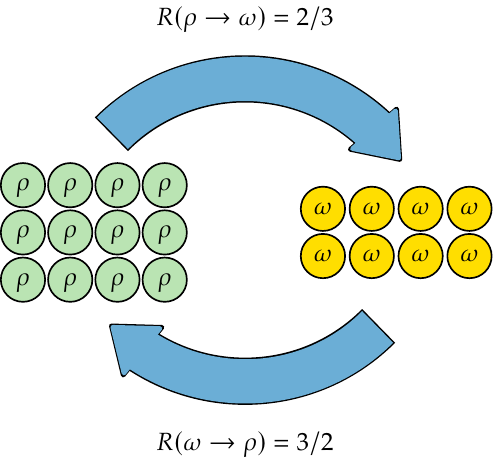}
\caption{
\textbf{Reversible interconversion between two states $\boldsymbol{\rho}$ and $\boldsymbol{\omega}$.} In this example, in the asymptotic many-copy limit it is possible to obtain $2$ copies of $\omega$ from each $3$ copies of $\rho$, and vice versa.}
\label{reversibility_fig}
\end{figure}

\section*{Irreversibility of entanglement manipulation}

By demonstrating an explicit example of a state which cannot be reversibly manipulated, we will show that reversibility of entanglement theory cannot be satisfied in general. We formalise this as follows.

\begin{thm} \label{main_irreversibility_thm}
The theory of entanglement manipulation is irreversible under non-entangling operations. More precisely, for the two-qutrit state $\omega_3 = \frac16 \sum_{i,j=1}^3 \!\left(\ketbraa{ii}{ii} - \ketbraa{ii}{jj}\right)$ it holds that
\begin{equation}
    E_c(\omega_3)=1>E_d(\omega_3)=\log_2(3/2).
\end{equation}
\end{thm}

To show this result, we introduce a general lower bound on the entanglement cost $E_c$ that can be efficiently computed as a semi-definite program. Our approach relies on a new entanglement monotone which we call the \emph{tempered negativity}, defined through a suitable modification of a well-known entanglement measure called the negativity~\cite{negativity}. The situation described by Theorem~\ref{main_irreversibility_thm} is depicted in Fig.~\ref{irreversibility_fig}. The full proof of the result is sketched in the Methods and described in detail in Supplementary Notes \ref{general}--\ref{main_results_note}.

\begin{figure}[h!]
\includegraphics{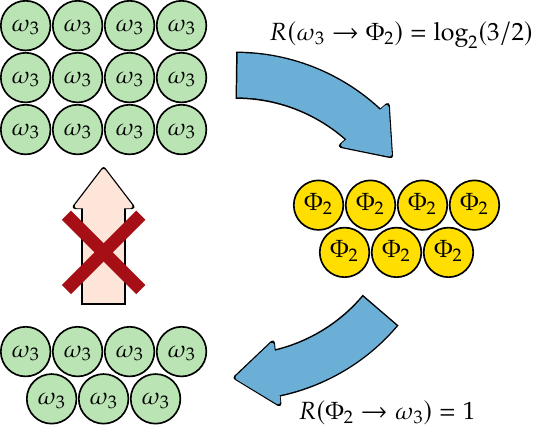}
\caption{\textbf{Irreversibility of entanglement manipulation.} Our main result in Theorem~\ref{main_irreversibility_thm} shows that the two-qutrit state $\omega_3$ cannot be reversibly manipulated under non-entangling transformations: 
we can extract only $\log_2(3/2)\approx 7/12$ entanglement bits per copy of $\omega_3$ asymptotically, but one full entanglement bit per copy is needed to generate it.\\    
Theorem~\ref{main_irreversibility_thm} can be strengthened and extended in several ways, which we overview in the Methods section and expound on in Supplementary Notes \ref{generation_note}--\ref{further_considerations_note}:\\
(1) We show, in particular, that irreversibility persists beyond non-entangling transformations: the 
conclusion of Theorem~\ref{main_irreversibility_thm} holds even when we allow for the generation of small amounts of entanglement (sub-exponential in the number of copies of the state), as quantified by several choices of entanglement measures such as the negativity or the standard robustness of entanglement~\cite{VidalTarrach}. What this means is that, in order to reversibly manipulate the state $\omega_3$, one would need to generate macroscopic (exponential) amounts of entanglement.\\ 
(2) We furthermore show that the irreversibility cannot be alleviated by allowing for a small non-vanishing error in the asymptotic transformation --- a property known as pretty strong converse~\cite{pretty-strong}.\\ 
(3) Finally, Theorem~\ref{main_irreversibility_thm} can also be extended to the theory of point-to-point quantum communication, exploiting the connections between entanglement manipulation and communication schemes~\cite{Bennett-error-correction, Berta2013}. This is considered in detail in a follow-up work~\cite{no_second_law_channels}.
\\
These extensions further solidify the fundamental character of the irreversibility uncovered in our work, showing that it affects both quantum states and channels, and that there are no ways to avoid it without incurring very large transformation errors or generating significant amounts of entanglement.}
\label{irreversibility_fig}
\end{figure}

\section*{Why non-entangling transformations?}

The intention behind our general, axiomatic framework is to prove irreversibility in as broad a setting as possible. The key strength of this approach is that irreversibility under the class of non-entangling transformations 
enforces irreversibility under any smaller class of processes, which includes the vast majority of different types of operations studied in the manipulation of entanglement~\cite{Horodecki-review, RT-review}; furthermore,  our result shows that even enlarging the previously considered classes of processes cannot enable reversibility, as long as the resulting transformations are non-entangling.

To better understand the need for and the consequences of such a general approach, let us compare our framework to another commonly employed model, that where entanglement is manipulated by means of local operations and classical communication (LOCC). In this context, irreversibility was first found by Vidal and Cirac~\cite{Vidal-irreversibility}. Albeit historically important, the LOCC model is built with a `bottom-up' mindset, and rests on the assumption that the two parties can only employ local quantum resources at all stages of the protocol. Already in the early days of quantum information, it was realised that relaxing such restrictions --- e.g.\ by supplying some additional resources --- can lead to improvements in the capability to manipulate entanglement~\cite{Horodecki2002, Martin-exact-PPT}. Although attempts to construct a reversible theory of entanglement along these lines have been unsuccessful~\cite{Horodecki2002, irreversibility-PPT}, the assumptions imposed therein have left open the possibility of the existence of a larger class of operations which could remedy the irreversibility.

The limitations of such bottom-up approaches are best illustrated with a thermodynamical analogy: in this context, they would lead to operational statements of the second law concerning, say, the impossibility of realising certain transformations by means of mechanical processes, but would not tell us much about electrical or nuclear processes. Indeed, since we have no guarantee that the ultimate theory of Nature will be quantum mechanical, it is possible to envision a situation where, for instance, the exploitation of some exotic physical phenomena by one of the parties could enhance entanglement transformations. To construct a theory as powerful as thermodynamics, we followed instead a `top-down', axiomatic approach, which --- as discussed above --- imposes only the weakest possible requirement on the allowed transformations, thereby ruling out reversibility under any physical processes.

The non-entangling operations considered here are examples of `resource non-generating operations', commonly employed in the study of many other quantum resource theories~\cite{RT-review}. In all of these other contexts, such operations have always been shown to lead to the \emph{reversibility} of the given theory. For instance, Gibbs-preserving maps~\cite{Brandao-thermo, horodecki_2013} in quantum thermodynamics are a broad, axiomatic formulation of the constraints governing thermodynamic transformations of quantum systems analogous to non-entangling operations. Under such operations, the theory of thermodynamics is known to be reversible~\cite{Faist2019}. An equivalent result has also been shown in the resource theory of quantum coherence~\cite{Winter2016, one-shot-c-dilution, Chitambar-reversible}, suggesting that reversibility could be a generic feature in the manipulation of different resources under all resource non-generating transformations. Our result, however, shows entanglement theory to be fundamentally different from thermodynamics and from all other known quantum resources: not even the vast class of all non-entangling maps can enable reversible entanglement manipulation. What this means is that, under the exact same assumptions that suffice to facilitate the reversibility of other quantum resources, entanglement remains irreversible.

\section*{Macroscopic entanglement generation is necessary for reversibility}

A similar axiomatic mindset to the one employed in our work has already proved to be useful. Notably, it led Brand\~{a}o and Plenio~\cite{BrandaoPlenio1, BrandaoPlenio2} to construct a theory of entanglement which was claimed to be fully reversible. Recently, an issue that casts some doubts on the validity of their mathematical proof has transpired~\cite{berta_gap}. In spite of this, 
it remains a possibility that the theory of entanglement proposed by Brand\~{a}o and Plenio may actually be reversible~\cite{berta_gap}, so let us discuss it here in detail. 
This theory features so-called asymptotically non-entangling operations, defined as those that may generate some limited amounts of entanglement, provided that any such supplemented resources are vanishingly small in the asymptotic limit.
This, on the surface, appears consistent with how fluctuations are treated in the theory of thermodynamics. However, the key question to ask here is: \emph{according to what measure} should one enforce the generated entanglement to be small? Brand\~{a}o and Plenio choose to quantify entanglement with the generalised robustness~\cite{VidalTarrach}. As we argue below, this a priori arbitrary choice turns out to be crucial to decide between reversibility and irreversibility. That is, there are reasonable entanglement measures using which reversibility only becomes possible at the price of exponential entanglement generation. In fact, even a minor change of the quantifier from the generalised robustness to the closely related standard robustness of entanglement~\cite{VidalTarrach} 
makes reversibility impossible. This entails that 
Brand\~{a}o and Plenio's operations, despite generating vanishingly little entanglement with respect to the generalised robustness, create macroscopically large amounts of it as quantified by other entanglement measures. We show, in fact, that this is not simply an issue with the particular framework of Brand\~{a}o and Plenio~\cite{BrandaoPlenio1, BrandaoPlenio2}, but rather a fundamental property of entanglement theory: \emph{any} attempt to achieve reversibility must necessarily lead to macroscopic entanglement generation.

To make this precise, consider a modified version of asymptotic entanglement manipulation. As previously, given $n$ copies of an initial state $\rho_{AB}$, we want to transform them into $m$ copies of a target state $\omega_{A'B'}$ with asymptotically vanishing error. To define the set of allowed transformations, based on Brand\~{a}o and Plenio's approach~\cite{BrandaoPlenio2}, we fix an entanglement measure $M$ and consider all those transformations $\Lambda_n$ on $n$ copies of the system $AB$ that are $(M,\delta_n)$-approximately non-entangling, in the sense that $M\left(\Lambda_n(\sigma_{A^nB^n})\right)\leq \delta_n$ for all separable states $\sigma_{A^nB^n}$, where the numbers $\delta_n$ quantify the magnitude of the entanglement fluctuations at each step of the process. The maximum ratio $m/n$ that can be achieved in the limit $n\to\infty$ determines the transformation rate under these operations. The modified notion of distillable entanglement, denoted $E_{d,\,\mathrm{NE}^M_{(\delta_n)_n}}\!\!(\rho)$, is then defined by choosing the maximally entangled bit $\Phi_2$ as the target state, and analogously the modified entanglement cost $E_{c,\,\mathrm{NE}^M_{(\delta_n)_n}}\!\!(\rho)$ corresponds to the transformation from $\Phi_2$ to a given state $\rho$. By choosing a suitable measure of entanglement and setting $\delta_n = 0$ for all $n$, we recover our original definition of non-entangling transformations NE.

The problem of choosing what measure $M$ to employ has no straightforward solution, as it is well known that there are many asymptotically inequivalent ways to quantify entanglement~\cite{Horodecki-review}; hence, constraining one such measure cannot guarantee that the supplemented entanglement is truly small according to all measures. From a methodological perspective, this arbitrariness is problematic: 
resource quantifiers should be endowed with an operational interpretation by means of a task defined in purely natural terms; presupposing a particular measure and using it to define the task in the first place makes the framework somewhat contrived, and does not take into consideration what happens when a different monotone is used.

Indeed, the choice of $M$ turns out to be pivotal. Brand\~{a}o and Plenio's main result~\cite{BrandaoPlenio1, BrandaoPlenio2} 
claims~\cite{berta_gap} that, with the specific choice of $M$ being the generalised robustness of entanglement~\cite{VidalTarrach}, entanglement can be manipulated reversibly even if we take $\delta_n \to 0$ as $n \to \infty$. In stark contrast, we now show that a completely opposite conclusion is reached when $M$ is taken to be either the standard robustness~\cite{VidalTarrach} or the entanglement negativity~\cite{negativity}.

\medskip
{
\begin{thm} \label{main_irreversibility_thm_2}
The theory of entanglement manipulation is irreversible under operations that generate sub-exponential amounts of entanglement according to the negativity $N$ or the standard robustness $R^s_\SEP$. Specifically, if $M=N$ or $M=R^s_\SEP$, then for any sequence $(\delta_n)_n$ such that $\delta_n = 2^{o(n)}$ it holds that
\begin{equation}
E_{c\!,\,\mathrm{NE}^M_{(\delta_n)_n}}\!\! (\omega_3)=1>E_{d\!,\,\mathrm{NE}^M_{(\delta_n)_n}}\!\!(\omega_3)=\log_2(3/2)\, .
\end{equation}
\end{thm}}

Comparing this to Brand\~{a}o and Plenio's 
conclusion, we can observe that the operations employed there may only hope to achieve reversibility by generating \emph{exponential} amounts of entanglement, as measured by either the negativity or the standard robustness.

We stress that there is no \emph{a priori} operationally justified reason to prefer the generalised robustness over the other monotones. If anything, the most operationally meaningful monotones to select here would be those defined directly in terms of practical tasks,
such as the entanglement cost $E_c$ itself; however, following this route actually trivialises the theory~\cite{BrandaoPlenio1, BrandaoPlenio2}, entailing that different choices of monotones need to be employed to give meaningful results.
Even between the generalised robustness (as employed by Brand\~ao and Plenio) and the standard robustness $R^s_\SEP$, it is actually the latter that admits a clearer operational interpretation in this context --- $R^s_\SEP$ quantifies exactly the entanglement cost of a state in the one-shot setting~\cite{brandao_2011}, when asymptotic transformations are not allowed. These ambiguities 
in the choice of a `good' measure, and the vastly disparate physical consequences of the different choices, put the physicality of the reversibility result 
claimed by Brand\~ao and Plenio~\cite{BrandaoPlenio1, BrandaoPlenio2} into question: why should one such framework be considered more physical than the other, irreversible ones?

Importantly, since the core concept of separability is independent of the particular choice of a measure, our axiomatic assumption of strict no-entanglement-generation bypasses the above problems completely: it removes the dependence on any entanglement measure and ensures that the physical constraints are enforced at all scales, therefore yielding an unambiguously physical model of general entanglement transformations. However, should such a requirement be considered too strict, our Theorem~\ref{main_irreversibility_thm_2} shows that irreversibility of entanglement is robust to fluctuations in the generated resources.

Let us also point out that the assumptions of Brand\~{a}o and Plenio (and of Theorem~\ref{main_irreversibility_thm_2}) are in fact more permissive than those typically employed in quantum thermodynamic frameworks~\cite{horodecki_2013, weilenmann_2016, thermo-review}, where one usually allows fluctuations in the sense of the consumption of small ancillary resources, but not fluctuations in the very physical laws governing the process. 
In a thermodynamic sense, entanglement transformations under approximately non-entangling maps could be compared to the manipulation of systems under transformations that do not conserve the overall energy --- a relaxation which would go against standard axiomatic assumptions~\cite{weilenmann_2016,thermo-review}. Importantly, no such `unphysical' fluctuations are necessary in order to establish the reversibility of thermodynamics~\cite{horodecki_2013,Faist2019} or other known quantum resources~\cite{RT-review}. We invite the interested reader to Supplementary Note~\ref{intermezzo_note}, where we discuss different notions of resource fluctuations in more detail.

\section*{Discussion}

Our results close a major open problem in the characterisation of entanglement~\cite{OpenProblemArxiv} by showing that a reversible theory of this resource cannot be established under any set of `free' transformations which do not generate entanglement. 
Indeed, from our characterisation we can conclude not only that entanglement generation 
is necessary for reversibility, but also that macroscopically large amounts of it must be supplemented. 
This shows that the framework proposed by Brand\~ao and Plenio~\cite{BrandaoPlenio1, BrandaoPlenio2} is effectively the smallest possible one that could allow reversibility, although only at the cost of significant entanglement expenditure.

That the seemingly small revision of the underlying technical assumptions we advocated by enforcing strict entanglement non-generation can have such far-reaching consequences --- namely, precluding reversibility --- is truly unexpected. 
In fact, as remarked before, the opposite of this phenomenon has been observed in a number of fundamentally important quantum resource theories, where the set of all resource--non-generating operations suffices to enable reversibility.
It is precisely the necessity to generate entanglement in order to reversibly manipulate it that distinguishes the theory of entanglement from thermodynamics and other quantum resources.
This fundamental difference contrasts not only with the previously established information-theoretic parallels, but also with the many links that have emerged between entanglement and thermodynamics in broader contexts such as many-body and relativistic physics~\cite{eisert_2010,eisert_2015,popescu_2006,harlow_2016}.
It then becomes an enthralling foundational problem to understand what makes entanglement theory special in this respect, and where its fundamental irreversibility may come from. Additionally, the axiomatic theory of entanglement manipulation delineated here leaves several outstanding follow-up questions: for instance, it would be very interesting to understand whether a closed expression for the associated entanglement cost can be established, and whether the phenomenon of entanglement catalysis~\cite{Plenio-catalysis,Horodecki-review} can play a role in this setting.

We remark that the recently identified gap in Brand\~ao and Plenio's proof~\cite{berta_gap}, which came to light after this work was completed, does not affect our results or conclusions in any way, since the methods that we use are independent of Refs.~\cite{BrandaoPlenio1, BrandaoPlenio2}. 
Our main finding --- that of entanglement irreversibility under non-entangling operations --- is complementary to the result of Brand\~ao and Plenio~\cite{BrandaoPlenio1, BrandaoPlenio2}, as we discussed above and in Supplementary Notes~\ref{generation_note}--\ref{intermezzo_note}. This recent development does, however, rekindle the question of whether entanglement \emph{can} be reversibly manipulated whatsoever~\cite{OpenProblemArxiv}, even in a more permissive framework such as Brand\~ao and Plenio's.

In conclusion, we have highlighted a fundamental difference between the theory of entanglement manipulation and thermodynamics, proving that no microscopically consistent second law can be established for the former. At its heart, our work reveals an inescapable restriction precipitated by the laws of quantum physics --- one that has no analogue in classical theories, and was previously unknown even within the realm of quantum theory.

\section*{Methods}

In the following we sketch the main ideas needed to arrive at a proof of our main result, Theorem~\ref{main_irreversibility_thm}, and extensions thereof.

\subsection{Asymptotic transformation rates under non-entangling operations}

We start by defining rigorously the fundamental quantities we are dealing with. Given two separable Hilbert spaces $\HH$ and $\HH'$ and the associated spaces of trace class operators $\T(\HH)$ and $\T(\HH')$, a linear map $\Lambda:\T(\HH)\to \T(\HH')$ is said to be \emph{positive and trace preserving} if it transforms density operators on $\HH$ into density operators on $\HH'$. As is well known, physically realisable quantum operations need to be \emph{completely} positive and not merely positive~\cite{NC}. While we could enforce this additional assumption without affecting any of our results, we will only need to assume the positivity of the transformations, establishing limitations also for processes more general than quantum channels.

Since we are dealing with entanglement, we need to make both $\HH$ and $\HH'$ bipartite systems. We shall therefore assume that $\HH = \HH_A\otimes \HH_B$ and $\HH'=\HH_{A'}\otimes \HH_{B'}$ have a tensor product structure. \emph{Separable states} on $AB$ are defined as those that admit a decomposition as in~\eqref{separable_main}. A positive trace-preserving operation $\Lambda: \T(\HH_A\otimes \HH_B)\to \T(\HH_{A'}\otimes \HH_{B'})$, which we shall denote compactly as $\Lambda_{AB\to A'B'}$, is said to be \emph{non-entangling} or \emph{separability-preserving} if it transforms separable states on $AB$ into separable states on $A'B'$. We will denote the set of non-entangling operations from $AB$ to $A'B'$ as $\sepp(AB\to A'B')$.

The central questions in the theory of entanglement manipulation are the following. Given a bipartite state $\rho_{AB}$ and a set of quantum operations, how much entanglement can be extracted from $\rho_{AB}$? How much entanglement does it cost to generate $\rho_{AB}$ in the first place? The ultimate limitations to these two processes, called entanglement distillation and entanglement dilution, respectively, are well captured by looking at the asymptotic limit of many copies. As remarked above, this procedure is analogous to the thermodynamic limit. The resulting quantities are called the \emph{distillable entanglement} and the \emph{entanglement cost}, respectively. We already discussed their intuitive operational definitions, so we now give their mathematical forms:

\begin{align}
    E_d(\rho_{AB}) &\coloneqq \sup\left\{R>0:\, \lim_{n\to\infty} \inf_{\Lambda_n \in \sepp\left(A^nB^n\to A_0^{\ceil{Rn}}B_0^{\ceil{Rn}}\right)} \left\| \Lambda_n\left( \rho_{AB}^{\otimes n} \right) - \Phi_2^{\otimes \ceil{Rn}} \right\|_1 = 0 \right\} , \label{distillable_main} \\[.5ex]
    E_c(\rho_{AB}) &\coloneqq \inf\left\{R>0:\, \lim_{n\to\infty} \inf_{\Lambda_n \in \sepp\left(A_0^{\floor{Rn}}B_0^{\floor{Rn}}\to A^n B^n\right)} \left\| \Lambda_n\left( \Phi_2^{\otimes \floor{Rn}} \right) - \rho_{AB}^{\otimes n} \right\|_1 = 0 \right\} . \label{cost_main}
\end{align}
Here, $A^nB^n$ is the system composed by $n$ copies of $AB$, $A_0B_0$ denotes a fixed two-qubit quantum system, and $\Phi_2=\ketbra{\Phi_2}$, with $\ket{\Phi_2} = \frac{1}{\sqrt2}\left(\ket{00} + \ket{11}\right)$, is the maximally entangled state of $A_0B_0$, also called the `entanglement bit'.

One question that could be raised at this point is: is our definition of transformation rates not too restrictive? Such a reservation could be motivated by the fact that, e.g., in the resource theory of quantum thermodynamics, employing only energy-conserving unitary transformations is known to be insufficient to achieve general transformations~\cite{Brandao-thermo}; to avoid this issue, additional resources are provided in the form of ancillary systems composed of a sublinear number of qubits, allowing one to circumvent the restrictions of energy conservation without affecting the underlying physics~\cite{sparaciari_2017, faist_2019-1}. Such an approach can be adapted to more general resources~\cite{sparaciari_2020}. In our setting, however, this is already implicitly included in the definition of $E_d$ and $E_c$, since such ancillary systems can be absorbed into the asymptotic transformation rates. That is, we could have equivalently defined
\begin{align}
E_{d}(\rho_{AB}) &\coloneqq \sup\left\{R>0:\, \lim_{n\to\infty} \inf_{\Lambda_n \in \sepp\left(A^nB^nA'_nB'_n\to A_0^{\ceil{Rn}}B_0^{\ceil{Rn}}\right)} \frac12 \left\| \Lambda_n\left( \rho_{AB}^{\otimes n} \otimes \tau_n \right) - \Phi_2^{\otimes \ceil{Rn}} \right\|_1 =0 \right\} ,  \\[.5ex]
E_{c}(\rho_{AB}) &\coloneqq \inf\left\{R>0:\, \lim_{n\to\infty} \inf_{\Lambda_n \in \sepp\left(A_0^{\floor{Rn}}B_0^{\floor{Rn}}A'_nB'_n\to A^n B^n\right)} \frac12 \left\| \Lambda_n\left( \Phi_2^{\otimes \floor{Rn}} \otimes \tau_n \right) - \rho_{AB}^{\otimes n} \right\|_1 = 0\right\} ,
\end{align}
where $(\tau_n)_n$ are arbitrary (possibly entangled) systems such that $\dim \tau_n = 2^{o(n)}$. The rates are not affected by the addition of such an ancilla, since its sub-exponential size means that any contributions to the rate due to $\tau_n$ will vanish asymptotically. This is addressed in more detail in Supplementary Note~\ref{intermezzo_note}.

\subsection{The main idea: tempered negativity}

Let us commence by looking at a well-known entanglement measure called the \emph{logarithmic negativity}~\cite{negativity, plenioprl}. For a bipartite state $\rho_{AB}$, this is formally defined by
\begin{equation}
    E_N(\rho_{AB}) \coloneqq \log_2 \left\|\rho_{AB}^\Gamma\right\|_1\, .
    \label{logarithmic_negativity_main}
\end{equation}
Here, $\Gamma$ denotes the \emph{partial transpose}, i.e.\ the linear map $\Gamma: \T(\HH_A \otimes \HH_B)\to \B(\HH_A\otimes \HH_B)$, where $\B(\HH_A\otimes \HH_B)$ is the space of bounded operators on $\HH_A\otimes \HH_B$, that acts as $\Gamma(X_A\otimes Y_B) = (X_A\otimes Y_B)^\Gamma \coloneqq X_A \otimes Y_B^\intercal$, with $^\intercal$ denoting the transposition with respect to a fixed basis, and is extended by linearity and continuity to the whole $\T(\HH_A \otimes \HH_B)$~\cite{PeresPPT}. It is understood that $E_N(\rho_{AB})=\infty$ if $\rho_{AB}^\Gamma$ is not of trace class. Remarkably, the logarithmic negativity does not depend on the basis chosen for the transposition. Also, since $\sigma^\Gamma_{AB}$ is a valid state for any separable $\sigma_{AB}$~\cite{PeresPPT}, this measure vanishes on separable states, i.e.
\begin{equation}
    \text{$\sigma_{AB}$ is separable}\quad \Longrightarrow \quad \left\|\sigma_{AB}^\Gamma\right\|_1 = 1 \quad \Longrightarrow \quad E_N(\sigma_{AB}) =0\, .
    \label{logneg_vanishes_separable}
\end{equation}

Given a non-negative real-valued function on bipartite states $E$ which we think of as an `entanglement measure', when can it be used to give bounds on the operationally relevant quantities $E_d$ and $E_c$? It is often claimed that in order for this to be the case, $E$ should obey, among other things, a particular technical condition known as \emph{asymptotic continuity}. Since a precise technical definition of this term is not crucial for this discussion, it suffices to say that it amounts to a strong form of uniform continuity, in which the approximation error does not grow too large in the dimension of the underlying space. 
While asymptotic continuity is certainly a critical requirement in general~\cite{Horodecki2000, Donald2002}, it is not always indispensable~\cite{Horodecki2000, negativity, wang_2016, GrandTour, nonclassicality, regula_2020-2}. The starting point of our approach is the elementary observation that the logarithmic negativity $E_N$, for instance, is \emph{not} asymptotically continuous, yet it yields an upper bound on the distillable entanglement~\cite{negativity}. The former claim can be easily understood by casting~\eqref{logarithmic_negativity_main} into the equivalent form
\begin{equation}
    E_N(\rho_{AB}) = \log_2 \sup\left\{ \Tr X\rho:\, \left\|X^\Gamma\right\|_\infty\leq 1 \right\} ,
    \label{logarithmic_negativity_variational}
\end{equation}
where $\|Z\|_\infty \coloneqq \sup_{\ket{\psi}}\left\| Z\ket{\psi}\right\|$ is the operator norm of $Z$, and the supremum is taken over all normalised state vectors $\ket{\psi}$. Since the trace norm and the operator norm are dual to each other, the continuity of $E_N$ with respect to the trace norm is governed by the operator norm of $X$ in the optimisation~\eqref{logarithmic_negativity_variational}. However, while the operator norm of $X^\Gamma$ is at most $1$, that of $X$ can only be bounded as $\|X\|_\infty \leq d \left\|X^\Gamma\right\|_\infty \leq d$, where $d \coloneqq \min\left\{ \dim(\HH_A), \dim(\HH_B)\right\}$ is the minimum of the local dimensions. This bound is generally tight; since $d$ grows exponentially in the number of copies, it implies that $E_N$ is not asymptotically continuous.

But then why is it that $E_N$ still gives an upper bound on the distillable entanglement? A careful examination of the proof by Vidal and Werner~\cite{negativity} (see the discussion surrounding Eq.~(46) there) reveals that this is only possible because the exponentially large number $d$ actually matches the value taken by the supremum in~\eqref{logarithmic_negativity_variational} on the maximally entangled state, that is, on the target state of the distillation protocol. Let us try to adapt this capital observation to our needs. Since we want to employ a negativity-like measure to lower bound the entanglement cost instead of upper bounding the distillable entanglement, we need a substantial modification.

The above discussion inspired our main idea: let us tweak the variational program in~\eqref{logarithmic_negativity_variational} by imposing that the operator norm of $X$ be controlled by the final value of the program itself. The logic of this reasoning may seem circular at first sight, but we will see that it is not so. For two bipartite states $\rho_{AB}, \omega_{AB}$, we define the \emph{tempered negativity} by
\begin{align}
\tn(\rho | \omega) &\coloneqq \sup\left\{ \Tr X\rho:\, \left\|X^\Gamma\right\|_\infty \leq 1,\ \|X\|_\infty=\Tr X\omega \right\} , \label{tn_main} \\
\tn(\rho) &\coloneqq \tn(\rho|\rho)\, , \label{stn_main}
\end{align}
and the corresponding \emph{tempered logarithmic negativity} by
\begin{equation}
\ten(\rho) \coloneqq \log_2 \tn(\rho)\, .
\label{ten_main}
\end{equation}
This definition encapsulates the above idea of tying together the value of the function and its continuity properties, and indeed will turn out to yield the desired lower bound on the entanglement cost.
Note the critical fact that in the definition of $\tn(\rho)$ the operator norm of $X$ is given precisely by the value of $\tn(\rho)$ itself.

\subsection{Properties of the tempered negativity}

The tempered negativity $\tn(\rho|\omega)$ given by~\eqref{tn_main} can be computed as a semi-definite program for any given pair of states $\rho$ and $\omega$, which means that it can be evaluated efficiently (in time polynomial in the dimension~\cite{vandenberghe_1996}). Moreover, it obeys three fundamental properties, the proofs of which can be found in Supplementary Note~\ref{tempered}. In what follows, the states $\rho_{AB},\omega_{AB}$ are entirely arbitrary.
\begin{enumerate}[(a)]
    \item Lower bound on negativity: $\left\|\rho^\Gamma\right\|_1 \geq \tn(\rho|\omega)$, and in fact $\left\|\rho^\Gamma\right\|_1 = \sup_{\omega'} \tn(\rho|\omega')$.
    \item Super-additivity:
    \begin{equation}
    \tn(\rho^{\otimes n}) \geq \tn(\rho)^n\, ,\qquad \ten(\rho^{\otimes n}) \geq n\, \ten(\rho)\, .
    \end{equation}
    \item The `$\epsilon$-lemma':
    \begin{equation}
        \frac12 \left\|\rho - \omega\right\|_1\leq \epsilon\quad \Longrightarrow\quad \tn(\rho|\omega) \geq (1-2\epsilon)\, \tn(\omega)\, .
    \end{equation}
\end{enumerate}
The tempered negativity, just like the standard (logarithmic) negativity, is monotonic under several sets of quantum operations commonly employed in entanglement theory, such as LOCC or positive partial transpose operations~\cite{Rains2001}, but \emph{not} under non-entangling operations. Quite remarkably, it still plays a key role in our approach.

\subsection{Sketch of the proof of Theorem~\ref{main_irreversibility_thm}}

To prove Theorem~\ref{main_irreversibility_thm}, we start by establishing the general lower bound
\begin{equation}
    E_c(\rho_{AB}) \geq \ten(\rho_{AB})
    \label{tn_lower_bounds_cost}
\end{equation}
on the entanglement cost of \emph{any} state $\rho_{AB}$ under non-entangling operations. To show~\eqref{tn_lower_bounds_cost}, let $R>0$ be any number belonging to the set in the definition of $E_c$~\eqref{cost} --- in quantum information, this is known as an achievable rate for entanglement dilution. By definition, there exists a sequence of non-entangling operations $\Lambda_n\in \sepp\left(A_0^{\floor{Rn}} B_0^{\floor{Rn}}\to A^nB^n\right)$ such that $\epsilon_n\coloneqq \frac12 \left\| \Lambda_n\left(\Phi_{2^{\floor{Rn}}}\right) - \rho^{\otimes n}\right\|_1\tendsn{} 0$, where we used the notation $\Phi_d\coloneqq \frac1d \sum_{i,j=1}^d \ketbraa{ii}{jj}$ for a two-qudit maximally entangled state, and observed that $\Phi_2^{\otimes k} = \Phi_{2^k}$.

A key step in our derivation is to write $\Phi_d$ --- which is, naturally, a highly entangled state --- as the difference of two multiples of separable states. (In fact, this procedure leads to the construction of a related entanglement monotone called the standard robustness of entanglement~\cite{VidalTarrach}; we consider it in detail in the Supplementary Information.) It has long been known that this can be done by setting
\begin{equation}
    \sigma_+ \coloneqq \frac{\id + d\Phi_d}{d(d+1)}\, ,\qquad \sigma_- \coloneqq \frac{\id-\Phi_d}{d^2-1}\, ,\qquad \Phi_d = d \sigma_+ - (d-1) \sigma_-\, .
\end{equation}
where $\id$ stands for the identity on the two-qudit, $d^2$-dimensional Hilbert space. Crucially, $\sigma_\pm$ are both separable~\cite{Horodecki1999}. Applying a non-entangling operation $\Lambda$ acting on a two-qudit system yields $\Lambda(\Phi_d) = d\Lambda(\sigma_+) - (d-1)\Lambda(\sigma_-)$. Since $\Lambda(\sigma_\pm)$ are again separable, we can then employ the observation that $\left\|\sigma_{AB}^\Gamma\right\|_1=1$ for separable states (recall~\eqref{logneg_vanishes_separable})
together with the triangle inequality for the trace norm, and conclude that
\begin{equation}
    \left\| \Lambda(\Phi_d)^\Gamma \right\|_1 \leq 2d-1\, .
    \label{negativity_transformed_maxent}
\end{equation}

We are now ready to present our main argument, expressed by the chain of inequalities
\begin{align*}
    2^{\floor{Rn}+1}-1\ &\textgeq{\eqref{negativity_transformed_maxent}}\ \left\| \Lambda_n\left(\Phi_{2^{\floor{Rn}}}\right)^\Gamma \right\|_1 \\
    &\textgeq{(a)}\ \tn\left( \Lambda_n \left(\Phi_{2^{\floor{Rn}}}\right) \Big| \rho^{\otimes n} \right) \\
    &\textgeq{(c)}\ \left( 1-2\epsilon_n\right) \tn \left( \rho^{\otimes n}\right) \\
    &\textgeq{(b)}\ \left( 1-2\epsilon_n\right) \tn \left( \rho\right)^n ,
\end{align*}
derived using~\eqref{negativity_transformed_maxent} together with the above properties (a)--(c) of the tempered negativity. Evaluating the logarithm of both sides, diving by $n$, and then taking the limit $n\to\infty$ gives $R\geq \ten(\rho)$. A minimisation over the achievable rates $R>0$ then yields~\eqref{tn_lower_bounds_cost}, according to the definition of $E_c$~\eqref{cost}.

We now apply~\eqref{tn_lower_bounds_cost} to the two-qutrit state
\begin{equation}
    \omega_3 = \frac12 P_3 - \frac12 \Phi_3 = \frac16 \sum_{i,j=1}^3 \!\left(\ketbraa{ii}{ii} - \ketbraa{ii}{jj}\right) ,
\end{equation}
where $P_3\coloneqq \sum_{i=1}^3 \ketbra{ii}$. To compute its tempered logarithmic negativity, we construct an ansatz for the optimisation in the definition~\eqref{tn_main} of $N_\tau$ by setting $X_3\coloneqq 2P_3 - 3\Phi_3$. Since it is straightforward to verify that $\left\|X_3^\Gamma\right\|_\infty = 1$ and $\left\|X_3\right\|_\infty = 2 = \Tr X_3 \omega_3$, this yields
\begin{equation}
    E_c(\omega_3) \geq \ten(\omega_3) \geq 1\, .
    \label{cost_omega3}
\end{equation}
In Supplementary Note~\ref{main_results_note}, we show that the above inequalities are in fact all equalities.

It remains to upper bound the distillable entanglement of $\omega_3$. This can be done by estimating its relative entropy of entanglement~\cite{Vedral1997}, which quantifies its distance from the set of separable states as measured by the quantum relative entropy~\cite{Umegaki1962}. Simply taking the separable state $P_3/3$ as an ansatz shows that
\begin{equation}
    E_d(\omega_3)\leq \log_2 \frac32\, ,
    \label{distillable_omega3}
\end{equation}
and once again this estimate turns out to be tight. Combining~\eqref{cost_omega3} and~\eqref{distillable_omega3} demonstrates a gap between $E_d$ and $E_c$, thus proving Theorem~\ref{main_irreversibility_thm} on the irreversibility of entanglement theory under non-entangling operations. 

\subsection{Consequences and further considerations}

Our result explicitly show that there cannot exist a single quantity that governs asymptotic entanglement transformations, thus ruling out a `second law' of entanglement theory under non-entangling operations. Specifically, it is already known that, were such a quantity to exist, it would have to equal the regularised relative entropy of entanglement $E_{r,\SEP}^\infty$~\cite{Donald2002,Horodecki2002}. But then consider the fact that $E_{r,\SEP}^\infty\big(\Phi_2^{\otimes 2}\big) = 2$ while, as we show in Supplementary Note~\ref{main_results_note}, $E_{r,\SEP}^\infty\big(\omega_3^{\otimes 3}\big) = 3 \log_2 \frac32 \approx 1.75$. Thus, if the second law held, then from two copies of $\Phi_2$ one should be able to obtain three copies of $\omega_3$. But Theorem~\ref{main_irreversibility_thm} explicitly shows that only two copies of $\omega_3$ can be obtained from two copies of $\Phi_2$.

An interesting aspect of our lower bound on the entanglement cost in~\eqref{cost_omega3} is that it can therefore be strictly better than the (regularised) relative entropy bound. Previously known lower bounds on entanglement cost which can be computed in practice are actually worse than the relative entropy~\cite{Piani2009, irreversibility-PPT}, which means that our methods provide a bound that is both computable and can improve on previous approaches.

As a final remark, we note that instead of the class of non-entangling (separability-preserving) operations, we could have instead considered all positive-partial-transpose--preserving maps, which are defined as those that leave invariant the set of states whose partial transpose is positive. Within this latter approach we are able to establish an analogous irreversibility result for the theory of entanglement manipulation, recovering and strengthening the findings of Wang and Duan~\cite{irreversibility-PPT}.
Explicit details are provided in the Supplementary Information.

\subsection{Necessity of macroscopic entanglement generation}

In Theorem~\ref{main_irreversibility_thm_2}, we strengthen the result of Theorem~\ref{main_irreversibility_thm} further by considering operations that are not required to be non-entangling, but only approximately so, allowing for the possibility of microscopic fluctuations in the form of small amounts of entanglement being generated. 

As we discussed in the main text, this mirrors the approach taken by Brand\~{a}o and Plenio
~\cite{BrandaoPlenio1, BrandaoPlenio2}, where \emph{reversibility} of entanglement was claimed 
under similar constraints. The reason we call that framework into question is that the entanglement generated by the `asymptotically non-entangling maps' (ANE) employed there, despite being small when quantified by the generalised robustness, can actually be very large when gauged with another measure, such as the standard robustness or the negativity. Instead of demonstrating this with an explicit example, we prove an even stronger statement, namely, that irreversibility \emph{must} persist if the generated entanglement is required to be small with respect to these other measures. It follows logically that any claimed restoration of reversibility requires macroscopic entanglement generation in the process.

To this end, as described in the main text, we consider a sequence of operations $\Lambda_n$ which are $(M,\delta_n)$-approximately non-entangling, in the sense that
\begin{equation}
    \text{$\sigma_{AB}$ is separable}\quad \Longrightarrow\quad M\left( \Lambda_n (\sigma_{AB})\right) \leq \delta_n,
    \label{ANE_M}
\end{equation}
where $(\delta_n)_{n\in \N} \in \R_+$ is a sequence governing the restrictions on entanglement generation, and $M$ is a choice of an entanglement measure. We denote the above class of operations as $\mathrm{NE}^M_{(\delta_n)_n}$, and the associated distillable entanglement and entanglement cost as $E_{d,\mathrm{NE}^M_{(\delta_n)_n}}$ and $E_{c,\mathrm{NE}^M_{(\delta_n)_n}}$, respectively. Our irreversibility result applies to the cases when either $M(\rho)=N(\rho)\coloneqq\frac12 \left(\big\|\rho^\Gamma\big\|_1-1\right)$ is the negativity~\cite{negativity} (whose logarithmic version we already encountered in Eq.~\eqref{logarithmic_negativity_main}), or $M=R_\SEP^s$ is the standard robustness of entanglement~\cite{VidalTarrach}, defined by $R_\SEP^s(\rho) \coloneqq \inf \left\{ r\geq 0: \exists\, \text{separable state $\sigma$}: \text{$\rho+r\sigma$ separable}\right\}$. Compare this with Brand\~{a}o and Plenio's choice of the generalised robustness, given by $R_\SEP^g(\rho) \coloneqq \inf \big\{ r\geq 0: \exists\, \text{state $\sigma$}: \text{$\rho+r\sigma$ separable}\big\}$; the only difference between the latter two expressions is whether or not $\sigma$ is required to be separable.

Theorem~\ref{main_irreversibility_thm_2} then tells us that as long as the generated entanglement stays sub-exponential according to $M = N$ or $M = R^s_\SEP$, then irreversibility persists. The key step in proving this result is an approximate monotonicity of the two measures under all $(M,\delta_n)$-approximately non-entangling operations; specifically, we can show that under the application of any map $\Lambda_n$ satisfying \eqref{ANE_M}, the corresponding measure cannot increase by more than a factor $O(1+\delta_n)$. But if $\delta_n = 2^{o(n)}$, then any such 
additional term will vanish in the limit $n \to \infty$, meaning that the basic idea of our proof of Theorem~\ref{main_irreversibility_thm} can be applied almost unchanged, as the asymptotic bounds will not be affected by $(M, \delta_n)$ entanglement generation. A full discussion of the proof and the requirements on entanglement creation required to achieve reversibility can be found in Supplementary Note~\ref{generation_note}.

This contrasts with the result claimed by Brand\~{a}o and Plenio~\cite{BrandaoPlenio1, BrandaoPlenio2}: there, choosing as $M$ the generalised robustness $R^g_\SEP$ 
is conjectured~\cite{berta_gap} to yield full reversibility of the theory. 
In support of this conjecture, note that due to 
Brand\~{a}o and Plenio's result concerned with entanglement dilution --- whose proof is not affected by the aforementioned issue~\cite{berta_gap} --- the entanglement cost of an arbitrary state under $\big(R^g_\SEP,\delta_n\big)$-approximately non-entangling operations, with $\delta_n \tendsn{} 0$, coincides with its regularised relative entropy of entanglement. In the case of $\omega_3$, this equals $\log_2(3/2)$, 
which matches its distillable entanglement. Therefore, while we still lack a general proof of reversibility that holds for all states, at least $\omega_3$ \emph{is} a reversible state under Brand\~{a}o and Plenio's asymptotically non-entangling operations provided that one makes the choice $M=R_\SEP^g$. 
However, modifying this choice ever so slightly 
by picking the \emph{standard} instead of the generalised robustness shatters reversibility altogether. 
The choice of the measure in~\eqref{ANE_M} is for all intents and purposes a free parameter, and --- as we just showed --- a crucial one, on which the conclusion hinges. This ambiguity is precisely why no one framework of this type can be deemed more physical than another: 
there does not appear to be a reason to consider $M = R^g_\SEP$ a better motivated choice than $M = R^s_\SEP$. 
Due to the inability to unambiguously define a sensible notion of `small' entanglement, especially when the macroscopic limit is involved, we thus posit that the only way to enforce fully physically consistent manipulation of entanglement is to forbid any entanglement generation whatsoever, as we have done in our approach based on non-entangling operations.

\subsection{Extension to quantum communication}

The setting of quantum communication is a strictly more general framework in which the manipulated objects are quantum channels themselves. Specifically, consider the situation where the separated parties Alice and Bob are attempting to communicate through a noisy quantum channel $\Lambda: \T(\HH_A)\to \T(\HH_{B})$. To every such channel we associate its Choi--Jamiołkowski state, defined through the application of the channel $\Lambda$ to one half of a maximally entangled state: $J_\Lambda \coloneqq \left[\idc_d \otimes \Lambda\right] (\Phi_d)$, where $\idc_d$ denotes the identity channel and $d$ is the local dimension of Alice's system, assumed for now to be finite. Such a state encodes all information about a given channel~\cite{Choi,Jamiolkowski72}. The parallel with entanglement manipulation is then made clear by noticing that communicating one qubit of information is equivalent to Alice and Bob realising a noiseless qubit identity channel, $\idc_2$. But the Choi--Jamiołkowski operator $J_{\idc_2}$ is just the maximally entangled state $\Phi_2$, so the process of quantum communication can be understood as Alice and Bob trying to establish a `maximally entangled state' in the form of a noiseless communication channel. The distillable entanglement in this setting is the \textit{(two-way assisted) quantum capacity} of the channel~\cite{Bennett-error-correction}, corresponding to the rate at which maximally entangled states can be extracted by the separated parties, and therefore the rate at which quantum information can be sent through the channel with asymptotically vanishing error. In a similar way, we can 
consider the \textit{entanglement cost of the channel}~\cite{Berta2013}, that is, the rate of pure entanglement that needs to be used in order to simulate the channel $\Lambda$.

We sketch the basic idea here, as it is very similar to the approach we took for quantum states above. The complete details of the proof in the channel setting will be published elsewhere~\cite{no_second_law_channels}.

The major difference between quantum communication and the manipulation of static entanglement arises in the way that Alice and Bob can implement the processing of their channels. Having access to $n$ copies of a quantum state $\rho_{AB}$ is fully equivalent to having the tensor product $\rho_{AB}^{\otimes n}$ at one's disposal, but the situation is significantly more complex when $n$ uses of a quantum channel $\Lambda$ are available, as they can be exploited in many different ways: they can be used in parallel as $\Lambda^{\otimes n}$; or sequentially, where the output of one use of the channel can be used to influence the input to the subsequent uses; or even in more general ways which do not need to obey a fixed causal order between channel uses, and can exploit phenomena such as superposition of causal orders~\cite{Chiribella-switch,oreshkov_2012}. This motivates us, once again, to consider a general, axiomatic approach that covers all physically consistent ways to manipulate quantum channels, as long as they do not generate entanglement between Alice and Bob if it was not present in the first place. Specifically, we will consider the following. Given $n$ channels $\Lambda_1, \ldots \Lambda_n$, we define an $n$-channel quantum process to be any $n$-linear map $\Upsilon$ such that $\Upsilon(\Lambda_1, \ldots, \Lambda_n)$ is also a valid quantum channel. Now, channels $\Gamma_{A\to B}$ such that $J_\Gamma$ is separable are known as entanglement-breaking channels~\cite{HSR}. We define a non-entangling process to be one such that $\Upsilon(\Gamma_1,\ldots,\Gamma_n)$ is entanglement breaking whenever $\Gamma_1, \ldots, \Gamma_n$ are all entanglement breaking. 

The quantum capacity $Q(\Lambda)$ is then defined as the maximum rate $R$ at which non-entangling $n$-channel processes can establish the noiseless communication channel $\idc_2^{\otimes \ceil{Rn}}$ when the channel $\Lambda$ is used $n$ times. 
As in the case of quantum state manipulation, the transformation error here is only required to vanish asymptotically. 
Analogously, the (parallel) entanglement cost $E_C(\Lambda)$ is given by the rate at which noiseless identity channels $\idc_2$ are required in order to simulate parallel 
copies of the given communication channel $\Lambda$. 

The first step of the extension of our results to the channel setting is then conceptually simple: we define the tempered negativity of a channel as 
\begin{equation}\begin{aligned}
	\ten(\Lambda) \coloneqq \sup_{\rho} \ten\!\left(\left[\idc_d \otimes \Lambda\right](\rho)\right),
\end{aligned}\end{equation}
where the supremum is over all bipartite quantum states $\rho \in \T(\HH_A \otimes \HH_{A})$ on two copies of the Hilbert space of Alice's system. A careful extension of the arguments we made for states --- accounting in particular for the more complicated topological structure of quantum channels --- can be shown~\cite{no_second_law_channels} to give
\begin{equation}\begin{aligned}
	E_C(\Lambda) \geq \ten (\Lambda)
	\label{lower_bound_cost_channel}
\end{aligned}\end{equation}
for any $\Lambda: A \to B$, whether finite- or infinite-dimensional. 

For our example of an irreversible channel, we will use the qutrit-to-qutrit channel $\Omega_3$ whose Choi--Jamiołkowski state is $\omega_3$; namely,
\begin{equation}\begin{aligned}
	\Omega_3 \coloneqq \frac{3}{2} \Delta - \frac{1}{2} \idc_3
\end{aligned}\end{equation}
where $\Delta(\cdot) = \sum_{i=1}^3 \ketbra{i} \cdot \ketbra{i}$ is the completely dephasing channel. Our lower bound~\eqref{lower_bound_cost_channel} on the entanglement cost then gives $E_C(\Omega_3) \geq \ten(\Omega_3) \geq \ten(\omega_3) \geq 1$.

To upper bound the quantum capacity of $\Omega_3$, several approaches are known. If the manipulation protocols we consider were restricted to adaptive quantum circuits, we could follow 
established techniques
~\cite{Bennett-error-correction, muller_thesis,PLOB} and use the relative entropy to obtain a bound very similar to the one we employed in the state case (Eq.~\eqref{distillable_omega3}). However, to maintain full generality, we will instead employ a recent result
~\cite{regula_2020-2} which showed that an upper bound on $Q$ under the action of arbitrary non-entangling protocols --- not restricted to quantum circuits, and not required to have a definite causal order --- is given by the max-relative entropy~\cite{Datta2009} between a channel and all entanglement-breaking channels. Using the completely dephasing channel $\Delta$ as an ansatz, we then get
\begin{equation}\begin{aligned}
	Q(\Omega_3) \leq \log_2 \frac{3}{2} < 1 \leq E_C(\Omega_3),
\end{aligned}\end{equation}
establishing the irreversibility in the manipulation of quantum channels under the most general transformation protocols.\\

\subsection*{Acknowledgments}
We are grateful to Philippe Faist, Martin B.\ Plenio, Mark M.\ Wilde, and Andreas Winter 
for discussions as well as for helpful comments and suggestions on the manuscript. We also thank Seok Hyung Lie for notifying us of a typo in a preliminary version of the paper. L.L.\ was supported by the Alexander von Humboldt Foundation. B.R.\ was supported by the Japan Society for the Promotion of Science (JSPS) KAKENHI Grant No.\ 21F21015, the JSPS Postdoctoral Fellowship for Research in Japan, and the Presidential Postdoctoral Fellowship from Nanyang Technological University, Singapore.

\let\oldaddcontentsline\addcontentsline
\renewcommand{\addcontentsline}[3]{}

\vspace*{-1.5\baselineskip}
\bibliographystyle{apsrev4-1a}
 \bibliography{biblio}
 
 \let\addcontentsline\oldaddcontentsline

\clearpage

\renewcommand{\theequation}{S\arabic{equation}}
\renewcommand{\thethm}{S\arabic{thm}}
\setcounter{equation}{0}
\setcounter{thm}{0}

\begin{center}
\vspace*{\baselineskip}
{\textbf{\large No second law of entanglement manipulation after all\\[2pt] --- Supplementary Information ---}}\\[1pt] \quad \\
\end{center}

\tableofcontents


\section{General definitions}\label{general}

Throughout the Supplementary Information, we work in the full generality of Hilbert spaces which can be infinite dimensional. An operator $X:\HH\to \HH$ on a Hilbert space $X$ is said to be bounded if $\|X\|_\infty \coloneqq \sup_{\ket{\psi}\in \HH,\, \|\ket{\psi}\|\leq 1} \left\|X\ket{\psi}\right\| < \infty$. The expression on the left-hand side is called the operator norm of $X$. The Banach space of bounded operators on a Hilbert space $\HH$ will be denoted with $\B(\HH)$. Its pre-dual is the Banach space of trace class operators, denoted with $\T(\HH)$. We remind the reader that an operator $T:\HH\to \HH$ is said to be of trace class if it is bounded and moreover $\sum_{j=0}^\infty \braket{e_j|\sqrt{T^\dag T}|e_j}<\infty$ converges for some --- and hence for all --- orthonormal bases $\{ \ket{e_j} \}_{j\in \N}$. Here, $\sqrt{T^\dag T}$ is the unique positive square root of the positive semi-definite bounded operator $T^\dag T$, with $T^\dag$ being the adjoint of $T$. Positive semi-definite trace-class operators inside $\T(\HH)$ form a cone, indicated with $\Tp(\HH)$. When normalised to have trace equal to one, operators in $\Tp(\HH)$ form the set of density operators, which we will denote by $\D(\HH)$. Hereafter, the subscript $\mathrm{sa}$ (e.g.\ $\B_{\mathrm{sa}}$) indicates a restriction to self-adjoint operators.

A linear map $\Lambda:\T(\HH_A) \to \T(\HH_B)$, i.e.\ from system $A$ to system $B$, is said to be:
\begin{enumerate}[(i)]
    \item positive, if $\Lambda\left( \Tp(\HH_A) \right) \subseteq \Tp(\HH_B))$;
    \item completely positive, if $\idc_k \otimes \Lambda : \T\big(\C^k\otimes \HH_A\big) \to \T\big(\C^k\otimes \HH_B\big)$ is positive for all $k\in \N$, where $\idc_k$ denotes the identity channel on the space of $k\times k$ complex matrices;
    \item trace preserving, if $\Tr \Lambda (X) = \Tr X$ for all $X$.
\end{enumerate}
We will denote the set of positive (respectively, completely positive) trace preserving maps from $A$ to $B$ with $\ptp_{A\to B}$ (respectively, $\cptp_{A\to B}$). Given a bounded linear map $\Lambda:\T(\HH_A) \to \T(\HH_B)$,\footnote{Here, `bounded' is intended in the Banach space sense; that is, we require that $\left\|\Lambda(T)\right\|_1\leq C \|T\|_1$ for some constant $C<\infty$ and all $T\in \T(\HH_A)$.} we can consider its adjoint $\Lambda^\dag: \B(\HH_B) \to \B(\HH_A)$, defined by the identity 
\bb
\Tr \left[ X\Lambda^\dag(Y) \right] = \Tr \left[\Lambda(X) Y\right] , \qquad \forall\ X\in \T(\HH_A)\, ,\quad \forall\ Y\in \B(\HH_B)
\label{adjoint}
\ee
If $\Lambda:\T(\HH_A) \to \T(\HH_B)$ is a positive and trace preserving linear map, then it satisfies that
\bb
\left\|\Lambda(T)\right\|_1 \leq \|T\|_1\qquad \forall\ T\in \T(\HH)\, .
\label{trace_norm_contraction}
\ee
This in particular implies that $\Lambda$ is bounded in the Banach space sense. Its adjoint $\Lambda^\dag$ is positive and unital, meaning that $\Lambda^\dag(\id_B) = \id_A$, and more generally
\bb
\left\|\Lambda^\dag (X) \right\|_\infty \leq \|X\|_\infty \qquad \forall\ X\in \B(\HH)\, .
\label{operator_norm_contraction}
\ee

\subsection{Robustness measures}

Let $AB$ be a quantum system with Hilbert space $\HH_{AB}\coloneqq \HH_A\otimes \HH_B$. The set of separable states on $AB$ can be defined as the closed convex hull of all product states, in formula
\begin{equation}
    \SEP^{1}_{AB} \coloneqq \cl\left( \co\left\{ \ketbra{\psi}_A \otimes \ketbra{\phi}_B:\, \ket{\psi}_A\in \HH_A,\, \ket{\phi}_B\in \HH_B,\, \braket{\psi|\psi}=1=\braket{\phi|\phi} \right\} \right) .
\end{equation}
As mentioned in the main text, Werner, Holevo, and Shirokov~\cite{Holevo2005} have shown that a state $\sigma_{AB}$ is separable if and only if it can be expressed as 
\begin{equation}
    \sigma_{AB} = \int \ketbra{\psi}_A \otimes \ketbra{\phi}_B\, \mathrm{d}\mu(\psi,\phi)\, ,
    \label{separable}
\end{equation}
where $\mu$ is a Borel probability measure on the product of the sets of local (normalised) pure states.

The cone generated by the set of separable states will be denoted with
\begin{equation}
    \SEP_{AB} \coloneqq \cone \left( \SEP^{1}_{AB} \right) \coloneqq \left\{ \lambda \sigma_{AB}:\, \lambda\geq 0,\, \sigma_{AB}\in \SEP^{1}_{AB}\right\} .
\label{SEP}
\end{equation}
As an outer approximation to $\SEP_{AB}$, one often employs the cone of positive operators that also have a positive partial transpose (PPT). In formula, this is given by
\begin{equation}
    \PPT_{\!AB} \coloneqq \left\{ T_{AB}\in \Tp(\HH_{AB}):\ T_{AB}^\Gamma\geq 0\right\} ,
\label{PPT}
\end{equation}
and $\Gamma$ stands for the partial transpose. As recalled already in the main text, this is defined by the expression
\begin{equation}
    \Gamma(X_A\otimes Y_B) = (X_A\otimes Y_B)^\Gamma \coloneqq X_A \otimes Y_B^\intercal
    \label{PT_simple_tensors}
\end{equation}
on simple tensors, and is extended to the whole $\T(\HH_{AB})$ by linearity and continuity. Some subtleties related to the infinite-dimensional case are discussed in the Supplementary Note~\ref{infinite_dim_note}. It has been long known that~\cite{PeresPPT}
\bb
\SEP_{AB} \subseteq \PPT_{\!AB}
\ee
for all bipartite systems $AB$. Leveraging this fact, one can introduce an easily computable entanglement measure known as the \emph{\textbf{logarithmic negativity}}, given by~\cite{negativity, plenioprl}
\begin{equation}
    E_N(\rho_{AB}) \coloneqq \log_2 \left\|\rho_{AB}^\Gamma\right\|_1 = \log_2 \sup\left\{ \Tr X\rho:\, \left\|X^\Gamma\right\|_\infty\leq 1 \right\} .
    \label{logarithmic_negativity}
\end{equation}

\begin{note}
Hereafter, unless otherwise specified, $\K_{AB}$ will denote one of the two cones $\K_{AB}=\SEP_{AB}$ or $\K_{AB}=\PPT_{\!AB}$ defined by~\eqref{SEP} and~\eqref{PPT}, respectively.
\end{note}

All states from now are understood to be on a bipartite system $AB$, although we will often drop the subscripts for the sake of readability. The \emph{\textbf{(standard) $\mathbfcal{K}$-robustness}} of a state $\rho$ is defined as
\begin{align}
R_\K^s(\rho) &\coloneqq \inf \left\{ \Tr \delta:\, \delta \in \K,\, \rho+\delta\in \K \right\}. \label{std_rob}
\end{align}
Here, it is understood that the variable $\delta$ is a trace-class operator. Much of the appeal of the expression in~\eqref{std_rob} is that it is a convex optimisation program, and even a semi-definite program (SDP) for the special case of $\K=\PPT$ (however, it is an infinite-dimensional optimisation when $\rho$ acts on an infinite-dimensional space). Note that $R_\PPT^s(\rho)\leq R_\SEP^s(\rho)$ holds for all states $\rho$, as a simple inspection of~\eqref{std_rob} reveals.

\begin{note}
Our definition of robustness $R^s_\K$ follows the convention of Vidal and Tarrach~\cite{VidalTarrach}. The robustness as constructed in Ref.~\cite{taming-PRA, taming-PRL}, instead, would be expressed as $1+R_\K$ in our current notation.
\end{note}

It turns out that $1+2R_\K^s(\rho)$ is nothing but the base norm of $\rho=\rho_{AB}$ as computed in the base norm space $\left( \Bsa\left( \HH_{AB}\right),\, \K_{AB},\, \Tr_{AB}\right)$. This is proved in Ref.~\cite[Lemma~25]{taming-PRA} for the case $\K=\SEP$, and in Supplementary Note~\ref{infinite_dim_note} for the case $\K=\PPT$. Leveraging this correspondence, one can establish the dual representation~\cite[Eq.~(23) and Lemma~25]{taming-PRA}
\bb
1+2 R_\K^s(\rho) = \sup\left\{ \Tr X\rho:\ X\in \left[ -\id, \id\right]_{\K^*} \right\} \eqqcolon \|\rho\|_\K ,
\label{std_rob_dual}
\ee
where
\begin{align}
\K^* &\coloneqq \left\{ Z_{AB}\in \Bsa\left( \HH_{AB}\right):\ \Tr[Z_{AB}W_{AB}]\geq 0\quad \forall\ W_{AB}\in \K_{AB} \right\} \subset \Bsa(\HH_{AB})\, , \label{dual_cone} \\
\left[ -\id, \id\right]_{\K^*} &\coloneqq \left\{ Z_{AB}\in \Bsa\left( \HH_{AB}\right):\ \left|\Tr[Z_{AB}W_{AB}]\right|\leq \Tr W_{AB}\quad \forall\ W_{AB}\in \K_{AB} \right\} \subset \Bsa(\HH_{AB}). \label{operator_interval}
\end{align}
The notation $\left[ -\id, \id\right]_{\K^*}$ is motivated by the fact that this set can be understood as an operator interval with respect to the cone $\K^*$, in the sense that $\left[ -\id, \id\right]_{\K^*} = \left( \id-\K^*\right) \cap \left( -\id + \K^*\right)$.
Combining expressions~\eqref{std_rob} and~\eqref{std_rob_dual} allows us to compute the robustness exactly in some cases. For instance, for a pure state $\Psi_{AB} = \ketbra{\Psi}_{AB}$ with Schmidt decomposition $\ket{\Psi}_{AB} = \sum_{j=0}^\infty \sqrt{\lambda_j} \ket{e_j}_A\ket{f_j}_B$ it holds that~\cite{VidalTarrach, taming-PRA, taming-PRL}
\bb
R_\K^s(\Psi) = \left( \sumno_{j=0}^\infty \sqrt{\lambda_j} \right)^2 - 1\, .
\label{robustness_pure}
\ee
The special case of this formula where $\Psi=\Phi_2^{\otimes k}$ is made of $k$ copies of an entanglement bit $\ket{\Phi_2} = \frac{1}{\sqrt2} \left( \ket{00} + \ket{11} \right)$ is especially useful. We obtain that
\bb
R_\K^s\left(\Phi_2^{\otimes k}\right) = 2^k -1\, .
\label{robustness_k_ebits}
\ee

On a different note, by combining the expression in~\eqref{std_rob_dual} with the elementary estimate $\|\rho\|_\K \geq \left\|\rho^\Gamma\right\|_1$ one deduces the following.

\begin{lemma}
For all states $\rho$, it holds that
\bb
R_\SEP^s(\rho) \geq R_\PPT^s(\rho) \geq \frac{\left\|\rho^\Gamma\right\|_1 - 1}{2}\, .
\label{std_rob_vs_negativity}
\ee
\end{lemma}

\subsection{Distillable entanglement and entanglement cost} \label{subsec_distillable_cost}

We will also need some notation for the set of quantum channels that preserve separability and PPT-ness, or --- in short --- $\K$-ness. Remember that we denote with $(\mathrm{C})\ptp_{A\to A'}$ the set of (completely) positive and trace-preserving maps from $A$ to $A'$. Since our results actually hold for all positive transformations, even those that are not completely positive, we will drop the complete positivity assumption from now on.

Thus, let us we define \emph{\textbf{$\mathbfcal{K}$-preserving operations}} between two bipartite systems $AB$ and $A'B'$ by
\bb
\kp\left(AB \to A'B'\right) \coloneqq \left\{ \Lambda\in \ptp\left(AB\to A'B'\right):\ \Lambda(\K_{AB})\subseteq \K_{A'B'}\right\} .
\label{kp}
\ee
When $\K=\SEP$, the above identity defines the set of non-entangling (or separability-preserving) operations employed in the main text. When $\K=\PPT$, we obtain the family of PPT-preserving operations instead. Note that the name of `PPT-preserving' maps has been used in the literature to refer to several distinct concepts; we stress that here we only impose that $\Lambda(\sigma)$ is PPT whenever $\sigma$ is. 

Let us comment briefly on this choice of transformations. As discussed in the main text, the intention here is to be as general as possible --- it would be beside the point to ask ourselves whether all non-entangling transformations are physically implementable in any given theory; in general this might not always be the case. In exactly the same way, not all transformations that obey the second law of thermodynamics are physical: consider e.g.\ one that does not preserve electric charge or angular momentum.
The assumption of no entanglement generation is merely a necessary condition for a transformation to be physical --- any practical process used for entanglement manipulation should not create extra entanglement from nothing, and hence should be in $\kp(AB \to A'B')$.

We record the following elementary yet important observation.

\begin{lemma}\label{monotonicity_lemma}
For $\K=\SEP,\PPT$, the $\K$-robustnesses~\eqref{std_rob} is monotonic under $\K$-preserving operations.
\end{lemma}

\begin{proof}
The proof follows standard arguments, and indeed an even stronger variant of the monotonicity of the robustnesses (selective monotonicity) has been shown e.g.\ in Ref.~\cite{Regula2017}; we repeat the basic argument here only for the sake of convenience. 
For a bipartite state $\rho_{AB}$, an arbitrary $\epsilon>0$, and some $\Lambda\in \kp\left(AB\to A'B'\right)$, let $\delta_{AB}\in \K_{AB}$ be such that $\rho_{AB}+\delta_{AB}\in \K_{AB}$ and $\Tr\delta_{AB}\leq R_K^s(\rho_{AB})+ \epsilon$. Then $\Lambda(\delta_{AB})\in \K_{A'B'}$ and also $\Lambda(\rho_{AB}) + \Lambda(\delta_{AB}) = \Lambda(\rho_{AB}+\delta_{AB})\in \K_{A'B'}$, so that
\bbb
R_\K^s\left( \Lambda(\rho_{AB})\right) \leq \Tr \Lambda(\delta_{AB}) = \Tr \delta_{AB} \leq R_\K^s(\rho_{AB})+ \epsilon\, .
\eee
Since this holds for arbitrary $\epsilon>0$, we deduce that $R_\K^s\left( \Lambda(\rho_{AB})\right)\leq R_\K^s(\rho_{AB})$, as claimed. 
\end{proof}

We now recall the definitions of distillable entanglement and entanglement cost of a state $\rho_{AB}$ under $\K$-preserving operations. We present here a slightly more general construction than in the Methods section of the main text, namely, one which incorporates a non-zero asymptotic error. Following e.g.\ Vidal and Werner~\cite[Eq.~(43)]{negativity}, for an arbitrary $\epsilon\in [0,1)$ let us set
\begin{align}
    E_{d,\,\kp}^\epsilon(\rho_{AB}) &\coloneqq \sup\left\{R>0:\, \limsup_{n\to\infty} \inf_{\Lambda_n \in \kp\left(A^nB^n\to A_0^{\ceil{Rn}}B_0^{\ceil{Rn}}\right)} \frac12 \left\| \Lambda_n\left( \rho_{AB}^{\otimes n} \right) - \Phi_2^{\otimes \ceil{Rn}} \right\|_1 \leq \epsilon \right\} , \label{distillable} \\[.5ex]
    E_{c,\,\kp}^\epsilon(\rho_{AB}) &\coloneqq \inf\left\{R>0:\, \limsup_{n\to\infty} \inf_{\Lambda_n \in \kp\left(A_0^{\floor{Rn}}B_0^{\floor{Rn}}\to A^n B^n\right)} \frac12 \left\| \Lambda_n\left( \Phi_2^{\otimes \floor{Rn}} \right) - \rho_{AB}^{\otimes n} \right\|_1 \leq \epsilon \right\} . \label{cost}
\end{align}
For a fixed $\rho$, the function $E_{d,\,\kp}^\epsilon(\rho)$ is non-decreasing in $\epsilon$, while $E_{c,\,\kp}^\epsilon(\rho)$ is non-increasing. Also, note that
\begin{equation}
    E_{d,\,\kp}(\rho_{AB})\coloneqq E_{d,\,\kp}^0(\rho_{AB})\, ,\qquad E_{c,\,\kp}(\rho_{AB})\coloneqq E_{c,\,\kp}^0(\rho_{AB})\, , \label{distillable_cost_0}
\end{equation}
coincide with the quantities discussed in the Methods section of the main text.

A variation on the notions of distillable entanglement and entanglement cost can be obtained by looking only at \emph{exact} transformations. The corresponding modified entanglement measures read
\begin{align}
    E_{d,\,\kp}^\mathrm{exact}(\rho_{AB}) &\coloneqq \sup\left\{R>0:\, \exists\ n_0:\ \forall\ n\geq n_0\ \ \exists\ \Lambda_n\in \kp\left(A^nB^n\to A_0^{\ceil{Rn}}B_0^{\ceil{Rn}}\right):\ \Lambda_n\left( \rho_{AB}^{\otimes n} \right) = \Phi_2^{\otimes \ceil{Rn}} \right\} , \label{distillable_exact} \\[.5ex]
    E_{c,\,\kp}^\mathrm{exact}(\rho_{AB}) &\coloneqq \inf\left\{R>0:\, \exists\ n_0:\ \forall\ n\geq n_0\ \ \exists\ \Lambda_n\in \kp\left(A_0^{\floor{Rn}}B_0^{\floor{Rn}}\to A^n B^n\right):\ \Lambda_n\left( \Phi_2^{\otimes \floor{Rn}} \right) = \rho_{AB}^{\otimes n} \right\} . \label{cost_exact}
\end{align}
Although less operationally meaningful than their error-tolerant counterparts~\eqref{distillable}--\eqref{cost}, the exact distillable entanglement and the exact entanglement cost are nevertheless useful sometimes. For instance, they can come in handy in establishing bounds, thanks to the simple inequalities
\bb
E_{d,\,\kp}^{\mathrm{exact}}(\rho) \leq E_{d,\,\kp}(\rho) \leq E_{c,\,\kp}(\rho) \leq E_{c,\,\kp}^{\mathrm{exact}}(\rho)\, , \label{trivial_0}
\ee
which hold for all bipartite states $\rho$. We now introduce formally the notion of reversibility for the theory of entanglement manipulation under $\K$-preserving transformations.

\begin{Def} \label{reversibility_def}
The theory of entanglement manipulation is said to be \emph{\textbf{reversible under $\mathbfcal{K}$-preserving operations}} if $E_{d,\, \kp}(\rho_{AB}) = E_{c,\, \kp}(\rho_{AB})$ holds for all bipartite states $\rho_{AB}$. Otherwise it is said to be \emph{\textbf{irreversible under $\mathbfcal{K}$-preserving operations}}.
\end{Def}

It is not difficult to realise that neither of the two families of $\K$-preserving operations, for $\K=\SEP$ and $\K=\PPT$, is a subset of the other. Hence, one would be tempted to deduce that the distillable entanglement and the entanglement cost under non-entangling and PPT-preserving operations do not obey any general inequality. This is however not the case, and the reason must be traced back to the very high degree of symmetry exhibited by the maximally entangled state, and --- more precisely --- to the fact that for isotropic states the `PPT criterion'~\cite{PeresPPT} is necessary and sufficient for separability~\cite{Horodecki1999}. The operational relation between the two classes can be inferred from~\cite[Remark on pp.~843--844]{BrandaoPlenio2} already, and we can formalise this observation as follows.

\begin{lemma} \label{sepp_vs_pptp_lemma}
For all bipartite states $\rho$ and all $\epsilon\in [0,1)$ it holds that
\begin{align}
E_{d,\,\sepp}^\epsilon(\rho) &\geq E_{d,\, \pptp}^\epsilon(\rho)\, , \qquad E_{c,\,\sepp}^\epsilon(\rho) \geq E_{c,\, \pptp}^\epsilon(\rho)\, , \\[.5ex]
E_{d,\,\sepp}^{\mathrm{exact}}(\rho) &\geq E_{d,\,\pptp}^{\mathrm{exact}}(\rho)\, ,\qquad E_{c,\,\sepp}^{\mathrm{exact}}(\rho)\geq E_{c,\,\pptp}^{\mathrm{exact}}(\rho)\, .
\end{align}
In particular, 
\begin{align}
E_{d,\,\sepp}(\rho) &\geq E_{d,\,\pptp}(\rho)\, ,\qquad E_{c,\,\sepp}(\rho) \geq E_{c,\,\pptp}(\rho)\, .
\end{align}
\end{lemma}

\begin{proof}
We prove only the first inequality, as all the others are completely analogous. Let $R$ be an achievable rate for $E_{d,\,\pptp}^\epsilon(\rho)$ at error threshold $\epsilon\in [0,1)$, and let $\Lambda_n\in \pptp$ be PPT-preserving operations satisfying that $\limsup_{n\to\infty}\frac12 \left\| \Lambda_n\left( \rho^{\otimes n}\right) - \Phi_2^{\otimes \ceil{Rn}} \right\|_1 \leq \epsilon$. To proceed further, define the twirling operation $\pazocal{T}_m$ on a $2^{m}\times 2^{m}$ bipartite system by~\cite{Horodecki1999}
\bb
\pazocal{T}_m (X) \coloneqq &\ \int \left(U\otimes U^*\right) X \left(U\otimes U^*\right)^\dag \mathrm{d}U \\
=&\ \Tr\left[ X \Phi_2^{\otimes m}\right] \Phi_2^{\otimes m} + \Tr\left[ X \left( \id - \Phi_2^{\otimes m}\right) \right] \frac{\id-\Phi_2^{\otimes m}}{4^{m}-1} \, ,
\label{twirling}
\ee
where $\mathrm{d}U$ denotes the Haar measure over the (local) unitary group. Note that $\pazocal{T}_m\in \sepp\cap \pptp$ --- in fact, $\pazocal{T}_m$ can be physically implemented with local operations and shared randomness --- and that the output states of $\pazocal{T}_m$ are all isotropic, that is, they are linear combinations of $\Phi_2^{\otimes m}$ and the maximally mixed state. For states of this form, it is known that the PPT criterion is necessary and sufficient for separability~\cite{Horodecki1999}.

We now claim that $\pazocal{T}_{\ceil{Rn}}\circ \Lambda_n\in \sepp$. To see why this is the case, note that for all states $\sigma \in \SEP \subseteq \PPT$ it holds that $\Lambda_n(\sigma)\in \PPT$ and hence $\left(\pazocal{T}_{\ceil{Rn}}\circ \Lambda_n\right)(\sigma)\in \PPT$. However, since the latter is an isotropic state, we conclude that in fact $\left(\pazocal{T}_{\ceil{Rn}}\circ \Lambda_n\right)(\sigma)\in \SEP$, proving the claim. Observe also that
\bb
\left\| \left( \pazocal{T}_{\ceil{Rn}}\circ\Lambda_n\right)\left( \rho^{\otimes n}\right) - \Phi_2^{\otimes \ceil{Rn}} \right\|_1 = \left\| \left( \pazocal{T}_{\ceil{Rn}}\circ\Lambda_n\right)\left( \rho^{\otimes n}\right) - \pazocal{T}_{\ceil{Rn}}\left(\Phi_2^{\otimes \ceil{Rn}}\right) \right\|_1 \leq \left\| \Lambda_n\left( \rho^{\otimes n}\right) - \Phi_2^{\otimes \ceil{Rn}} \right\|_1\, ,
\label{twirling_decreases_trace_norm}
\ee
where the last inequality is a consequence of the contractivity of the trace norm under positive trace preserving maps~\cite{Ruskai1994}, or more mundanely of the triangle inequality applied to the integral representation~\eqref{twirling} of $\pazocal{T}_{\ceil{Rn}}$. The above relation implies that the distillation rate $R$ is achievable at error threshold $\epsilon$ by means of the separability-preserving operations $\pazocal{T}_{\ceil{Rn}} \circ \Lambda_n$, i.e.\ $E_{d,\,\sepp}(\rho)\geq R$. Taking the infimum in $R$ we obtain the sought inequality $E_{d,\,\sepp}(\rho) \geq E_{d,\,\pptp}(\rho)$.
\end{proof}

\section{Tempered robustness and tempered negativity}\label{tempered}

The main idea is to introduce a modified version of the standard $\K$-robustness in~\eqref{std_rob} by modifying the dual program in~\eqref{std_rob_dual}. Namely, for a pair of states $\rho,\omega$, let us define the \emph{\textbf{$\boldsymbol{\omega}$-tempered $\mathbfcal{K}$-robustness}} by
\begin{align}
1+2\tsr_{\!\K}(\rho | \omega) &\coloneqq \sup\left\{ \Tr X\rho:\ X\in \left[ -\id, \id\right]_{\K^*},\ \|X\|_\infty=\Tr X\omega \right\} , \label{tsr} \\
\tsr_{\!\K} (\rho ) &\coloneqq \tsr_{\!\K}(\rho|\rho)\, . \label{stsr}
\end{align}
Here, the operator interval $\left[ -\id, \id\right]_{\K^*}$ is defined by~\eqref{operator_interval}. Note that the constraint $\|X\|_\infty=\Tr X\omega$ can be rewritten as $-\left( \Tr X\omega\right)\id \leq X\leq \left( \Tr X\omega\right)\id$. Therefore, the expression in~\eqref{tsr} is a convex program, and even an SDP for the special case $\K=\PPT$. What this additional constraint is trying to tell us is that the support $\supp\omega$ of $\omega$ lies entirely within the eigenspace of $X$ corresponding to the eigenvalue with the largest modulus. At this point, a little thought shows that $\tsr_{\!\K}(\rho | \omega)$ depends in fact only on $\rho$ and $\supp \omega$. Analogously, $\tsr_{\!\K}(\rho)$ depends only on $\supp \rho$.

We also introduce a further quantity, the \emph{\textbf{$\boldsymbol{\omega}$-tempered negativity}}, defined by
\begin{align}
\tn(\rho | \omega) &\coloneqq \sup\left\{ \Tr X\rho:\, \left\|X^\Gamma\right\|_\infty \leq 1,\ \|X\|_\infty=\Tr X\omega \right\} , \label{tn} \\
\tn(\rho) &\coloneqq \tn(\rho|\rho)\, . \label{stn}
\end{align}
Exactly as above, it does not take long to realise that the expression in~\eqref{tn} is in fact an SDP, and that $\tn(\rho|\omega)$ depends only on $\rho$ and $\supp\omega$, while $\tn(\rho)$ depends only on $\supp\rho$. The corresponding \emph{\textbf{tempered logarithmic negativity}} is
\bb
\ten(\rho) \coloneqq \log_2 \tn(\rho)\, .
\label{ten}
\ee

The main elementary properties of the tempered robustness and negativity --- related to their monotonicity, multiplicativity, and various bounds between the quantities --- are gathered in the following proposition.

\begin{prop} \label{elementary_prop}
For $\K=\SEP, \PPT$ and for all pairs of states $\rho,\omega$ on a bipartite system $AB$, it holds that:
\begin{enumerate}[(a)]
\item $0\leq \tsr_{\!\K}(\rho|\omega)\leq R_\K^s(\rho)$, and $R_\K^s(\rho) = \sup_{\omega'} \tsr_\K(\rho|\omega')$;
\item $1\leq \tn(\rho|\omega)\leq \left\|\rho^\Gamma\right\|_1$, and $\left\|\rho^\Gamma\right\|_1 = \sup_{\omega'} \tn(\rho|\omega')$;
\item $\tsr_{\!\K}$ is monotonic under the simultaneous action of any $\K$-preserving map $\Lambda\in \kp\left(AB\to A'B'\right)$, in formula
\bb
\tsr_{\!\K} \left( \Lambda(\rho)\, \big|\, \Lambda(\omega) \right) \leq \tsr_{\!\K}\left(\rho|\omega\right) .
\ee
\item the inequalities
\bb
\tsr_{\! \SEP}(\rho|\omega) \geq \tsr_{\! \PPT}(\rho|\omega) \geq \frac{\tn(\rho|\omega)-1}{2}
\label{inequality_tsr_tn}
\ee
are satisfied.
\item $\tn$ is super-multiplicative and hence $\ten$ is super-additive, in formula
\begin{align}
\tn\left( \rho_1\otimes \rho_2 \big| \omega_1\otimes \omega_2\right) &\geq \tn(\rho_1|\omega_1)\, \tn(\rho_2|\omega_2) \, , \label{supermultiplicativity_tn}\\
\ten\!\left( \rho_1\otimes \rho_2 \big| \omega_1\otimes \omega_2\right) &\geq \ten(\rho_1|\omega_1) + \ten(\rho_2|\omega_2)\, , \label{superadditivity_ten}
\end{align}
for all states $\rho_1,\rho_2,\omega_1,\omega_2$.
\end{enumerate}
\end{prop}

\begin{proof}
We proceeed one claim at a time.
\begin{enumerate}[(a)]

\item Taking $X=\id$ in the definition of $\tsr_{\!\K}$~\eqref{tsr} yields immediately that $\tsr_{\!\K}(\omega|\tau)\geq 0$. Also, since we obtained~\eqref{tsr} by adding one more constraint to the dual program~\eqref{std_rob_dual} for the standard robustness, it is clear that the value of the the supremum can only decrease, implying that $\tsr_{\!\K}(\rho|\omega)\leq R_\K^s(\rho)$.

On the other hand, it is not difficult to verify that the operators $X$ in the dual formulation of $R_\K^s$~\eqref{std_rob_dual} can always be assumed to be compact and in fact even of finite rank. Indeed, thanks to the fact that $\HH_A$ and $\HH_B$ are separable Hilbert spaces, we can pick sequences of finite-dimensional projectors $(P_A^N)_{N\in \N}$ and $(P_B^N)_{N\in \N}$ such that $(P_A^N\otimes P_B^N) \rho_{AB} (P_A^N\otimes P_B^N)\tends{}{N\to\infty} \rho_{AB}$ in trace norm. For any given $X_{AB} \in [-\id,\id]_{\K^*}$, the finite-rank operators $X_N\coloneqq (P_A^N\otimes P_B^N) X_{AB} (P_A^N\otimes P_B^N)$ satisfy that $X_N\in [-\id,\id]_{\K^*}$, simply because $\sigma_{AB} \mapsto (P_A^N\otimes P_B^N) \sigma_{AB} (P_A^N\otimes P_B^N)$ sends $\K$ into itself and is trace non-increasing. Moreover, $\Tr[\rho X_N]\tends{}{N\to\infty} \Tr \rho X$.

This shows that $X$ in~\eqref{std_rob_dual} can be taken to be of finite rank. For any such $X$, there will exist a state $\omega_X$ such that $\|X\|_\infty = \Tr X\omega_X$; in fact, it suffices to have the support of $\omega_X$ span the eigenspace of $X$ corresponding to the eigenvalue with maximum modulus. Therefore,
\bbb
\sup_{\omega'} \tsr_{\!\K}(\rho|\omega') = \sup_{\substack{X\in [-\id, \id]_{\K^*},\, \omega': \\ \|X\|_\infty = \Tr X\omega'}} \Tr X\rho \geq \sup_{\substack{X\in [-\id, \id]_{\K^*},\\ \rk X<\infty}} \Tr X\rho = R_\K^s(\rho)\, .
\eee

\item The lower bound $\tn(\rho|\omega)\geq 1$ can be retrieved by setting $X=\id$ in~\eqref{tn}. The fact that $\tn(\rho|\omega)\leq \left\|\rho^\Gamma\right\|_1$ follows by comparing~\eqref{tn} with the dual form of the negativity on the rightmost side of~\eqref{logarithmic_negativity}. The equality $\left\|\rho^\Gamma\right\|_1 = \sup_{\omega'} \tn(\rho|\omega')$ is proved as for claim~(a). One starts by showing that the operator $X$ in the rightmost side of~\eqref{logarithmic_negativity} can be assumed to have finite rank. To see this, it suffices to observe that $X_N= (P_A^N\otimes P_B^N) X_{AB} (P_A^N\otimes P_B^N)$ defined as before satisfies that
\bbb
\left\|X_N^\Gamma\right\|_\infty = \left\| P_A^N \otimes \left( P_B^N\right)^\intercal X_{AB}^\Gamma\, P_A^N \otimes \left( P_B^N\right)^\intercal\right\|_1\leq \left\|X^\Gamma\right\|_1\, .
\eee
Since $\Tr[\rho X_N]\tends{}{N\to\infty} \Tr \rho X$, considering the sequence of finite-rank operators $(X_N)_N$ instead of $X$ in~\eqref{logarithmic_negativity} leads to the same value of the optimisation.

\item It suffices to write that
\bb
    \tsr_{\!\K} \left( \Lambda(\rho)\, \big|\, \Lambda(\omega) \right) &= \sup\left\{ \Tr X\Lambda(\rho):\ X\in \left[ -\id, \id\right]_{\K^*},\ \|X\|_\infty=\Tr X\Lambda(\omega) \right\} \\
    &\texteq{(i)} \sup\left\{ \Tr \Lambda^\dag(X)\rho:\ X\in \left[ -\id, \id\right]_{\K^*},\ \|X\|_\infty=\Tr \Lambda^\dag(X)\omega \right\} \\
    &\textleq{(ii)} \sup\left\{ \Tr Y\rho:\ Y\in \left[ -\id, \id\right]_{\K^*},\ \|Y\|_\infty=\Tr Y\omega \right\} \\
    &= \tsr_{\!\K}(\rho|\omega)\, .
\ee
Note that in~(i) we just used the definition of adjoint map. Justifying~(ii) requires a bit more care. We start by observing that if $X\in \left[ -\id, \id\right]_{\K^*}$ then $Y= \Lambda^\dag(X)$ satisfies
\bbb
\sup_{\sigma\in \K,\, \Tr \sigma=1} \left| \Tr Y\sigma \right| = \sup_{\sigma\in \K,\, \Tr \sigma=1} \left| \Tr X\Lambda(\sigma) \right| \leq \sup_{\sigma'\in \K,\, \Tr \sigma'=1} \left| \Tr X \sigma' \right| \leq 1\, ,
\eee
where the inequality holds because $\Lambda(\sigma)$ is a normalised quantum state and belongs to $\K$. (We remind the reader that our definition of $\K$-preserving maps imposes that any such map is also positive and trace preserving.) Moreover, $\left\|Y\right\|_\infty \leq \|X\|_\infty = \Tr X \Lambda(\omega) = \Tr Y \omega$ thanks to the positivity and unitality of $\Lambda^\dagger$ (see~\eqref{operator_norm_contraction}); this is in fact an equality, because on the other hand $\Tr Y\omega \leq \|Y\|_\infty \|\omega\|_1 = \|Y\|_\infty$. Since $Y=\Lambda^\dag(X)$ satisfies that $Y\in \left[ -\id, \id\right]_{\K^*}$ and moreover $\|Y\|_\infty=\Tr Y\omega$, we deduce the inequality in~(ii).

\item It is easy to see from~\eqref{tsr} that $\tsr_{\!\K}(\rho|\omega)$ is monotonically decreasing with respect to the inclusion ordering on the cone $\K$ for all fixed $\rho$ and $\omega$, meaning that $\K_1\subseteq \K_2$ implies that $\tsr_{\!\K_1}(\rho|\omega)\geq \tsr_{\!\K_2}(\rho|\omega)$. Since $\SEP\subseteq \PPT$, the first inequality in~\eqref{inequality_tsr_tn} follows.

We now move on to the second. Note that $\left\|X^\Gamma\right\|_\infty\leq 1$ 
entails that $X\in \left[-\id, \id\right]_{\SEP^*}$, simply because for all $\sigma\in \SEP$ with $\Tr\sigma=1$ one has that
\bbb
\left|\Tr [X\sigma]\right| = \left| \Tr\left[X^\Gamma \sigma^\Gamma\right] \right| \leq \left\|X^\Gamma\right\|_\infty \left\|\sigma^\Gamma\right\|_1\leq 1\, ,
\eee
where we remembered that $\sigma^\Gamma\geq 0$ and hence $\left\|\sigma^\Gamma \right\|_1=\Tr \sigma^\Gamma = \Tr \sigma =1$. Hence, the set on the right-hand side of~\eqref{tn} is contained in that on the right-hand side of~\eqref{tsr}, which shows that $\tn(\rho|\omega)\leq 1+2\tsr_{\!\SEP}(\rho|\omega)$.

\item To show the super-multiplicativity of $\tn$, it suffices to make a tensor product ansatz inside~\eqref{tn}, obtaining that
\bb
\tn\left( \rho_1\otimes \rho_2 \big| \omega_1\otimes \omega_2\right) &= \sup\left\{ \Tr X_{12} (\rho_1\otimes \rho_2):\ \left\| X_{12}^\Gamma \right\|_\infty \leq 1,\ \|X_{12}\|_\infty=\Tr X_{12}(\omega_1\otimes \omega_2) \right\} \\
&\textgeq{(iii)} \sup\left\{ \Tr (X_1\!\otimes\! X_2) (\rho_1\!\otimes \!\rho_2):\ \left\|X_i^\Gamma \right\|_\infty\leq 1,\ \left\|X_i\right\|_\infty = \Tr X_i \omega_i \right\} \\
&= \tn(\rho_1|\omega_1)\, \tn(\rho_2|\omega_2) \, .
\ee
Here, the inequality in~(iii) can be proved by noting that $\left\|X_i^\Gamma \right\|_\infty\leq 1$ entails that $\left\|(X_1\otimes X_2)^\Gamma \right\|_\infty = \left\|X_1^\Gamma \right\|_\infty \left\| X_2^\Gamma \right\|_\infty \leq 1$. The super-additivity of $\ten$ in~\eqref{superadditivity_ten} follows immediately. \qedhere
\end{enumerate}
\end{proof}

In addition to the basic properties established above, our main results will rely on one more technical property of the tempered quantities. This is a perturbative version of Proposition~\ref{elementary_prop}(a), allowing us to relate the robustness $R^s_\K(\rho')$ of a given state with the tempered robustness $\tsr_{\!\K}(\rho)$ of another state which is sufficiently close to it. The following lemma can quite rightly be regarded as lying at the heart of our method.

\begin{lemma}[(The $\epsilon$-lemma)] \label{crucial_lemma}
For all states $\rho,\rho'$ such that
\bb
\epsilon\coloneqq \frac12 \left\|\rho - \rho'\right\|_1 \leq \frac12\, ,
\ee
it holds that
\bb
R^s_\K (\rho') \geq \tsr_{\!\K}(\rho'|\rho) \geq (1-2\epsilon) \tsr_{\!\K}(\rho) - \epsilon
\label{crucial}
\ee
and also
\bb
\tn(\rho'|\rho) \geq (1-2\epsilon)\, \tn(\rho)\, .
\label{crucial_tn}
\ee
\end{lemma}

\begin{proof}
The first inequality in~\eqref{crucial} is just an application of Proposition~\ref{elementary_prop}(a). As for the second, using the definition~\eqref{tsr} of $\tsr_{\!\K}$ as well as H\"older's inequality we see that
\bb
1+2\tsr_{\!\K}(\rho' | \rho) &= \sup\left\{ \Tr X\rho':\ X\in \left[ -\id, \id\right]_{\K^*},\ \|X\|_\infty=\Tr X\rho \right\} \\
&\geq \sup\left\{ \Tr X\rho - \left\|X\right\|_\infty \left\|\rho - \rho'\right\|_1 :\ X\in \left[ -\id, \id\right]_{\K^*},\ \|X\|_\infty=\Tr X\rho \right\} \\
&= \sup\left\{ (1-2\epsilon)\Tr X\rho :\ X\in \left[ -\id, \id\right]_{\K^*},\ \|X\|_\infty=\Tr X\rho \right\} \\
&= (1-2\epsilon) \left( 1+ 2\tsr_{\!\K}(\rho)\right) ,
\ee
which becomes~\eqref{crucial} upon elementary algebraic manipulations. The proof of~\eqref{crucial_tn} is entirely analogous:
\bb
\tn(\rho'|\rho) &= \sup\left\{ \Tr X\rho':\ \left\|X^\Gamma\right\|_\infty \leq 1,\ \|X\|_\infty=\Tr X\rho \right\} \\
&\geq \sup\left\{ \Tr X\rho - \left\|X\right\|_\infty \left\|\rho - \rho'\right\|_1:\ \left\|X^\Gamma\right\|_\infty \leq 1,\ \|X\|_\infty=\Tr X\rho \right\} \\
&= \sup\left\{ (1-2\epsilon)\Tr X\rho :\ \left\|X^\Gamma\right\|_\infty \leq 1,\ \|X\|_\infty=\Tr X\rho \right\} \\
&= (1-2\epsilon)\, \tn(\rho)\, .
\ee
This concludes the proof.
\end{proof}

\section{Main results: irreversibility of entanglement manipulation} \label{main_results_note}

Here we state our main results concerning the theory of entanglement manipulation for quantum states. The extension of the argument to quantum channels will be tackled in full detail separately~\cite{no_second_law_channels}.

\begin{thm} \label{main_tool_thm}
For $\K=\SEP$ or $\K=\PPT$, the entanglement cost under $\K$-preserving operations satisfies that
\bb
\inf_{\epsilon\, \in\, [0,\,1/2)} E_{c,\,\kp}^\epsilon (\rho) \geq \tl_{\!\K}(\rho)\, ,
\label{main_tool_formal}
\ee
where
\bb
\tl_{\!\K}(\rho) \coloneqq \limsup_{n\to\infty} \frac1n \log_2 \left( 1+\tsr_{\!\K}\left(\rho^{\otimes n}\right) \right) \geq \ten(\rho)\, ,
\label{main_tool}
\ee
and the tempered logarithmic negativity $\ten$ is defined by~\eqref{ten}.
\end{thm}

\begin{rem}
An interesting consequence of the above result is that the tempered logarithmic negativity is a lower bound on the standard entanglement cost under local operations and classical communication (LOCC), denoted $E_{c,\, \locc}$, in formula
\bb
E_{c,\, \locc}(\rho_{AB}) \geq \ten(\rho_{AB}) \qquad \forall\ \rho_{AB}\, .
\ee
The entanglement cost under LOCC is a notoriously hard quantity to compute; it is given by the regularised entanglement of formation~\cite{Hayden-EC}, and the regularisation is known to be necessary due to Hastings's counterexample to the additivity conjectures~\cite{Shor2004, Hastings2008}. Previously known lower bounds include the regularised relative entropy of entanglement~\cite{Vedral1997} and the squashed entanglement~\cite{Tucci1999, squashed, faithful}, both of which are extremely hard to evaluate in general (albeit for different reasons). 
The former can be in turn lower bounded by either Piani's measured relative entropy of entanglement~\cite{Piani2009}, which has the advantage of doing away with regularisations, or by the $E_\eta$ measure recently proposed by Wang and Duan~\cite{irreversibility-PPT}, which is particularly convenient computationally because it is given by a semi-definite program (SDP). Both of these lower bounds on the LOCC entanglement cost, that inferred by Piani's results and that relying on $E_\eta$, are quite useful, but are known to be weaker than that given by the regularised relative entropy of entanglement.

The tempered negativity provides us with an independent lower bound on the LOCC entanglement cost that can strictly improve on the regularised relative entropy one. This latter fact will be apparent from the proof of Theorem~\ref{irreversibility_thm}. Also, since it is also given by an SDP, our bound is still computationally friendly. We are aware of no other quantity possessing these properties.
\end{rem}

\begin{proof}[Proof of Theorem~\ref{main_tool_thm}]
The following argument could be marginally simplified at the level of notation by resorting to the results of Brand\~{a}o and Plenio~\cite{BrandaoPlenio2}. However, for the sake of readability we prefer to give a more direct and self-contained proof.

Call $AB$ the bipartite system where $\rho$ lives. Let $R$ be an achievable rate for the entanglement cost $E_{c,\, \kp}^\epsilon (\rho)$ at some error threshold $\epsilon\in [0,1/2)$. Consider a sequence of operations $\Lambda_n\in \kp_{A_0^{\floor{Rn}} B_0^{\floor{Rn}}\to A^n B^n}$, with $A_0,B_0$ being single-qubit systems, such that
\bb
\epsilon_n\coloneqq \frac12 \left\| \Lambda_n\left( \Phi_2^{\otimes \floor{Rn}}\right) - \rho^{\otimes n} \right\|_1
\label{main_tool_proof_eq1}
\ee
with
\bb
\limsup_{n\to\infty} \epsilon_n\leq \epsilon < \frac12\, .
\label{main_tool_proof_eq2}
\ee
For all sufficiently large $n$, we then write
\bb
2^{\floor{Rn}} &\texteq{(i)} 1 + R_\K^s\left( \Phi_2^{\otimes \floor{Rn}}\right) \\
&\textgeq{(ii)} 1 + R_\K^s\left( \Lambda_n\left(\Phi_2^{\otimes \floor{Rn}}\right)\right) \\
&\textgeq{(iii)} (1-2\epsilon_n) \left( 1+\tsr_{\!\K}\left(\rho^{\otimes n}\right) \right) + \epsilon_n \\
&\geq (1-2\epsilon_n) \left( 1+\tsr_{\!\K}\left(\rho^{\otimes n}\right) \right) \\
\label{epsilon_lemma_at_work}
\ee
Here, in~(i) we just recalled the value of the standard robustness of maximally entangled states~\eqref{robustness_k_ebits}, (ii)~comes from the monotonicity of $R_\K^s$ under $\K$-preserving operations (Lemma~\ref{monotonicity_lemma}), and (iii)~is an application of the $\epsilon$-lemma (Lemma~\ref{crucial_lemma}). Taking the logarithm, dividing by $n$, and computing the limit for $n\to\infty$ yields
\bb
R &= \lim_{n\to\infty} \frac{\floor{Rn}}{n} \\
&\geq \limsup_{n\to\infty} \frac1n \log_2 \left( (1-2\epsilon_n) \left( 1+\tsr_{\!\K}\left(\rho^{\otimes n}\right) \right) \right) \\
&\geq \limsup_{n\to\infty} \frac1n \log_2 \left( 1+\tsr_{\!\K}\left(\rho^{\otimes n}\right) \right) + \liminf_{n\to\infty} \frac1n \log_2 (1-2\epsilon_n) \\
&\texteq{(iv)} \limsup_{n\to\infty} \frac1n \log_2 \left( 1+\tsr_{\!\K}\left(\rho^{\otimes n}\right) \right) \\
&= \tl_{\!\K} (\rho) \, .
\ee
where (iv)~is a consequence of the fact that $\epsilon_n$ is bounded away from $1/2$, as per~\eqref{main_tool_proof_eq2}. This completes the proof of the first inequality~\eqref{main_tool_formal}.

As for the second inequality~\eqref{main_tool}, we observe that
\bb
\tl_{\!\K} (\rho) &= \limsup_{n\to\infty} \frac1n \log_2 \left( 1+\tsr_{\!\K}\left(\rho^{\otimes n}\right) \right) \\
&\textgeq{(v)} \limsup_{n\to\infty} \frac1n \log_2\frac{\tn\left(\rho^{\otimes n}\right) + 1}{2} \\
&\geq \limsup_{n\to\infty} \left( \frac1n \log_2 \tn\left(\rho^{\otimes n}\right) - \frac1n \right) \\
&\textgeq{(vi)} \limsup_{n\to\infty} \left( \frac1n \log_2 \left( \tn\left(\rho\right)^n \right) - \frac1n \right) \\
&\texteq{(vii)} \ten(\rho)\, .
\ee
Here, (v)~is an application of the lower bound in Proposition~\ref{elementary_prop}(d), in~(vi) we leveraged the super-multiplicativity of the tempered negativity (Proposition~\ref{elementary_prop}(e)), and finally~(vii) is just the definition~\eqref{ten} of tempered logarithmic negativity.
\end{proof}

Before we state and prove our result on the irreversibility of entanglement, we need to recall and discuss two well-known bounds on the distillable entanglement. Lower bounds on $E_{d,\,\kp}(\rho)$ can be obtained by looking at smaller classes of operations included in the set of all $\K$-preserving ones. A typical choice is the set of local operations assisted by one-way classical communication, say from Alice to Bob, denoted with $\locc_\to$. In this setting, Devetak and Winter's hashing inequality~\cite{devetak2005} states that
\bb
E_{d,\,\locc_\to}(\rho_{AB}) \geq I_{\mathrm{coh}}(A\rangle B)_\rho \coloneqq S(\rho_B) - S(\rho_{AB})\, ,
\label{hashing}
\ee
where $S(\omega) \coloneqq -\Tr \omega \log_2\omega$ is the von Neumann entropy, and $\rho_B\coloneqq \Tr_A \rho_{AB}$ is the reduced state of $\rho_{AB}$ on Bob's side. Since local operations assisted by one-way classical communication are both non-entangling and PPT-preserving, in formula $\locc_\to\subseteq \kp$, we see that $E_{d,\,\locc_\to}(\rho_{AB})\leq E_{d,\, \kp}(\rho_{AB})$. In particular, the rightmost side of the hashing inequality~\eqref{hashing} lower bounds the distillable entanglement under $\K$-preserving operations, i.e.
\bb
E_{d,\, \kp}(\rho_{AB}) \geq I_{\mathrm{coh}}(A\rangle B)_\rho\, .
\label{hashing_kp}
\ee
To establish an upper bound on $E_{d,\, \kp}(\rho_{AB})$, instead, we can introduce a relative entropy measure defined by~\cite{Vedral1997}
\bb
E_{r,\,\K}(\rho_{AB}) \coloneqq \inf_{\sigma_{AB}\in \K_{AB} \cap \D(\HH)} D(\rho_{AB}\|\sigma_{AB})\, ,
\label{relent_kness}
\ee
Since $\K_{AB}\otimes \K_{A'B'} \subseteq \K_{AA'BB'}$, the function $E_{r,\, \K}$ is sub-additive, and then Fekete's lemma~\cite{Fekete1923} implies that its regularisation
\bb
E_{r,\,\K}^\infty(\rho_{AB}) \coloneqq \lim_{n\to\infty} \frac1n E_{r,\,\K}\left(\rho_{AB}^{\otimes n}\right) ,
\label{Er_regularised}
\ee
is well defined and satisfies that $E_{r,\,\K}^\infty(\rho_{AB})\leq E_{r,\,\K}(\rho_{AB})$. It does not take long to realise that $E_{r,\,\K}$ is monotonic under $\K$-preserving operations. This amounts to an elementary observation once one remembers that the relative entropy is non-increasing under the simultaneous application of any positive trace preserving map~\cite{Alex2017}. Since $E_{r,\, \K}$ is also asymptotically continuous~\cite{Donald1999} (see also~\cite[Lemma~7]{tightuniform}), its regularisation can be shown to be an upper bound on the distillable entanglement under $\K$-preserving operations~\cite{Vedral1998,hayashi_book}:
\begin{equation}\label{eq:distillable_upper_Er}
E_{d,\,\kp}(\rho_{AB}) \leq E_{r,\,\K}^\infty(\rho_{AB}).
\end{equation}

We are now ready to make use of the above Theorem~\ref{main_tool_thm} to prove irreversibility of entanglement manipulation under both non-entangling and PPT-preserving operations. To this end, according to Definition~\ref{reversibility_def} (cf.~\eqref{trivial_0}) it suffices to exhibit an example of a bipartite state $\rho_{AB}$ for which $E_{d,\, \kp}(\rho_{AB}) < E_{c,\, \kp}(\rho_{AB})$. Our candidate is a two-qutrit state, with Hilbert space $\HH_A\otimes \HH_B = \C^3\otimes \C^3$. Denote the local computational basis of the two qutrits $A,B$ with $\{\ket{j}\}_{j=1,2,3}$. Define the projector onto the maximally correlated subspace and the maximally entangled state by
\bb
P_3\coloneqq \sum_{j=1}^3 \ketbra{jj}\, ,\qquad \ket{\Phi_3} = \frac{1}{\sqrt3} \sum_{j=1}^3 \ket{jj}\, ,\qquad \Phi_3\coloneqq \ketbra{\Phi_3}\, ,
\label{P_3_and_Phi_3}
\ee
respectively. Then, construct the state
\bb
\omega_3 = \omega_3^{AB} \coloneqq \frac12 \left( P_3 - \Phi_3 \right) .
\label{omega_3}
\ee
We now show the following, proving and extending Theorem 1 from the main text of the paper.

\begin{thm} \label{irreversibility_thm}
The two-qutrit state $\omega_3$ defined by~\eqref{omega_3} satisfies that
\bb
E_{d,\, \sepp}(\omega_3) = E_{d,\, \pptp}(\omega_3) = \log_2 \frac32 \approx 0.585
\label{omega_3_distillable}
\ee
but
\bb
E_{c,\, \sepp}^\epsilon (\omega_3) = E_{c,\, \pptp}^\epsilon (\omega_3) = 1
\label{omega_3_cost}
\ee
for all $\epsilon\in [0,\, 1/2)$. In particular, the resource theory of entanglement is irreversible under either non-entangling or PPT-preserving operations. 
\end{thm}

\begin{rem}
The above result not only guarantees that the entanglement cost of the state $\omega_3$ under non-entangling operations is $1$. It also establishes a `pretty strong' converse~\cite{pretty-strong} for this value of the rate. Namely, every protocol that attempts to prepare $\omega_3$ from entanglement bits at a rate smaller than $1$ must incur an asymptotic error that is not only non-vanishing, but actually larger than a constant. This constant is $1/2$ in the current formulation of Theorem~\ref{irreversibility_thm}. However, we will see in Lemma~\ref{error_rate_tradeoff_lemma} that a careful analysis actually yields a slightly larger value of $2/3$. An even stronger statement (strong converse) can be shown for distillable entanglement, where no error smaller than $1$ can improve the transformation rates whatsoever. For simplicity, we have omitted these extensions from the statement of Theorem~\ref{irreversibility_thm}, and we instead refer the interested reader to Supplementary Note~\ref{strong_converses_note} for a more in-depth discussion of (pretty) strong converses and error-rate trade-offs.
\end{rem}

\begin{proof}[Proof of Theorem~\ref{irreversibility_thm}]
We have that
\begin{equation}
\log_2 \frac32 \texteq{(i)} I_{\mathrm{coh}}\left( A\rangle B\right)_{\omega_3} \,\textleq{(ii)}\, E_{d,\, \kp}(\omega_3) \,\textleq{(iii)}\, E_{r,\, \K}^\infty (\omega_3) \,\textleq{(iv)}\, E_{r,\, \K} (\omega_3) \,\textleq{(v)}\, D\left( \omega_3 \| P_3/3 \right) \,\texteq{(vi)}\, \log_2 \frac32\, .
\end{equation}
Here, (i)~is an elementary computation, (ii)~follows from the hashing inequality~\eqref{hashing_kp},
(iii)~is a consequence of the upper bound on distillable entanglement in~\eqref{eq:distillable_upper_Er}, (iv)~descends from the aforementioned sub-additivity of the relative entropy of $\K$-ness, (v)~is deduced by taking as ansatz in~\eqref{relent_kness} the state $\sigma_{AB} = P_3/3\in \SEP_{AB}\subseteq \PPT_{AB}$, and finally (vi)~comes again from a direct calculation. This proves~\eqref{omega_3_distillable}.

As for the entanglement cost, irreversibility of entanglement under $\K$-preserving operations hinges on the crucial inequality $E_{c,\, \kp}^\epsilon(\omega_3)\geq 1$. Hereafter, $\epsilon\in [0,\,1/2)$ is a fixed constant. Thanks to Theorem~\ref{main_tool_thm}, it suffices to show that $\tn(\omega_3)\geq 2$. To this end, using the notation defined in~\eqref{P_3_and_Phi_3}, let us consider the operator
\bb
X_3 \coloneqq 2P_3 - 3 \Phi_3\, ,
\label{X_3}
\ee
Its eigenvalues are $2$ (with multiplicity $2$), $0$ (with multiplicity $6$) and $-1$ (with multiplicity $1$). Since $X_3$ is normal (i.e.\ it commutes with its adjoint --- in fact, $X_3$ is Hermitian), its operator norm coincides with the maximum modulus of an eigenvalue. Therefore,
\bb
\|X_3\|_\infty = 2\, .
\label{irreversibility_proof_eq3}
\ee
Calling $F_3\coloneqq \sum_{i,j=1}^3 \ketbraa{ij}{ji}$ the swap operator, it does not take long to verify that the partial transpose of $X_3$ evaluates to
\bb
X_3^\Gamma = 2P_3 - F_3\, .
\label{irreversibility_proof_eq4}
\ee
Since $X_3^\Gamma$ has eigenvalues $+1$ (with multiplicity $6$) and $-1$ (with multiplicity $3$),
\bb
\left\|X_3^\Gamma\right\|_\infty = 1\, .
\label{irreversibility_proof_eq5}
\ee
Also,
\bb
\Tr X_3\omega_3 = 2 = \|X_3\|_\infty\, .
\label{irreversibility_proof_eq6}
\ee
Thanks to~\eqref{irreversibility_proof_eq5} and~\eqref{irreversibility_proof_eq6}, we see immediately that $X$ is a suitable ansatz for~\eqref{stn}. Using it, we find that $\tn(\omega_3) \geq 2$. Thanks to Theorem~\ref{main_tool_thm}, this implies that
\bb
E_{c,\, \kp}(\omega_3)\geq \ten(\omega_3) = \log_2 \tn(\omega_3)\geq 1\, .
\label{decisive}
\ee
This shows that the theory of entanglement manipulation is irreversible under $\K$-preserving operations, i.e.\ under either non-entangling or PPT-preserving operations.

For completeness we now show that the inequalities in~\eqref{decisive} are in fact all tight; this will establish~\eqref{omega_3_cost} and conclude the proof. Start by observing that
\bb
E_{c,\,\pptp}^\epsilon (\omega_3) \textleq{(vii)} E_{c,\,\sepp}^\epsilon (\omega_3) \textleq{(viii)} E_{c,\,\sepp}^{\mathrm{exact}}(\omega_3)\, ,
\ee
where (vii)~follows from Lemma~\ref{sepp_vs_pptp_lemma}, while (viii)~is an application of the elementary inequality~\eqref{trivial_0}. We now argue that $E_{c,\,\sepp}^{\mathrm{exact}}(\omega_3)\leq 1$, by providing an explicit example of a non-entangling operation $\Lambda$ from a two-qubit to a two-qutrit system such that $\Lambda(\Phi_2) = \omega_3$. Construct
\bb
\Lambda(X) \coloneqq \Tr[X\Phi_2] \omega_3 + \Tr\left[X (\id-\Phi_2)\right] \tau_3\, ,
\ee
where $\tau_3\coloneqq \frac{\id-P_3}{6} = \frac16 \sum_{j\neq k} \ketbra{jk}$. Since $\left\{\Tr[\sigma \Phi_2]:\, \sigma\in \SEP\cap \D(\HH)\right\} = [0,\, 1/2]$, to show that $\Lambda$ is non-entangling it suffices to prove that $\lambda \omega_3 + (1-\lambda)\tau_3\in \SEP$ for all $\lambda\in [0,\,1/2]$. 
Since the claim is trivial for $\lambda=0$, because $\tau_3$ is manifestly separable, by convexity it suffices to prove it for $\lambda=1/2$. Let us write
\bb
\frac12 \left(\omega_3 + \tau_3\right) = \frac{1}{12}\id + \frac16 P_3 -\frac14 \Phi_3 = \pazocal{P}\left( \ketbra{+}\otimes \ketbra{-} \right) ,
\label{irreversibility_proof_eq0}
\ee
where $\ket{\pm}\coloneqq \frac{1}{\sqrt2}\left( \ket{1}\pm \ket{2}\right)$, and $\pazocal{P}$ is the non-entangling quantum operation defined by
\bb
\pazocal{P}(X) &\coloneqq \frac16 \sum_{\pi\in S_3} \int_0^{2\pi} \left( U_{\theta,\pi}\otimes U_{-\theta,\pi} \right) X \left( U_{\theta,\pi}\otimes U_{-\theta,\pi} \right)^\dag \frac{\mathrm{d}^3\theta}{(2\pi)^3}\, , \\
U_{\theta,\pi} &\coloneqq \sum_{j=1}^3 e^{i\theta_j} \ketbraa{\pi(j)}{j}\, .
\label{irreversibility_proof_eq1}
\ee
with $S_3$ denoting the symmetric group over a set of $3$ elements. Note that the last equality in~\eqref{irreversibility_proof_eq0}, which can be proved by inspection using the representation in~\eqref{irreversibility_proof_eq1}, amounts to the sought separable decomposition of the state $\frac12 (\omega_3 + \tau_3)$. This establishes that $E_{c,\, \pptp}^\epsilon(\omega_3) \leq E_{c,\, \sepp}^\epsilon (\omega_3) \leq 1$ and concludes the proof.
\end{proof}

\begin{rem}
One can wonder what other types of states cannot be reversibility manipulated. This is far from obvious, since the axiomatic classes of operations NE or PPTP are typically much more powerful than previously employed types of transformations; in particular, several types of states which have been used to show irreversibility in specific settings are actually \emph{reversible} under NE or PPTP transformations.

The prime example of this is the antisymmetric state, defined on a bipartite system with Hilbert space $\C^d\otimes \C^d$ by
\bb
\alpha_d \coloneqq \frac{\id-F}{d(d-1)}\, ,
\label{antisymmetric}
\ee
where $F\coloneqq \sum_{i,j=1}^d \ketbraa{ij}{ji}$ is the flip operator. This state gained fame as the `universal counterexample' which violates many properties obeyed by other types of quantum states~\cite{math-ent}: for example, it is known that its manipulation is highly irreversible under LOCC --- its distillable entanglement is of order $1/d$, while its entanglement cost is lower bounded by a $d$-independent non-zero constant.~\cite{Christandl2012}.
Curiously, however, $\alpha_d$ was also the first example of a mixed state whose manipulation is reversible under all PPT operations~\cite{Martin-exact-PPT} --- these transformations (hereafter simply denoted with PPT) are all maps $\Lambda$ such that $\idc_R \otimes \Lambda$ is PPT-preserving for all ancillary systems $R$~\cite{Rains2001}, and are therefore a strict subset of the PPT-preserving operations considered herein. The reason why reversibility can be achieved in this setting is that the entanglement cost of $\alpha_d$ can be significantly lowered by considering PPT transformations instead of LOCC, allowing it to reach order $1/d$.
 
However, in Ref.~\cite{irreversibility-PPT}, a related class of states supported on the asymmetric subspace was used to show the \emph{ir}reversibility of entanglement manipulation under PPT operations. In particular, for the state $\rho_p \coloneqq p \ketbra{v_1} + (1-p) \ketbra{v_2}$ with $\ket{v_1} = (\ket{01}-\ket{10})/\sqrt{2}$ and $\ket{v_2} = (\ket{02} - \ket{20})/\sqrt{2}$, it was shown that $E_{d,\rm{PPT}}(\rho_p) < 1 = E_{c,\rm{PPT}} (\rho_p)$. One might then wonder if these states could serve as a similar example of irreversibility for the larger class of PPT-preserving operations. However, this cannot be the case. To see this, we can use the fact that the quantity $E_\eta$ considered in~\cite{irreversibility-PPT} constitutes a lower bound on the distillable entanglement $E_{d,\rm{PPTP}}$, but already in~\cite{irreversibility-PPT} it was shown that $E_\eta(\rho_p) = 1$, meaning that  $E_{d,\rm{PPTP}}(\rho_p) = 1$ and this state is actually reversible under PPT-preserving maps.
 
In a way, this suggests that the state $\omega_3$ is somewhat special, since its entanglement cost cannot be brought down low enough to match its distillable entanglement, even if we allow the extended classes of operations NE or PPTP. It would be interesting to study in more detail the special properties of $\omega_3$ which induce this behaviour, and to understand exactly what types of states exhibit irreversibility in entanglement manipulation under NE and PPTP.
\end{rem}

\begin{rem}
The proof of Theorem~\ref{irreversibility_thm} actually allows us to compute also the zero-error costs of $\omega_3$, namely
\bb
E_{c,\, \sepp}^{\mathrm{exact}}(\omega_3) = E_{c,\, \pptp}^{\mathrm{exact}}(\omega_3) = 1\, .
\ee
As it turns out, the same entanglement cost of exact preparation of $1$ can be achieved by means of a strict subset of PPT-preserving operations, namely, the aforementioned PPT operations. In fact, already the result by Audenaert et al.~\cite{Martin-exact-PPT} guarantees that $E_{c,\, \ppt}^{\mathrm{exact}}(\omega_3) = E_N(\omega_3)=1$, where $E_N$ is the logarithmic negativity~\eqref{logarithmic_negativity}, because $\omega_3$ has vanishing `binegativity',\footnote{This just means that $\big|\omega_3^\Gamma\big|^\Gamma\geq 0$, where $\left|X\right|\coloneqq \sqrt{X^\dag X}$ is the operator absolute value.} as a straightforward check reveals. We can arrive at the same conclusion thanks to the complete characterisation of the exact PPT entanglement cost recently proposed by Wang and Wilde~\cite{Xin-exact-PPT}.
\end{rem}

\section{How much entanglement must be generated to achieve reversibility?} \label{generation_note}

We first recall the framework and the claimed results of Brand\~ao and Plenio~\cite{BrandaoPlenio2} in detail. To begin, we need to fix some notation. For two bipartite quantum systems $AB$, $A'B'$, a given non-negative function $M:\D(\HH_{A'B'})\to \R_+ \cup \{+\infty\}$ on the set of states on $A'B'$ that vanishes on $\K_{A'B'}\cap \D(\HH_{A'B'})$, and some $\delta\geq 0$, we define the set of $\left(M, \delta\right)$-approximately $\K$-preserving maps by
\bb
\kp^{M}_\delta \left(AB\to A'B'\right) \coloneqq \left\{ \Lambda\in \ptp\left(AB\to A'B'\right):\ M\left(\Lambda(\sigma_{AB})\right) \leq \delta \ \ \forall \sigma_{AB} \in \K_{AB} \cap \D(\HH_{AB}) \right\} .
\label{M_approximately_K_preserving}
\ee
Typically, $M$ will be chosen to be an entanglement measure~\cite{Vedral1997,Horodecki-review}. In what follows, we will in fact assume that $M$ is actually a family of functions defined on each bipartite quantum system. Given a sequence $(\delta_n)_{n\in \N}$ with $\delta_n\geq 0$ for all $n$, one can then define the distillable entanglement and entanglement cost under 
\textbf{$\boldsymbol{\left(M, (\delta_n)_n\right)}$-approximately $\mathbfcal{K}$-preserving maps} by setting
\begin{align}
    E_{d,\,\kp^{M}_{(\delta_n)}}^\epsilon\!(\rho_{AB}) &\coloneqq \sup\left\{R>0:\, \limsup_{n\to\infty} \inf_{\Lambda_n \in \kp^{M}_{\delta_n} \left(A^nB^n\to A_0^{\ceil{Rn}}B_0^{\ceil{Rn}}\right)} \frac12 \left\| \Lambda_n\left( \rho_{AB}^{\otimes n} \right) - \Phi_2^{\otimes \ceil{Rn}} \right\|_1 \leq \epsilon \right\} \label{distillable2} \\[.5ex]
    E_{c,\,\kp^{M}_{(\delta_n)}}^\epsilon\!(\rho_{AB}) &\coloneqq \inf\left\{R>0:\, \limsup_{n\to\infty} \inf_{\Lambda_n \in \kp^{M}_{\delta_n} \left(A_0^{\floor{Rn}}B_0^{\floor{Rn}}\to A^n B^n\right)} \frac12 \left\| \Lambda_n\left( \Phi_2^{\otimes \floor{Rn}} \right) - \rho_{AB}^{\otimes n} \right\|_1 \leq \epsilon \right\} . \label{cost2}
\end{align}
We also set $E_{d,\,\kp^{M}_{(\delta_n)}} \coloneqq E_{d,\,\kp^{M}_{(\delta_n)}}^0$ and $E_{c,\,\kp^{M}_{(\delta_n)}} \coloneqq E_{c,\,\kp^{M}_{(\delta_n)}}^0$. With this notation, we say that the theory of entanglement manipulation is reversible under $\left(M, (\delta_n)_n\right)$-approximately $\K$-preserving operations if it holds that $E_{d,\,\kp^{M}_{(\delta_n)}}(\rho_{AB}) = E_{c,\,\kp^{M}_{(\delta_n)}}(\rho_{AB})$ for all states $\rho_{AB}$ on all bipartite quantum systems $AB$. We also say that the theory of entanglement manipulation is reversible under \textbf{$\boldsymbol{M}$-asymptotically $\mathbfcal{K}$-preserving operations} if for all states $\rho_{AB}$ there exists a sequence $(\delta_n)_{n\in \N}$ such that $\delta_n \tendsn{} 0$ and $E_{d,\,\kp^{M}_{(\delta_n)}}(\rho_{AB}) = E_{c,\,\kp^{M}_{(\delta_n)}}(\rho_{AB})$.

To state Brand\~ao and Plenio's conjecture in this framework, we first need to introduce another entanglement monotone closely related to the standard robustness. Recalling first the definition~\eqref{std_rob} of $R^s_\K$, namely, $R_\K^s(\rho) = \inf \left\{ \Tr \delta\!: \delta \!\in\! \K,\, \rho+\delta\!\in\! \K \right\}$, the \emph{\textbf{generalised robustness}} (or \emph{global robustness})~\cite{VidalTarrach, harrow_2003,Steiner2003} is defined similarly as
\begin{equation}
R_\K^g(\rho) = \inf \left\{ \Tr \delta:\, \delta\geq 0,\, \rho+\delta\in \K \right\}.\end{equation}
It is also an entanglement monotone, and many similarities between the two robustness measures have been found; for example, for any pure state $\Psi$ it holds that $R^s_\K(\Psi)=R^g_\K(\Psi)$~\cite{harrow_2003,Steiner2003,taming-PRL,taming-PRA}. The two can, however, exhibit very different properties, as we will explicitly demonstrate below (see also Supplementary  Note~\ref{further_considerations_note}).

With this language, Brand\~ao and Plenio's 
claim 
is that entanglement becomes reversible under $R^g_\SEP$-asymptotically $\SEP$-preserving maps.%
\footnote{We bring to the reader's attention the recently discovered technical issues underlying the proof of the main result of~\cite{BrandaoPlenio2}, as detailed in~\cite{berta_gap}. For this reason, we state the result here as a `conjecture', and its validity is an open question. We nevertheless find it useful to discuss the result here in detail as we take a conceptual inspiration from the framework of~\cite{BrandaoPlenio2}. Our findings are independent of whether this result is ultimately found to be correct or not.}
Employing the simplified notation $\kp^{g}_{(\delta_n)} \coloneqq \kp^{\vphantom{\dag}}_{(\delta_n)}{\vphantom{\kp^t}}^{\hspace{-2.6ex}R^g_\SEP}\hspace{.5ex}$, we formalise their 
claim as follows.

\begin{cj}[(Reversibility under asymptotically non-entangling operations~\cite{BrandaoPlenio2})]
For any state $\rho_{AB}$ acting on a finite-dimensional Hilbert space, there exists a sequence $(\delta_n)_{n\in \N}$ such that $\displaystyle \lim_{n \to \infty} \delta_n = 0$ and
\begin{equation}\begin{aligned}
    E_{d,\,\kp^{g}_{(\delta_n)}}(\rho_{AB}) = E_{c,\,\kp^{g}_{(\delta_n)}}(\rho_{AB}) = E_{r,\K}^\infty(\rho_{AB})\, ,
\end{aligned}\end{equation}
where $E_{r,\K}^\infty$ is the regularised relative entropy measure defined by~\eqref{Er_regularised}.
\end{cj}
Brand\~ao and Plenio argue that the above means that entanglement can be reversibly interconverted without generating macroscopic amounts of it, since the supplemented entanglement is constrained by $\delta_n$ which vanishes in the asymptotic limit. This is certainly true if one quantifies entanglement with the generalised robustness. However, this is an a priori arbitrary choice: one could analogously choose the standard robustness $R^s_\K$ as a quantifier, and consider the (sequence of) sets of operations $\kp^{s}_{(\delta_n)} \coloneqq \kp^{R^s_\SEP}_{(\delta_n)}$ defined by~\eqref{M_approximately_K_preserving} for the special case $M=R^s_\SEP$.
Alternatively, when the negativity $N(\rho) \coloneqq \frac{1}{2}\big(\big\|\rho^\Gamma\big\|_1-1\big)$ is used as the entanglement measure, we can instead look at the operations $\kp^N_{(\delta_n)}$.
Recalling that $R^s_\K(\rho)  \geq  \frac{1}{2}\big(\big\|\rho^\Gamma\big\|_1-1\big) = N(\rho)$, this choice of definition ensures that $\kp^s_\delta \subseteq \kp^N_\delta$.

An extension of our result is then as follows.
\begin{thm} \label{irreversibility_thm2}
For any sequence $(\delta_n)_n$ such that $\delta_n = 2^{o(n)}$, the two-qutrit state $\omega_3$ defined by~\eqref{omega_3} satisfies that
\bb
E_{d,\, \sepp^s_{(\delta_n)}}(\omega_3) &= E_{d,\, \pptp^s_{(\delta_n)}}(\omega_3) = E_{d,\, \sepp^N_{(\delta_n)}}(\omega_3) = E_{d,\, \pptp^N_{(\delta_n)}}(\omega_3) \\& = \log_2 \frac32 \approx 0.585
\label{omega_3_distillable2}
\ee
but
\bb
E_{c,\, \sepp^s_{(\delta_n)}}^\epsilon (\omega_3) &= E_{c,\, \pptp^s_{(\delta_n)}}^\epsilon (\omega_3) = E_{c,\, \sepp^N_{(\delta_n)}}^\epsilon (\omega_3) = E_{c,\, \pptp^N_{(\delta_n)}}^\epsilon (\omega_3) \\ &= 1
\label{omega_3_cost2}
\ee
for all $\epsilon\in [0,1/2)$. In particular, the resource theory of entanglement is irreversible under any class of operations which does not generate an amount of entanglement that grows exponentially in $n$, as quantified by either the standard robustness $R^s_\K$ or by the negativity $N$.
\end{thm}

Here $\delta_n = 2^{o(n)}$ means that $\delta_n$ has a sub-exponential behaviour in $n$: for any $k > 0$, there exists $m \in \N$ such that $\delta_n < 2^{kn}$ for all $n \geq m$. Consequently, to have any hope of recovering reversibility, one needs $\delta_n$ (and hence the generated entanglement) to grow exponentially: there must exist a choice of $k>0$ such that, for all $m \in \N$, $\delta_n \geq 2^{kn}$ for at least one $n \geq m$; in other words, $\delta_n$ is lower bounded by $2^{kn}$ infinitely often.\footnote{Following the original notation introduced by Hardy and Littlewood~\cite{hardy_1914}, we could denote this behaviour as $\delta_n = 2^{\Omega(n)}$. However, the commonly used notation $\Omega(n)$ actually refers to a stronger property~\cite{knuth_1976}, so we do not use it here.}

\begin{rem*}
One can wonder whether $R^g_\K$ can be considered as a more operationally meaningful measure of the supplemented entanglement, justifying its use over other measures such as $R^s_\K$ or $N$ and thus substantiating the reversibility 
conjecture of~\cite{BrandaoPlenio2} over the irreversibility result of Theorem~\ref{irreversibility_thm2}. We do not believe that there is any compelling reason to do so: although $R^g_\K$ admits a very general operational interpretation as the quantifier of the advantage that a given state provides in channel discrimination tasks~\cite{Ryuji-Bart,Bae2019,taming-PRL}, $R^s_\K$ has an arguably even more relevant application, as it exactly quantifies the one-shot entanglement cost under non-entangling operations~\cite{brandao_2011}. On the technical side, both of the quantities suffer from very similar issues in the many-copy limit, as they do not satisfy asymptotic continuity~\cite{Donald1999}. 

This is no coincidence, as Brand\~{a}o and Plenio have shown that the choice of an asymptotically continuous monotone to quantify the supplemented entanglement leads to the trivialisation of the framework~\cite[Section~V]{BrandaoPlenio2}. Note that almost all the most commonly used entanglement measures and all of those with the strongest operational meanings are in fact asymptotically continuous. Examples include the entanglement of formation~\cite{nielsen_2000, tightuniform}, the (LOCC) entanglement cost~\cite{tightuniform}, the squashed entanglement~\cite{squashed, Alicki-Fannes}, and the (regularised) relative entropy of entanglement~\cite{Donald1999, tightuniform}. In fact, among the most widely used entanglement measures, the only one that is \emph{not} asymptotically continuous is the logarithmic negativity~\cite{negativity,plenioprl}. For this reason, we regard the failure of asymptotic continuity for the robustnesses as an issue of some conceptual importance, one that may cast some doubts on the status of approximately $\K$-preserving maps.

The choice of $R^g_\K$ in~\cite{BrandaoPlenio2} is motivated \textit{a posteriori} by the fact that it 
is claimed to lead to reversibility, rather than by \textit{a prori} physical considerations. We are therefore inclined to believe that there is \emph{no} unique and indisputable choice of a suitable entanglement measure, and we consider Theorem~\ref{irreversibility_thm2} to serve as evidence that the irreversibility of entanglement revealed in our work is very robust, and that avoiding it requires a very careful and deliberate choice of an entanglement monotone --- according to other, equally reasonable choices, the generated entanglement must be exponentially large. 
\end{rem*}

\begin{rem*}
We should also note in passing that between the two sets of operations $\sepp^s_{(\delta_n)}$ and $\sepp^N_{(\delta_n)}$ that we considered in  Theorem~\ref{irreversibility_thm2} above, the former is perhaps more adherent to our intuitive notion of approximately non-entangling maps. Indeed, since the standard robustness of entanglement $R^s_\SEP$ is a faithful measure, i.e.\ it is strictly positive on all entangled states, transformations in $\sepp^s_0$ are in fact non-entangling. Transformations in $\sepp^N_0$, on the contrary, map separable states to PPT states that can very well be entangled. However, since $\sepp^s_\delta \subseteq \sepp^N_\delta$, showing the irreversibility of entanglement under the operations $\sepp^N_{\delta}$ constitutes a strictly stronger result, and indeed shows also that generating PPT entangled states is not sufficient to recover reversibility --- any reversible protocol must create highly non-PPT entanglement.
\end{rem*}

The first step in proving the Theorem is the following lemma, which establishes an approximate monotonicity of $R^s_\K$ under approximately $\K$-preserving maps, whether quantified by $R^s_\K$ itself or by the negativity.
\begin{lemma}\label{lem:approx_monotone2}
For any $\Lambda \in \kp^s_\delta({AB \to A'B'})$, it holds that
\begin{equation}\begin{aligned}
    R^s_\K(\Lambda(\rho_{AB}))+1 \leq  ( 1 + 2 \delta) \, (R^s_\K(\rho_{AB})+1) .
\end{aligned}\end{equation}
Similarly, for any $\Lambda \in \kp^N_\delta({AB \to A'B'})$, it holds that
\begin{equation}\begin{aligned}
    \frac{1}{2}\left(\left\|\Lambda(\rho_{AB})^\Gamma\right\|_1 + 1\right) \leq  ( 1 + 2 \delta) \, (R^s_\K(\rho_{AB})+1) .
\end{aligned}\end{equation}
\end{lemma}
\begin{proof}
Let us take $\Lambda \in \kp^s_\delta({AB \to A'B'})$ and consider any feasible decomposition for the standard robustness of $\rho$ as $\rho = \sigma - \tau$ where $\sigma, \tau \in \K$ (noting that these are in general only unnormalised states). Since $R^s_\K\left(\frac{\Lambda(\sigma)}{\Tr \sigma}\right) \leq \delta$ and $R^s_\K\left(\frac{\Lambda(\tau)}{\Tr \tau}\right) \leq \delta$, for any $\epsilon >0 $ there exist decompositions
\begin{equation}\begin{aligned}
    \frac{\Lambda(\sigma)}{\Tr \sigma} =  \sigma' -  \tau', \qquad  \frac{\Lambda(\tau)}{\Tr \tau} =  \sigma'' -  \tau'',
\end{aligned}\end{equation}
for some $\sigma', \sigma'', \tau', \tau'' \in \K$ such that $\Tr \tau', \Tr \tau'' \leq \delta + \epsilon$. Then
\begin{equation}\begin{aligned}
    \Lambda(\rho) = \left[ (\Tr \sigma) \sigma' + (\Tr \tau) \tau'' \right] - \left[ (\Tr \sigma) \tau' + (\Tr \tau) \sigma'' \right].
\end{aligned}\end{equation}
This constitutes a valid feasible solution for the robustness of $\Lambda(\rho)$, giving
\begin{equation}\begin{aligned}
    R^s_\K(\Lambda(\rho)) + 1 &\leq \Tr \sigma \Tr \sigma' + \Tr \tau \Tr \tau''\\
    &= (\Tr \tau + 1) (\Tr \tau' + 1) + \Tr \tau \Tr \tau''\\
    &\leq (\Tr \tau + 1) (\Tr \tau' + 1) + (\Tr \tau + 1) \Tr \tau''\\
    &\leq (\Tr \tau + 1) (\delta + \epsilon + 1) + (\Tr \tau + 1) (\delta + \epsilon).
\end{aligned}\end{equation}
Since this holds for any feasible value of $\Tr \tau$, it must also hold that $R^s_\K(\Lambda(\rho)) + 1 \leq (1+2\delta + 2 \epsilon) (R^s_\K(\rho)+1)$, as $R^s_\K$ is defined precisely as the infimum of all feasible values of $\Tr \tau$. Since $\epsilon > 0$ was arbitrary, the desired statement follows.

The case of $\Lambda \in \kp_\delta^N$ is similar. For any decomposition $\rho = \sigma - \tau$ with $\sigma, \tau \in \K$, the triangle inequality gives
\begin{equation}\begin{aligned}
    \left\|\Lambda(\rho)^\Gamma\right\|_1 &\leq \Tr \sigma \left\|\frac{\Lambda(\sigma)^\Gamma}{\Tr \sigma}\right\|_1 + \Tr \tau \left\|\frac{\Lambda(\tau)^\Gamma}{\Tr \tau}\right\|_1\\
    &\leq (\Tr \sigma)(1 + 2\delta) + (\Tr \tau) (1 + 2\delta)\\
    &= (1 + 2 \Tr \tau) (1 + 2\delta)
\end{aligned}\end{equation}
where we used that $\frac{1}{2}\big(\big\|\Lambda(\sigma)^{\Gamma}\big\|_1-1\big)\leq \delta$ and analogously for $\tau$. Rearranging, we get
\begin{equation}\begin{aligned}
    \frac{1}{2} \left(\left\|\Lambda(\rho)^\Gamma\right\|_1 + 1\right) &\leq \delta + 2 \delta \Tr \tau + \Tr \tau + 1\\
    &\leq (1 + 2 \delta) (\Tr \tau + 1).
\end{aligned}\end{equation}
Minimising over all feasible values of $\Tr \tau$ gives $\frac{1}{2} \left(\big\|\Lambda(\rho)^\Gamma\big\|_1 + 1\right) \leq (1+2\delta) (R^s_\K(\rho)+1)$, as was to be shown.
\end{proof}

Theorem~\ref{irreversibility_thm2} then relies on the following extension of Theorem~\ref{main_tool_thm}.

\begin{thm} \label{main_tool_thm2}
For $\K=\SEP$ or $\K=\PPT$, the entanglement cost under approximately $\K$-preserving operations satisfies that
\bb
\inf_{\epsilon\, \in\, [0,\,1/2)} E_{c,\,\kp^s_{(\delta_n)}}^\epsilon (\rho) \geq \tl_{\!\K}(\rho)\, \geq \ten(\rho)
\label{main_tool_formal2}
\ee
and
\bb
\inf_{\epsilon\, \in\, [0,\,1/2)} E_{c,\,\kp^N_{(\delta_n)}}^\epsilon (\rho) \geq \ten(\rho)
\label{main_tool_formal2b}
\ee
for any sequence $(\delta_n)_n$ such that $\delta_n = 2^{o(n)}$.
\end{thm}
\begin{proof}
In complete analogy with the proof of Theorem~\ref{main_tool_thm}, for any feasible sequence $(\Lambda_n)_{n\in \N}$ of maps such that $\Lambda_n \in \kp^s_{\delta_n}$ and $\big\|\Lambda_n(\Phi_2^{\otimes \floor{Rn}}) - \rho^{\otimes n}\big\|_1 \eqqcolon \epsilon_n \tendsn{} \epsilon<1/2$, we can write
\bb
2^{\floor{Rn}} &= 1 + R_\K^s\left( \Phi_2^{\otimes \floor{Rn}}\right) \\
&\geq (1+2\delta_n)^{-1} \left( 1 + R_\K^s\left( \Lambda_n\left(\Phi_2^{\otimes \floor{Rn}}\right)\right) \right) \\
&\geq (1+2\delta_n)^{-1} (1-2\epsilon_n) \left( 1+\tsr_{\!\K}\left(\rho^{\otimes n}\right) \right)
\ee
where now we used Lemma~\ref{lem:approx_monotone2} to incorporate the approximate entanglement non-generation. Then
\bb
R &= \lim_{n\to\infty} \frac{\floor{Rn}}{n} \\
&\geq \limsup_{n\to\infty} \frac1n \log_2 \left( 1+\tsr_{\!\K}\left(\rho^{\otimes n}\right) \right) + \liminf_{n\to\infty} \frac1n \log_2 (1-2\epsilon_n) - \limsup_{n\to\infty} \frac1n \log_2 (1+2\delta_n) \\
&= \tl_{\!\K} (\rho) - \limsup_{n\to\infty} \frac1n \log_2 (1+2\delta_n) \\
&= \tl_{\!\K} (\rho)\, ,
\ee
where in the last line we used the fact that
\bb
\lim_{n \to \infty} \frac{\log_2 \delta_n}{n} = 0
\ee
by hypothesis. The rest of the proof of the first part of the Theorem is then exactly the same as in Theorem~\ref{main_tool_thm}.

The second part of the proof is very similar in that it follows Theorem~\ref{main_tool_thm}, but it goes directly to the tempered negativity $N_\tau$ rather than the intermediate quantity $\tl_{\!\K}$. Taking now any feasible sequence $(\Lambda_n)_{n\in \N}$ with $\Lambda_n \in \kp^N_{\delta_n}$, we have
\bb
2^{\floor{Rn}} &= 1 + R_\K^s\left( \Phi_2^{\otimes \floor{Rn}}\right) \\
&\textgeq{(i)} (1+2\delta_n)^{-1} \frac{ 1 + \left\| \Lambda_n\left(\Phi_2^{\otimes \floor{Rn}}\right)^\Gamma\right\|_1 }{2} \\
&\textgeq{(ii)} (1+2\delta_n)^{-1} \frac{ 1 + N_\tau \left(\left. \Lambda_n\left(\Phi_2^{\otimes \floor{Rn}}\right) \right| \rho^{\otimes n} \right) }{2} \\
&\textgeq{(iii)} (1+2\delta_n)^{-1} (1-2\epsilon_n) \frac{  1+  N_\tau \left(\rho^{\otimes n}\right) }{2},
\ee
where in (i) we used Lemma~\ref{lem:approx_monotone2}, in (ii) Proposition~\ref{elementary_prop}(b), and in (iii) the $\epsilon$-lemma (Lemma~\ref{crucial_lemma}). This gives
\bb
R &= \lim_{n\to\infty} \frac{\floor{Rn}}{n} \\
&\geq \limsup_{n \to \infty}  \frac1n \log_2 \frac{  1+ N_\tau \left(\rho^{\otimes n}\right)  }{2} + \liminf_{n\to\infty} \frac1n \log_2 (1-2\epsilon_n) - \limsup_{n\to\infty} \frac1n \log_2 (1+2\delta_n)\\
&\geq \limsup_{n\to\infty} \left(\frac1n \log_2  N_\tau(\rho^{\otimes n}) - \frac1n\right) - \limsup_{n\to\infty} \frac1n \log_2 \delta_n \\
&\geq \limsup_{n\to\infty} \left(\frac1n \log_2 N_\tau(\rho)^n - \frac1n\right) \\
&= E_N^\tau(\rho)
\ee
using the super-multiplicativity of the tempered negativity (Proposition~\ref{elementary_prop}(e)) and the assumption that $\delta_n = 2^{o(n)}$.
\end{proof}

As the final ingredient that will be required in the proof of Theorem~\ref{irreversibility_thm2}, we need to show that the distillable entanglement cannot increase even if we allow sub-exponential entanglement generation.

\begin{lemma}\label{lem:distillation_allequal}
Consider any state $\rho_{AB}$ and let $\K=\SEP$ or $\K=\PPT$. For any $\epsilon \in [0,1)$ and any non-negative sequence $(\delta_n)_n$ 
it holds that
\bb
E_{d,\,\kp^s_{(\delta_n)}}^\epsilon\!(\rho) = E_{d,\,\kp^N_{(\delta_n)}}^\epsilon\!(\rho)= E_{d,\,\kp^g_{(\delta_n)}}^\epsilon\!(\rho)\, .
\label{eq:distillation_allequal}
\ee
Moreover, if $\delta_n = 2^{o(n)}$ then also
\bb
E_{d,\,\kp}^\epsilon(\rho) = E_{d,\,\kp^s_{(\delta_n)}}^\epsilon\!(\rho) = E_{d,\,\kp^N_{(\delta_n)}}^\epsilon\!(\rho)= E_{d,\,\kp^g_{(\delta_n)}}^\epsilon\!(\rho)\, .
\label{eq:distillation_really_allequal}
\ee
\end{lemma}

\begin{proof}
The proof will proceed in two steps. First, we establish expressions for the minimal error achievable in distillation with $(M, \delta)$-approximately $\K$-preserving operations. Then, we argue that, for sub-exponential $\delta_n$, the contributions from the parameter $\delta_n$ to the performance of the distillation task can be absorbed into the transformation rates, effectively preventing any improvement in the asymptotic distillation error.

Consider first any operation $\Lambda \in \kp^g_\delta\left( A^nB^n\to A_0^m B_0^m\right)$ for a fixed $\delta$, where $m$ is a generic positive integer. We would then like to understand exactly the error in the transformation from $\rho^{\otimes n}$ to the maximally entangled state at some rate $R$, which we will for now quantify using the fidelity $F(\omega,\tau)\coloneqq \|\sqrt{\omega}\sqrt{\tau}\|_1^2$. 
We have
\bb
F\left(\Lambda(\rho^{\otimes n}), \Phi_2^{\otimes m}\right) &= \Tr  \Lambda\left(\rho^{\otimes n}\right) \Phi_2^{\otimes m}\\
&= \Tr \rho^{\otimes n} \Lambda^\dagger\big( \Phi_2^{\otimes m}\big)
\ee
using that $\Phi_2$ is a pure state. Notice now that, by definition of the generalised robustness, for any $\sigma \in \K \cap \D(\HH)$ we have $\Lambda(\sigma) \leq (1+\delta) \sigma'$ for some  $\sigma' \in \K \cap \D(\HH)$. Thus, for any $\sigma \in \K \cap \D(\HH)$,
\bb
\Tr \Lambda^\dagger\big( \Phi_2^{\otimes m}\big) \sigma &\leq \Tr \Phi_2^{\otimes m} (1+\delta) \sigma'\\
&\leq (1+\delta) \frac{1}{2^{m}}
\ee
where in the first line we used the positivity of $\Phi_2$, and in the second that the maximal overlap of $\Phi_2^{\otimes m}$ with a separable state is $\frac{1}{2^{m}}$~\cite{shimony_1995}. Noting also that $0 \leq \Lambda^\dagger\big( \Phi_2^{\otimes m}\big) \leq \id$ due to the fact that $\Lambda$ is positive and trace preserving, we get
\bb\label{eq:brandao_fidelity}
\sup_{\Lambda \in \kp^g_\delta} F\left(\Lambda(\rho^{\otimes n}), \Phi_2^{\otimes m}\right) &= \sup_{\Lambda \in \kp^g_\delta} \Tr \rho \Lambda^\dagger\big(\Phi_2^{\otimes m}\big)\\
&\leq \sup \lset \Tr \rho^{\otimes n} W \bar 0 \leq W \leq \id,\; \Tr W \sigma \leq (1+\delta) \frac{1}{2^{m}} \ \ \forall \sigma \in \K \cap \D(\HH) \rset \\
&\eqqcolon \varphi_{n,m}(\delta)\, .
\ee

The argument for operations $\kp^N_\delta$, where now negativity is the figure of merit, proceeds analogously. We have, for any $\sigma \in \K \cap \D(\HH)$ and any $\Lambda \in \kp^N_\delta$, that
\bb
\Tr \Lambda^\dag\big( \Phi_2^{\otimes m}\big) \sigma &= \Tr \big(\Phi_2^{\otimes m}\big)^\Gamma \Lambda(\sigma)^\Gamma\\
&\leq \left\|\big(\Phi_2^{\otimes m}\big)^\Gamma\right\|_\infty \big\|\Lambda(\sigma)^\Gamma\big\|_1\\
&\leq \frac{1}{2^{m}} (1 + 2 \delta)
\ee
where in the second line we used the Cauchy--Schwarz inequality, and in the third we used that $\frac{1}{2}\left(\big\|\Lambda(\sigma)^\Gamma\big\|-1\right) \leq \delta$ and that $\big\|\Phi_2^\Gamma\big\|_\infty = \frac{1}{2}$. This gives
\bb\label{eq:brandao_fidelity2}
\sup_{\Lambda \in \kp^N_\delta} F\left(\Lambda(\rho^{\otimes n}), \Phi_2^{\otimes m}\right) &\leq \sup \lset \Tr \rho^{\otimes n} W \bar 0 \leq W \leq \id,\; \Tr W \sigma \leq (1+2 \delta) \frac{1}{2^{m}} \ \ \forall \sigma \in \K \cap \D(\HH) \rset \\
&= \varphi_{n,m}(2\delta)\, ,
\ee
exactly the same as in~\eqref{eq:brandao_fidelity} up to the 
substitution $\delta \mapsto 2\delta$, which we will see to be immaterial.

For the other direction, define the separable~\cite{Horodecki1999} state $\tau_m\coloneqq (\id-\Phi_2^{\otimes m})/(4^m-1)\in \K\cap \D(\HH)$, which satisfies also that~\cite{Horodecki1999} $\Phi_2^{\otimes m} + (2^{m}-1) \tau_m \in \K$. Note that $R^s_\K\left(\Phi_2^{\otimes m}\right) = 2^m-1$ by~\eqref{robustness_k_ebits}, and therefore $\tau_m$ is the state that achieves the minimum in the definition of $R^s_\K$ \eqref{std_rob} for $\Phi_2^{\otimes m}$.
Take any operator $W$ such that $0 \leq W \leq \id$ and $\Tr W \sigma \leq (1+\delta) \frac{1}{2^{m}}$ for all $\sigma \in \K \cap \D(\HH)$, where we assume that $(1+\delta) \frac{1}{2^{m}}\leq 1$ without loss of generality. Define the map $\Gamma_W$ by
\begin{equation}\begin{aligned}
    \Gamma_W (X) = \left(\Tr W X\right) \Phi_2^{\otimes m} + \left( \Tr (\id - W) X \right) \tau_m\, .
\end{aligned}\end{equation}
This map is explicitly completely positive and trace preserving, and we can furthermore verify that, 
since $0\leq \Tr W\sigma\leq (1+\delta) \frac{1}{2^{m}}$, for any $\sigma \in  \K \cap \D(\HH)$ we have
\bb
\Gamma_W(\sigma) \in \co\left\{ \tau_m,\, (1+\delta) \frac{1}{2^{m}}\,\Phi_2^{\otimes m} + \left( 1-(1+\delta) \frac{1}{2^{m}} \right) \tau_m \right\} ,
\ee
where $\co$ denotes the convex hull. Since $\tau_m\in \K$ and $\Phi_2^{\otimes m} + (2^{m}-1) \tau_m \in \K$, also
\bb
(1+\delta) \frac{1}{2^{m}}\,\Phi_2^{\otimes m} + \left( 1-(1+\delta) \frac{1}{2^{m}} \right) \tau_m + \delta \tau_m = \frac{1+\delta}{2^m} \left( \Phi_2^{\otimes m} + \left( 2^m-1 \right) \tau_m \right) \in \K\, .
\ee
The convexity of $R^s_\K$ (which follows directly from the convexity of $\K$ itself) then implies that $R^s_\K(\Gamma_W(\sigma)) \leq \delta$, and hence $\Gamma_W \in \kp^s_{\delta}$. 
Noting that $F\big(\Gamma_W(\rho^{\otimes n}), \Phi_2^{\otimes m}\big) = \Tr \rho^{\otimes n} W$, optimising over all feasible $W$ yields
\begin{equation}\begin{aligned}
    \sup_{\Lambda \in \kp^s_\delta} F\left(\Lambda(\rho^{\otimes n}), \Phi_2^{\otimes m}\right) &\geq \sup \lset \Tr \rho^{\otimes n} W \bar 0 \leq W \leq \id,\; \Tr W \sigma \leq (1+\delta) \frac{1}{2^{m}} \ \ \forall \sigma \in \K \cap \D(\HH) \rset \\
    &= \varphi_{n,m}(\delta)\, ,
\end{aligned}\end{equation}
where the function $\varphi_{n,m}(\delta)$ is defined in~\eqref{eq:brandao_fidelity}.
Using the inclusion between the sets of operations $\kp^s_{\delta} \subseteq \kp^g_{\delta}$ and $\kp^s_{\delta} \subseteq \kp^N_{\delta}$, we therefore obtain that for all $n,m,\delta$ such that $(1+\delta)\frac{1}{2^m}\leq 1$
\bb
\sup_{\Lambda \in \kp^s_\delta} F\left(\Lambda(\rho^{\otimes n}), \Phi_2^{\otimes m}\right) &= \sup_{\Lambda \in \kp^g_\delta} F\left(\Lambda(\rho^{\otimes n}), \Phi_2^{\otimes m}\right) = \varphi_{n,m}(\delta)\, ,\\
\varphi_{n,m}(\delta) &\leq \sup_{\Lambda \in \kp^N_\delta} F\left(\Lambda(\rho^{\otimes n}), \Phi_2^{\otimes m}\right) \leq \varphi_{n,m}(2\delta)\, .
\label{fidelity_distillation_bounds}
\ee
We now pass from the fidelity to the trace distance. We claim that for $M \in \{ R^s_\K,\, R^g_\K,\, N \}$ it holds that
\bb
\inf_{\Lambda \in \kp^M_\delta} \frac12 \left\|\Lambda(\rho^{\otimes n}) - \Phi_2^{\otimes m}\right\|_1 = 1 - \sup_{\Lambda \in \kp^M_\delta} F\left(\Lambda(\rho^{\otimes n}), \Phi_2^{\otimes m}\right) .
\label{fidelity_to_trace_norm}
\ee
To see why this is the case, it suffices to observe that we can always twirl the output of $\Lambda$ without loss of generality, i.e.\ we can substitute $\Lambda \mapsto \T_m\circ \Lambda$, where $\T_m$ is defined as in~\eqref{twirling}. Doing so does not change the fact that $\Lambda \in \kp^M_\delta$, simply because $M$ is convex and invariant under local unitaries; furthermore, twirling leaves $F\left(\Lambda(\rho^{\otimes n}), \Phi_2^{\otimes m}\right) = \Tr \Lambda(\rho^{\otimes n}) \Phi_2^{\otimes m}$ invariant and never increases $\frac12 \left\|\Lambda(\rho^{\otimes n}) - \Phi_2^{\otimes m}\right\|_1$ (see~\eqref{twirling_decreases_trace_norm}). At the same time, since the output state of $\Lambda$ is now twirled, and hence it is a convex combination of the orthogonal states $\Phi_2^{\otimes m}$ and $\tau_m$, we have that  
\bb
\frac12 \left\|\Lambda(\rho^{\otimes n}) - \Phi_2^{\otimes m}\right\|_1 = 1 - \Tr \Lambda(\rho^{\otimes n}) \Phi_2^{\otimes m} = 1 - F\left(\Lambda(\rho^{\otimes n}), \Phi_2^{\otimes m}\right) ,
\ee
completing the proof of~\eqref{fidelity_to_trace_norm}. Combining~\eqref{fidelity_distillation_bounds} with~\eqref{fidelity_to_trace_norm} we arrive at
\bb
\inf_{\Lambda \in \kp^s_\delta} \frac12 \left\|\Lambda(\rho^{\otimes n}) - \Phi_2^{\otimes m}\right\|_1 &= \inf_{\Lambda \in \kp^g_\delta} \frac12 \left\|\Lambda(\rho^{\otimes n}) - \Phi_2^{\otimes m}\right\|_1 = 1-\varphi_{n,m}(\delta)\, ,\\
1-\varphi_{n,m}(2\delta) &\leq \inf_{\Lambda \in \kp^N_\delta} \frac12 \left\|\Lambda(\rho^{\otimes n}) - \Phi_2^{\otimes m}\right\|_1 \leq 1-\varphi_{n,m}(\delta)\, .
\label{error_distillation_bounds}
\ee
The above relations imply immediately that
\bb
E_{d,\,\kp^s_{(\delta_n)}}^\epsilon\!(\rho) = E_{d,\,\kp^g_{(\delta_n)}}^\epsilon\!(\rho) \leq E_{d,\,\kp^N_{(\delta_n)}}^\epsilon\!(\rho)
\ee
for all sequences $(\delta_n)_n$. We now proceed to show that in fact also $E_{d,\,\kp^N_{(\delta_n)}}^\epsilon\!(\rho)\leq E_{d,\,\kp^g_{(\delta_n)}}^\epsilon\!(\rho)$, which will complete the proof of the claim~\eqref{eq:distillation_allequal}. To this end, fix $\zeta>0$, and let $R \geq E_{d,\,\kp^N_{(\delta_n)}}^\epsilon\!(\rho) - \zeta$ be an achievable rate for entanglement distillation from $\rho$ at error $\epsilon$ with operations from $\kp^N_{(\delta_n)}$; by definition, this implies that
\bb
\limsup_{n\to\infty} \inf_{\Lambda \in \kp^N_\delta} \frac12 \left\|\Lambda(\rho^{\otimes n}) - \Phi_2^{\otimes \ceil{Rn}}\right\|_1 \leq \epsilon\, ,
\ee
and thus, by~\eqref{error_distillation_bounds}, that $\liminf_{n\to\infty} \varphi_{n,\ceil{Rn}}(2\delta_n)\geq 1-\epsilon$. Now, since $\frac{1+\delta_n}{2^{\ceil{(R-\zeta)n}}}\geq \frac{1+2\delta_n}{2^{\ceil{Rn}}}$ for all sufficiently large $n$, we see by direct inspection of the definition of $\varphi_{n,m}(\delta)$ in~\eqref{eq:brandao_fidelity} that
\bb
\varphi_{n,\, \ceil{(R-\zeta)n}}(\delta_n)\geq \varphi_{n,\ceil{Rn}}(2\delta_n)
\ee
for all sufficiently large $n$, thus implying that also
\bb
\liminf_{n\to\infty} \varphi_{n,\, \ceil{(R-\zeta)n}}(\delta_n)\geq 1-\epsilon\, .
\ee
Again by~\eqref{error_distillation_bounds}, this is equivalent to stating that $R-\zeta$ is an achievable rate for $E_{d,\,\kp^g_{(\delta_n)}}^\epsilon\!(\rho)$. Putting all together, we see that
\bb
E_{d,\,\kp^g_{(\delta_n)}}^\epsilon\!(\rho) \geq R-\zeta \geq E_{d,\,\kp^N_{(\delta_n)}}^\epsilon\!(\rho) - 2\zeta\, .
\ee
Since $\zeta>0$ was arbitrary, we conclude that indeed $E_{d,\,\kp^g_{(\delta_n)}}^\epsilon\!(\rho) \geq E_{d,\,\kp^N_{(\delta_n)}}^\epsilon\!(\rho)$, as desired. This establishes~\eqref{eq:distillation_allequal}.

The reasoning for~\eqref{eq:distillation_really_allequal} is analogous. If $\delta_n=2^{o(n)}$ then $\frac{1+\delta_n}{2^{\ceil{Rn}}}\leq \frac{1}{2^{\ceil{(R-\zeta)n}}}$ for all $\zeta>0$ and all sufficiently large $n$. Therefore, if $R$ is an achievable rate for $E_{d,\,\kp^g_{(\delta_n)}}^\epsilon\!(\rho)$ then $R-\zeta$ must be achievable for $E_{d,\,\kp}^\epsilon(\rho)$; since $\zeta>0$ is arbitrary, this ensures that $E_{d,\,\kp^g_{(\delta_n)}}^\epsilon\!(\rho)\leq E_{d,\,\kp}^\epsilon(\rho)$. The other inequality follows trivially from the inclusion $\kp \subseteq \kp^g_{\delta_n}$ for any $\delta_n$. This proves also~\eqref{eq:distillation_really_allequal}.
\end{proof}

\begin{proof}[Proof of Theorem~\ref{irreversibility_thm2}]
For the entanglement cost, it holds that
\begin{equation}
    1 \texteq{(i)} E_N^\tau(\omega_3) \textleq{(ii)} E_{c,\,\kp^N_{(\delta_n)}}^\epsilon (\omega_3) \textleq{(iii)} E_{c,\,\kp^s_{(\delta_n)}}^\epsilon (\omega_3) \textleq{(iii)} E_{c,\,\kp}^\epsilon(\omega_3) \texteq{(i)} 1
\end{equation}
and so all inequalities must be equalities. Here, we used (i)~our results in Theorem~\ref{irreversibility_thm}, (ii)~Theorem~\ref{main_tool_thm2}, and (iii)~the inclusion $\kp \subseteq \kp^s_{\delta_n} \subseteq \kp^N_{\delta_n}$ for any $\delta_n$.

For the distillable entanglement, Lemma \ref{lem:distillation_allequal} gives
\begin{equation}
    E_{d,\,\kp^s_{(\delta_n)}}\!(\omega_3) = E_{d,\,\kp^g_{(\delta_n)}}\!(\omega_3) = E_{d,\,\kp^N_{(\delta_n)}}\!(\omega_3) = E_{d,\,\kp} (\omega_3) = \log_2 \frac32,
\end{equation}
where the last equality was already shown in Theorem~\ref{irreversibility_thm}.
\end{proof}

\section[Fluctuations in entanglement manipulation protocols and comparison with thermodynamics]{Fluctuations in entanglement manipulation protocols\texorpdfstring{\\}{ }and comparison with thermodynamics} \label{intermezzo_note}

When studying the asymptotic behaviour of quantum systems, allowing for some type of microscopic fluctuation is arguably very natural from a thermodynamical perspective. Indeed, thermodynamics is a theory of macroscopic states and quantities, which should be unaffected by microscopic degrees of freedom and by the fluctuations associated with them.
Although in Supplementary Note~\ref{generation_note} we justified our framework for entanglement manipulation by demonstrating the issues that come with enforcing only asymptotic (rather than exact) entanglement non-generation, it is natural to wonder: Is it not more physically meaningful to allow for fluctuations at the level of microscopic systems, rather than completely forbid entanglement generation? Let us discuss the types of fluctuations that are allowed in our entanglement manipulation framework, to illustrate why our setting is in fact very similar to others considered previously in the literature, and to pinpoint the key difference with Brand\~{a}o and Plenio's approach~\cite{BrandaoPlenio1, BrandaoPlenio2}.

Although fluctuations are a natural component of our theory, not all fluctuations are equal:  in the context of state-to-state transformations in entanglement theory (cf.\ the discussion in Section~\ref{subsec_distillable_cost}), and more generally in that of quantum resource theories~\cite{RT-review}, we find it useful to introduce the following classification.
\begin{enumerate}[(I)]
    \item \emph{Fluctuations at the level of microscopic transformation error.} Specifically, one typically requires that transformations of the form $\rho \to \omega$ are only realised approximately as $\Lambda(\rho^{\otimes n}) \approx \omega^{\otimes m}$ with some non-zero error (usually measured by the trace norm), and only in the limit $n \to \infty$ does the error vanish. Such fluctuations are common to most approaches to quantum resource theories, including quantum thermodynamics as well as entanglement theory~\cite{RT-review}. Crucially, without such fluctuations, even quantum thermodynamics is no longer a reversible theory, and many state transformations become impossible~\cite{fang_2020,regula_2020-2,fang_2022}.
    
    We do note, however, that in some contexts 
    one may in fact be interested in exact, i.e.\ zero-error, entanglement distillation~\cite{Hayashi2006, Hayashi-VV} or dilution~\cite{Terhal-Horodecki,Martin-exact-PPT,Xin-exact-PPT}. In the language we are using here, such settings would explicitly rule out microscopic transformation errors (as in \eqref{distillable_exact}--\eqref{cost_exact}), or would require them to vanish exceedingly --- i.e.\ super-exponentially --- fast~\cite{Duan2016}. For these reasons, these approaches are generally regarded as less physically motivated than the one we take here, which explicitly allows for vanishingly small transformation errors. 
    
    \item \emph{Fluctuations in the resources consumed by the process.} For example, in quantum thermodynamics, one often includes a small ancillary system to activate a desired transformation~\cite{Brandao-thermo,sparaciari_2017}. This is arguably necessary to circumvent the energy super-selection rule that prohibits the creation of coherence between different energy levels under an energy-preserving unitary: without the ancilla, an energy-incoherent state would remain energy-incoherent under thermal operations, thus preventing dilution of resources altogether~\cite{Brandao-thermo}. Including a small ancillary system that carries some coherence in the energy eigenbasis solves the problem, but its small size means that its contributions to the rate will asymptotically vanish, leaving the underlying physics unaffected.

    More detailed insights on sensible physical requirements to impose on the ancillary system can be deduced from~\cite{sparaciari_2017}. There the authors look at transformations of the form
    \bb
    \rho^{\otimes n} \longrightarrow \rho^{\otimes n} \otimes \eta_n \longrightarrow \Tr_2 \left[ U \left(\rho^{\otimes n} \otimes \eta_n\right) U^\dag \right] ,
    \ee
    where $U$ is required to be an energy-preserving unitary, $\eta_n$ is a `small' and `not too energetic' ancilla, and $\Tr_2$ denotes partial trace over the ancillary modes. Although the setting of that work is slightly different than the one considered here or in~\cite{Brandao-thermo}, the conclusions are nevertheless insightful: 
    in~\cite{sparaciari_2017} it is found that, although fluctuations (II) are indeed necessary to enable the reversibility of the theory, ancillae with subexponential dimension $\log \dim\eta_n = O\big(\sqrt{n\log n}\big)$ and sublinear Hamiltonian operator norm $\|H\|_\infty=O\big(n^{2/3}\big)$ suffice to established the desired result. 
    
    Since they play such a key role in thermodynamics, it is certainly meaningful to try to take into account fluctuations of this type in entanglement theory as well. However, although it may not be apparent at first sight, these fluctuations are already implicitly included in 
   our framework --- and in fact, in most of the commonly encountered formulations of entanglement theory. They take the form of sublinear fluctuations in the number of entanglement bits consumed by either the distillation or the dilution process. But in the definitions of $E_d$~\eqref{distillable} and $E_c$~\eqref{cost} we only care about asymptotic rates, so sublinear fluctuations are suppressed in the limit.
    
    In other words: suppose that in our original definitions in Section~\ref{subsec_distillable_cost} we chose to consider a larger class of transformations than $\K$-preserving ones; a general transformation of this new class would be obtained by: (a)~attaching an ancilla of $o(n)$ many qubits per party initialised in any (possibly entangled) state, and (b)~performing a $\K$-preserving operation on the joint system. In this way, the allowed transformations are of the form
\bb
\rho^{\otimes n} \longrightarrow \rho^{\otimes n}\otimes \tau_n \longrightarrow \Lambda_n\left( \rho^{\otimes n} \otimes \tau_n \right) ,
\ee
where $\Lambda_n\in \kp$ is an arbitrary $\K$-preserving operation, and $\tau_n = \tau^{A'_nB'_n}_n$ is an ancillary state over a system $A'_nB'_n$. Our assumptions on this ancilla are as follows:
\begin{enumerate}[(i)]
\item $A'_nB'_n$ can be an arbitrary bipartite system, possibly dependent on $n$, but it has to have subexponential dimension, i.e.\ $\dim (A'_nB'_n) = \dim \tau_n = 2^{o(n)}$; that is, $\log \dim(A'_nB'_n) = o(n)$ is required to be sublinear, or in other words $A'_nB'_n$ should be made of sublinearly many qubits;
\item Apart from this, $\tau_n$ is completely arbitrary. It could be for instance a maximally entangled state composed of sublinearly many ebits.
\end{enumerate}

Formally, the new distillable entanglement and entanglement cost now look like this:
\begin{align}
\widetilde{E}_{d,\,\kp}(\rho_{AB}) &\coloneqq \sup\left\{R>0:\, \lim_{n\to\infty} \inf_{\Lambda_n \in \kp} \frac12 \left\| \Lambda_n\left( \rho_{AB}^{\otimes n}\! \otimes\! \tau_n \right) - \Phi_2^{\otimes \ceil{Rn}} \right\|_1\! =0,\ \lim_{n\to\infty}\! \frac{\log \dim \tau_n}{n} = 0 \right\} , \label{distillable_new} \\[.5ex]
\widetilde{E}_{c,\,\kp}(\rho_{AB}) &\coloneqq \inf\left\{R>0:\, \lim_{n\to\infty} \inf_{\Lambda_n \in \kp} \frac12 \left\| \Lambda_n\left( \Phi_2^{\otimes \floor{Rn}} \otimes \tau_n \right) - \rho_{AB}^{\otimes n} \right\|_1\! = 0,\ \lim_{n\to\infty}\! \frac{\log \dim \tau_n}{n} = 0 \right\} . \label{cost_new}
\end{align}
Here $\dim \tau_n$ denotes the dimension of the $A'_nB'_n$ system of the ancilla.

We can now ask ourselves: can it be that $\widetilde{E}_{d,\,\kp}(\rho_{AB}) > E_{d,\,\kp}(\rho_{AB})$ or $\widetilde{E}_{c,\,\kp}(\rho_{AB}) < E_{c,\,\kp}(\rho_{AB})$? In other words, can it happen that granting a sublinear number of ancillary qubits allows for better rates in either entanglement distillation or entanglement dilution? The answer turns out to be negative, and the fundamental reason is that ancillae made of a sublinear number of qubits cannot change the rate. In full detail, a proof can be constructed as follows.

\begin{lemma}
For any state $\rho_{AB}$, it holds that
\bb
\widetilde{E}_{d,\,\kp}(\rho_{AB}) = E_{d,\,\kp}(\rho_{AB})\qquad \text{and}\qquad \widetilde{E}_{c,\,\kp}(\rho_{AB}) = E_{c,\,\kp}(\rho_{AB})\, .
\ee
\end{lemma}

\begin{proof}
It is clear that $\widetilde{E}_{d,\,\kp}(\rho_{AB}) \geq E_{d,\,\kp}(\rho_{AB})$ and $\widetilde{E}_{c,\,\kp}(\rho_{AB}) \leq E_{c,\,\kp}(\rho_{AB})$, so we only need to prove the converse inequalities. To do so, assume that $R$ is an achievable rate for $\widetilde{E}_{c,\,\kp}(\rho_{AB})$, so that we can find a sequence of ancillae $\tau_n$ and $\K$-preserving operations $\Lambda_n\in \kp$ with the property that
\bb
\frac12 \left\| \Lambda_n\left( \Phi_2^{\otimes \floor{Rn}} \otimes \tau_n \right) - \rho_{AB}^{\otimes n} \right\|_1\tendsn{} 0\, .
\label{error}
\ee
Fix $\delta>0$, and take $n$ large enough so that $2^{n\delta} \geq \dim \tau_n$. Since any state can be prepared from the maximally entangled state of the same dimension via LOCC, we can prepare $\tau_n$ starting from $\Phi_2^{\otimes \ceil{n\delta}}$ and using an LOCC $\Xi_n$, i.e.\ consuming $\ceil{n\delta}$ many ebits. In formula,
\bb
\Xi_n\left( \Phi_2^{\otimes \ceil{\delta n}} \right) = \tau_n\, .
\ee
Now, consider the operation $\Lambda'_n \coloneqq \Lambda_n \circ \left(I^{\otimes \floor{Rn}}\otimes \Xi_n\right)$, where $I^{\otimes \floor{Rn}}$ denotes the identity channel on the first $\floor{Rn}$ qubits. Note that $\Lambda'_n$ is still $\K$-preserving, because the set of $\K$-preserving operations is closed under composition, and LOCCs are $\K$-preserving. Moreover, it satisfies that
\bb
\Lambda'_n\left( \Phi_2^{\otimes (\floor{Rn} +\ceil{\delta n})} \right) = \Lambda_n \left( \Phi_2^{\otimes \floor{Rn}} \otimes \Xi_n\left( \Phi_2^{\otimes\ceil{\delta n}} \right) \right) = \Lambda_n \left( \Phi_2^{\otimes \floor{Rn}} \otimes \tau_n \right) .
\ee
Thus, thanks to~\eqref{error}
\bb
\frac12 \left\| \Lambda'_n\left( \Phi_2^{\otimes (\floor{Rn} +\ceil{\delta n})} \right) - \rho_{AB}^{\otimes n} \right\|_1\tendsn{} 0\, .
\ee
Since now there is no ancilla, we obtain that the rate
\bb
\lim_{n\to\infty} \frac{\floor{Rn} + \ceil{\delta n}}{n} = R+\delta
\ee
is actually achievable for $E_{c,\kp}(\rho_{AB})$. Therefore,
\bb
E_{c,\kp}(\rho_{AB}) \leq R + \delta\, .
\ee
Taking the infimum over $R$ yields that $E_{c,\kp}(\rho_{AB}) \leq \widetilde{E}_{c,\kp}(\rho_{AB}) + \delta$. Since $\delta>0$ was arbitrary, we obtain that in fact
\bb
E_{c,\kp}(\rho_{AB}) \leq \widetilde{E}_{c,\kp}(\rho_{AB}) \, ,
\ee
which is what we wanted to show.

To show that $\widetilde{E}_{d,\,\kp}(\rho_{AB}) \leq E_{d,\,\kp}(\rho_{AB})$, assume that $R$ is an achievable rate for $\widetilde{E}_{d,\,\kp}(\rho_{AB})$, i.e.\ that there exists a valid choice of $(\tau_n)_n$ and $\K$-preserving operations $\Lambda_n\in \kp$ such that
\bb
\frac12 \left\| \Lambda_n\left( \rho^{\otimes n} \otimes \tau_n \right) - \Phi_2^{\otimes \ceil{Rn}} \right\|_1\tendsn{} 0\, .
\ee
Noting that $\tau_n$ is a finite-dimensional system, the optimal decomposition for the standard robustness in~\eqref{std_rob} exists, and we can write
\begin{equation}\begin{aligned}
    \tau_n = \left( 1 + R^s_\K (\tau_n)\right) \sigma_{n}^+ - R^s_\K(\tau_n) \, \sigma_n^-
\end{aligned}\end{equation}
for some normalised states $\sigma_n^\pm \in \K$. Observe then that for any state $\sigma \in \K_{A^nB^n}$, it holds that
\begin{equation}\begin{aligned}
    \Lambda_n(\sigma \otimes \tau_n) &= \left( 1 + R^s_\K (\tau_n) \right) \Lambda_n \big( \sigma \otimes \sigma_n^+ \big) - R^s_\K(\tau_n) \, \Lambda_n \big( \sigma \otimes \sigma_n^- \big)\\
    &= \left( 1 + R^s_\K (\tau_n) \right) \sigma' - R^s_\K(\tau_n) \, \sigma''
\end{aligned}\end{equation}
for some states $\sigma',\sigma'' \in \K$ due to the fact that $\K$ is closed under tensor product. This implies that
\begin{equation}\begin{aligned}
    R^s_\K \left( \Lambda_n(\sigma \otimes \tau_n) \right) \leq R^s_\K (\tau_n) \leq \dim \tau_n - 1,
\end{aligned}\end{equation}
where the last inequality is due to the fact that the value of $R^s_\K(\Phi_{\dim \tau_n}) = \dim \tau_n - 1$ is maximal among states of the same dimension.
As $\dim \tau_n = 2^{o(n)}$ by assumption, this entails that each $\Lambda_n(\cdot \otimes \tau_n)$ is an $(R^s_\K, 2^{o(n)})$-approximately $\K$-preserving operation (cf. Supplementary Note~\ref{generation_note}). The proof is thus concluded by recalling from Lemma~\ref{lem:distillation_allequal} that such operations cannot increase the rate of distillation compared to exactly $\K$-preserving operations, that is, $R \leq E_{d,\,\kp}(\rho)$.
\end{proof}

    In the context of quantum thermodynamics, we can compare this to works which studied the manipulation of quantum systems under so-called Gibbs-preserving operations (e.g.~\cite{Faist2019}), defined to be all channels $\Lambda$ such that $\Lambda(\gamma) = \gamma'$ where $\gamma,\gamma'$ are the equilibrium states of the input and output system, respectively. Although no explicit resource fluctuations are allowed in such a definition, the Gibbs-preserving framework recovers the exact same asymptotic rates as frameworks that do consider fluctuations in the consumed resources~\cite{Brandao-thermo, sparaciari_2017} --- this is precisely because the type-II fluctuations are implicit there, as one only looks at the rates.
    
    \item \emph{Genuine fluctuations in the fundamental physical laws governing the process.} These fluctuations are excluded in most of the literature on quantum thermodynamics: for example, in Refs.~\cite{Brandao-thermo, sparaciari_2017} the unitary operators acting on the joint system (including the ancilla) are required to be \emph{exactly} and not only approximately energy-preserving. This is well justified from a physical perspective: analogously, we would not claim that energy is only approximately preserved in a piece of uranium because we lost track of some neutrinos; include those back into the picture, and you will restore exact energy preservation. It should be noted at this point that although the law of energy conservation is, to the best of our knowledge, exactly obeyed, this fact alone does not play a decisive role for the validity of thermodynamics, which is a macroscopic rather than microscopic theory. Still, the fact that it is conceivable that Nature preserves energy exactly shows that a setting excluding genuine type-III fluctuations --- yet encompassing type-I and type-II fluctuations --- is somewhat reasonable, if not completely satisfactory.
    
    Importantly, in most of the known physical theories including thermodynamics, disallowing type-III fluctuations does not change the underlying physics whatsoever. Although they \emph{can} in principle be considered, as we elaborate below,
    genuine fluctuations of the physical laws governing thermodynamical processes --- such as global unitarity and energy conservation~\cite{Brandao-thermo} --- are 
    not necessary to construct a reversible theory of quantum thermodynamics~\cite{Brandao-thermo, Faist2019, Faist2019-technical, Faist2019b}.

    A possible objection to shis claim could be that the maps considered in~\cite{Brandao-thermo,sparaciari_2017}
    are in fact not exactly but only approximately Gibbs-preserving. While this could seem \emph{prima facie} an example of a type-III fluctuation, it is only a spurious one, essentially because of the aforementioned fact that these maps can be thought of as Gibbs-preserving operations acting on a larger quantum system with an ancilla of sublinear size. These apparent type-III fluctuations are thus in fact type-II fluctuations in disguise, and as such they are included also in our original framework.
     
    What our results show in this context is that Brand\~{a}o and Plenio's asymptotically non-entangling operations~\cite{BrandaoPlenio1, BrandaoPlenio2} admit fluctuations that are \emph{genuinely} of type III (see Section~\ref{generation_note}). 
    Specifically, they cannot be re-absorbed into the type-II category --- if this were the case, then the entanglement cost $E_{c,\, \mathrm{ANE}}$ under asymptotically non-entangling operations~\cite[Definition~III.2]{BrandaoPlenio2} would necessarily be equal to that under non-entangling operations, which we denoted by $E_{c,\, \sepp}$. Since the former is given by the regularised relative entropy of entanglement%
    \footnote{Here we note that, despite the issue with the proof of~\cite{BrandaoPlenio2} uncovered in~\cite{berta_gap}, this part of the argument by Brand\~ao and Plenio is not affected. The proof can also be obtained by combining the framework of~\cite{BrandaoPlenio2} with the asymptotic equipartition property of $R^g_\SEP$, for which an alternative proof is given in ~\cite{datta_2009-2}.}
    ~\cite{BrandaoPlenio2, datta_2009-2}, in formula $E_{c,\, \mathrm{ANE}} = E_{r,\SEP}^\infty$, we would deduce that the entanglement cost under non-entangling operations must also equal $E_{c,\,\sepp} = E_{r,\SEP}^\infty$. But that this is not the case in general is precisely the content of our main result (Theorem~\ref{irreversibility_thm}), whose proof reveals that $E_{c,\, \sepp}(\omega_3) = 1 > \log_2 \frac32 = E_{r,\, \SEP}^\infty(\omega_3)$.

    Nevertheless, one can argue that allowing type-III fluctuations is actually the sensible thing to do in entanglement manipulation, as it allows one to achieve an even greater level of generality and a stricter adherence to the spirit of thermodynamics. 
    This is especially important when discussing no-go results such as asymptotic irreversibility: if reversibility could be restored with just an unequivocally vanishing amount of fluctuations --- even ones of type III --- then, arguably, such irreversibility would not be robust. This is precisely the motivation for our Theorem~\ref{irreversibility_thm2}, where we clarify that it is impossible to have any reversible framework of entanglement that generates only small amounts of it: using as an entanglement quantifier either the standard robustness of entanglement or the negativity, we show that these `fluctuations' must in fact be macroscopically large, putting into question the physicality of any conceivable reversible theory of entanglement.

\end{enumerate}

What is pivotal to understanding the consequences of our results is that thermodynamics is a reversible theory even when fluctuations of type III are not allowed. This is addressed in works such as Ref.~\cite{Brandao-thermo,Faist2019-technical, Faist2019b}, where the manipulation of quantum states under 
thermal operations was considered, and extended also to quantum channels in~\cite{Faist2019} under the Gibbs-preserving framework. That is, fluctuations of types I and II fully suffice to recover the entropy as the unique quantity governing the thermodynamical transformations of macroscopic systems, and the physical constraints of energy conservation can be enforced at all scales, including microscopic ones. In fact, we stress that in the Gibbs-preserving framework ---  which is completely analogous to the 
non-entangling operations used in our work --- even fluctuations of type I are sufficient to establish the reversibility of thermodynamics~\cite{Faist2019}, but one can equivalently consider type-I and type-II fluctuations.

Therefore, our entanglement manipulation framework under non-entangling operations --- which does allow type-I and type-II but forbids type-III fluctuations --- follows exactly the same reasoning as the quantum thermodynamics frameworks of Refs.~\cite{horodecki_2013, Brandao-thermo, weilenmann_2016, thermo-review, yungerhalpern_2016, Faist2019, sparaciari_2017, Faist2019-technical, Faist2019b}, and yet it leads to a completely opposite conclusion, as per Theorem~\ref{irreversibility_thm}. It strikes us as surprising the stark contrast between thermodynamics, which is asymptotically reversible, and entanglement theory, which turns out to be fundamentally irreversible under the same assumptions, no matter how much one strives to avoid it.

That is not to say that there is something inherently wrong with the framework proposed by Brand\~ao and Plenio~\cite{BrandaoPlenio2} --- quite the contrary, there is no intrinsic reason why fluctuations of type~III \emph{cannot} be allowed, even if they are not needed in thermodynamics. However, due to the ambiguity in defining `small' amounts of entanglement in the asymptotic limit, even such permissive fluctuations are not enough in light of our Theorem~\ref{irreversibility_thm2}, unless they are made non-vanishingly large. 
Our work therefore provides insight into how any potential reversible entanglement framework would function and the size of `fluctuations' it would need to allow.

In conclusion, although connections between entanglement theory and thermodynamics such as the one conjectured in~\cite{BrandaoPlenio2} can 
possibly be 
established, our work conclusively shows that 
reversibility of entanglement, if at all possible,
cannot 
play out in the same way as it 
does in the theory of thermodynamics 
(with type-I and type-II fluctuations, but without type-III fluctuations), or indeed not even 
in a similar way (with universally small type-III fluctuations). 
A fundamental and inexorable difference thus exists between the two theories.

\section{Further considerations} \label{further_considerations_note}

\subsection{Confessions of \texorpdfstring{$\boldsymbol{\omega_3}$}{omega\_3}}

Here we explain the intuition behind our choice of the state $\omega_3$ as the counterexample to the reversibility of entanglement manipulation (Theorems~\ref{irreversibility_thm} and~\ref{irreversibility_thm2}). 
It will rely on the relation between two entanglement monotones that we have encountered in the course of this work --- the standard robustness $R^s_\K$ and the generalised robustness $R^g_\K$.
Crucially, the work of Brand\~{a}o and Plenio~\cite{BrandaoPlenio2} connected each of these quantities with the operational tasks of entanglement distillation and dilution. In particular, Ref.~\cite{BrandaoPlenio2} (cf.~\cite{brandao_2011}) showed that, when considering only a single copy of a state rather than an asymptotic rate, the entanglement cost under $\K$-preserving operations is given exactly by the logarithm of $1+R^s_\K(\rho)$. This then implies that the asymptotic exact cost of entanglement (Eq.~\eqref{cost_exact}) is given by
\begin{equation}
E_{c,\,\kp}^\mathrm{exact} (\rho) = \lim_{n \to \infty} \frac{1}{n} \log_2 \left( 1+R^s_\K\left(\rho^{\otimes n}\right) \right).
\label{cost_exact_regularised_rob}
\end{equation}
The entanglement cost $E_{c,\,\kp}$ also needs to account for an asymptotically vanishing transformation error, and it can thus be expressed by suitably `smoothing' the robustness over the $\epsilon$-ball around a given state, i.e.
\begin{equation}\label{smoothed_Rs}
E_{c,\,\kp} (\rho) = \lim_{\epsilon \to 0^+} \lim_{n \to \infty} \frac{1}{n} \inf_{\substack{\rho' \in \D(\HH_{AB}^{\otimes n}),\\ \frac12 \left\|\rho' - \rho^{\otimes n}\right\|_1 \leq \epsilon}} \log_2 \left( 1+R^s_\K(\rho') \right).
\end{equation}
On the other hand, employing a connection established between the generalised robustness $R^g_\K$ and the regularised relative entropy $E_{r,\K}^\infty$~\cite{BrandaoPlenio2,datta_2009-2}, 
the distillable entanglement can be tightly upper bounded as~\cite{BrandaoPlenio2,hayashi_book}
\begin{equation}\label{smoothed_Rg}
E_{d,\,\kp} (\rho) \leq E_{r,\,\K}^\infty(\rho_{AB}) = \lim_{\epsilon \to 0^+} \lim_{n \to \infty} \frac{1}{n} \inf_{\substack{\rho' \in \D(\HH_{AB}^{\otimes n}), \\ \frac12 \left\|\rho' - \rho^{\otimes n}\right\|_1 \leq \epsilon}} \log_2 \left( 1+R^g_\K(\rho') \right) .
\end{equation}
These expressions lead to a very curious fact: showing the irreversibility of entanglement manipulation under the class of $\K$-preserving maps 
can be done by exhibiting a gap between the smoothed regularisations of $R^s_\K$ and $R^g_\K$ in~\eqref{smoothed_Rs} and~\eqref{smoothed_Rg}.

The daunting expressions for the regularised robustness measures do not immediately make studying reversibility any easier. Let us then start by asking a more basic question: is there a gap between the standard and the generalised entanglement robustness, $R^s_\SEP$ and $R^g_\SEP$? We stumbled upon the state $\omega_3$ precisely when looking for a way to demonstrate such a gap. Indeed, we have already seen in the course of proving Theorem~\ref{irreversibility_thm} that
\begin{equation}\begin{aligned}
R^s_\SEP (\omega_3) = \frac{3}{4}.
\end{aligned}\end{equation}
It is also not difficult to show that the decomposition $\omega_3 + \frac{1}{2} \Phi_3 = \frac{1}{2} P_3$ is optimal for the generalised robustness of this state, giving
\begin{equation}\begin{aligned}
R^g_\SEP (\omega_3) = \frac{1}{2}.
\end{aligned}\end{equation}
We thus see an explicit gap between the two robustness measures. The direct connections between $R^s_\SEP$ and $R^g_\SEP$ on one side and, respectively, the entanglement cost and distillable entanglement on the other then motivated us to look into $\omega_3$ as a possible candidate for a state whose cost could be strictly larger than the entanglement that can be distilled from it. 
This leads directly to the considerations expounded in the main text of the paper and to the main results of this work.

\subsection{Strong converses and error-rate trade-offs} \label{strong_converses_note}

According to Definition~\ref{reversibility_def}, the theory of entanglement manipulation is considered reversible under some set of operations if the distillable entanglement and the entanglement cost coincide for all states. As we have seen, these latter two quantities are defined in terms of asymptotic transformations in which the allowed error is required to vanish as the number of copies of the state grows. In information theory, classical as well as quantum, one can study also weaker notions of transformation rates, dubbed \emph{strong converse rates}. In this setting, one requires instead that the error, as measured by half of the trace norm distance between the output and the target state, is bounded away from its maximum value of $1$. 
Intuitively, what this means is that attempting to distill or dilute entanglement at a rate larger than the strong converse one necessarily incurs an error that grows to $1$, making the protocol useless.
Formally, we can define the strong converse rates of distillation and dilution, respectively, by
\begin{equation}
E_{d,\,\kp}^\dag(\rho_{AB}) \coloneqq \sup_{\epsilon\in [0,1)} E_{d,\,\kp}^\epsilon(\rho_{AB})\, ,\qquad E_{c,\,\kp}^\dag(\rho_{AB}) \coloneqq \inf_{\epsilon\in [0,1)} E_{c,\,\kp}^\epsilon(\rho_{AB})\, .
\end{equation}
From these definitions and from the discussion that inspired them it is clear that $E_{d,\,\kp}^\dag(\rho)\geq E_{d,\,\kp}(\rho)$ and $E_{c,\,\kp}^\dag(\rho) \leq E_{c,\,\kp}(\rho)$ for all states $\rho$.

We could thus wonder whether there is a possibility of achieving a larger distillable entanglement or a lower entanglement cost, perhaps even restoring a (substantially weaker) form of reversibility, by passing to the corresponding strong converse rates. In the case of entanglement distillation from the state $\omega_3$ defined in~\eqref{omega_3}, such possibility can be ruled out by exploiting a result of Hayashi~\cite[Theorem~8.7]{hayashi_book}~(see also~\cite{BrandaoPlenio2}), who showed that the regularised relative entropy distance from the set of states in $\K$ defined in~\eqref{Er_regularised} is also an upper bound  on the strong converse distillable entanglement, i.e.\ $E_{d,\,\kp}(\rho_{AB}) \leq E_{d,\,\kp}^\dag(\rho_{AB}) \leq E_{r,\, \K}^\infty(\rho_{AB})$ (cf.~\eqref{eq:distillable_upper_Er}). Our result in Theorem~\ref{irreversibility_thm} then implies that $E_{d,\, \kp}(\omega_3) = E_{d,\,\kp}^\dagger(\omega_3) = \frac32$. Indeed, coupled with Lemma~\ref{lem:distillation_allequal}, the same strong converse bound holds for distillation with operations that generate sub-exponential amounts of entanglement.

As for the entanglement cost, we have not yet been able to prove that no improvement can be obtained by considering the corresponding strong converse rate. However, we deem such a possibility quite unlikely, 
and we present strong evidence against it by showing that any such error would have to be very large, which effectively rules out this approach as a practically viable way of improving entanglement dilution. We leave the complete solution of this problem, which we formalise below, as an open question left for future work.

\begin{cj} \label{strong_converse_Ec_cj}
For $\K=\SEP,\PPT$, the two-qutrit state $\omega_3$ defined by~\eqref{omega_3} satisfies that
\begin{equation}
    \inf_{\epsilon \in [0,1)} E_{c,\, \kp}^\epsilon (\omega_3) = E_{c,\, \kp}^\dag(\omega_3) = 1\, .
\label{strong_converse_Ec}
\end{equation}
\end{cj}
If the above conjecture is true, then we see that 
\begin{equation}
E_{d,\, \kp}(\omega_3) = E_{d,\, \kp}^\dag(\omega_3) = \log_2 \frac32 < 1 = E_{c,\, \kp}^\dag(\omega_3) = E_{c,\, \kp}(\omega_3)\, ,
\end{equation}
i.e.\ the gap between distillable entanglement and entanglement cost remains even upon passing to the strong converse exponents.

Theorem~\ref{irreversibility_thm} already guarantees that 
\bb
\inf_{\epsilon \in \left[0,\, \epsilon_{\max}\right)} E_{c,\, \kp}^\epsilon (\omega_3) = 1
\label{epsilon_max}
\ee
holds for $\epsilon_{\max}=1/2$. The problem we are faced with now is how to increase the value of $\epsilon_{\max}$ in~\eqref{epsilon_max}. In general, we would like to study more in depth the behaviour of the function $E_{c,\, \kp}^\epsilon (\omega_3)$ for values of $\epsilon$ close to $1$. A partial result is as follows:

\begin{lemma} \label{error_rate_tradeoff_lemma}
For all $\delta\in (0,1]$, it holds that
\bb
\inf_{\epsilon\in [0,\, 1-\delta)} E_{c,\, \sepp}^\epsilon(\omega_3) \geq \inf_{\epsilon\in [0,\, 1-\delta)} E_{c,\, \pptp}^\epsilon (\omega_3) \geq \left\{ \begin{array}{ll} 1 & \text{if $1/3\leq \delta\leq 1$,} \\[1ex] \log_2 \frac{3(1-\delta)}{2-3\delta} & \text{if $0< \delta< 1/3$.} \end{array} \right.
\label{}
\ee
\end{lemma}

Let us discuss the consequences of this result. First, we see that performing entanglement dilution at the rate $1$ --- which we know is optimal for $E_{c,\,\sepp}$ --- must incur an error of at least $\frac23$ asymptotically. This means that we can in fact take $\epsilon_{\max}=2/3$ in~\eqref{epsilon_max}, which already improves upon the error $\epsilon_{\max}=1/2$ that we obtained in Theorem~\ref{irreversibility_thm}. Furthermore, attempting to perform dilution of $\omega_3$ at a rate below $1$ must yield an even larger error; the smaller the rate, the greater the error. One commonly refers to such behaviour as an \emph{error-rate trade-off}.

\begin{proof}[Proof of Lemma~\ref{error_rate_tradeoff_lemma}]
Let $R$ be an achievable rate for the entanglement cost of $\omega_3$ at error threshold $\epsilon<1-\delta$. That is, let there exist a sequence of non-entangling operations $\Lambda_n\in \sepp_{A_0^{\floor{Rn}}B_0^{\floor{Rn}}\to A^nB^n}$, with $A_0,B_0$ being single-qubit systems, such that
\bb
\Omega_n \coloneqq \Lambda_n\left( \Phi_2^{\otimes \floor{Rn}}\right),\qquad \epsilon_n\coloneqq \frac12 \left\| \Omega_n - \omega_3^{\otimes n}\right\|_1\, ,\qquad \limsup_{n\to\infty}\epsilon_n < 1-\delta\, .
\label{pretty_eq0}
\ee
For $\K=\SEP, \PPT$, we have that
\bb
2^{\floor{Rn}} &\textgeq{(i)} 1 + R^s_{\K} \left( \Omega_n\right) \\
&\texteq{(ii)} 1 + \sup_{\Omega'} \tsr_{\K} \left( \Omega_n\, \big|\, \Omega'\right) \\
&\textgeq{(iii)} \frac12 \sup_{\Omega'} \left\{ \tn \left( \Omega_n\, \big|\, \Omega'\right) + 1 \right\} \\
&\geq \frac12\, \tn \left( \Omega_n\, \big|\, \omega_3^{\otimes n} \right) .
\label{pretty_eq1}
\ee
Here, the inequality in~(i) is just a rephrasing of that in step~(ii) of~\eqref{epsilon_lemma_at_work}. In the subsequent lines of~\eqref{pretty_eq1}, the identity in~(ii) and the inequality in~(iii) follow from Proposition~\ref{elementary_prop}, items~(a) and~(d), respectively.

To continue our argument, we need to lower bound $\tn \left( \Omega_n\, \big|\, \omega_3^{\otimes n} \right)$, as usual. While in the case of vanishing error the operator $X_3$ defined by~\eqref{X_3} was the right ansatz for the optimisation~\eqref{tn} defining the tempered negativity, this time we will make a more sophisticated choice. Using the notation defined in~\eqref{P_3_and_Phi_3}, set
\bb
X_3(\delta) \coloneqq \left\{ \begin{array}{ll} 2P_3 - 3\Phi_3 & \text{if $1/3\leq \delta\leq 1$,} \\[1ex] \frac{3}{2-3\delta}\left((1-\delta) P_3 - \Phi_3 \right) & \text{if $0< \delta< 1/3$.} \end{array}\right.
\label{X_3_delta}
\ee
We now verify that (a)~$\left\|X_3(\delta)^\Gamma\right\|_\infty=1$ and (b)~$\left\|X_3(\delta)\right\|_\infty = \Tr X_3(\delta)\omega_3$. Thanks to~\eqref{irreversibility_proof_eq5}--\eqref{irreversibility_proof_eq6}, we can limit ourselves to the case where $0<\delta< 1/3$. As for~(a), we have that
\bb
\left\|X_3(\delta)^\Gamma\right\|_\infty = \frac{3}{2-3\delta} \left\| (1-\delta) P_3 - \frac13 F_3\right\|_\infty = \frac{3}{2-3\delta} \max\left\{ \frac23 - \delta,\, \frac13 \right\} = 1\, ,
\ee
where $F_3=\sum_{j,k=1}^3 \ketbraa{jk}{kj}$ is as usual the flip operator, and the second equality is elementary because $P_3$ and $F_3$ commute. Concerning~(b), it suffices to observe that
\bb
\left\|X_3(\delta)\right\|_\infty = \frac{3}{2-3\delta} \max\left\{ \delta,\, 1-\delta \right\} = \frac{3(1-\delta)}{2-3\delta} = \Tr X_3(\delta)\omega_3\, ,
\label{X_3_delta_operator_norm}
\ee
where again the computation of the operator norm is elementary because $P_3$ and $\Phi_3$ commute. From the above claims~(a) and~(b), we deduce immediately that also (a')~$\left\|\left(X_3(\delta)^{\otimes n}\right)^\Gamma\right\|_\infty=1$ and (b')~$\left\|X_3(\delta)^{\otimes n}\right\|_\infty = \Tr \left[ X_3(\delta)^{\otimes n} \omega_3^{\otimes n}\right]$. This in turn ensures that $X_3(\delta)^{\otimes n}$ is a feasible ansatz for the optimisation that defines $\tn \left( \Omega_n\, \big|\, \omega_3^{\otimes n} \right)$ (cf.~\eqref{tn}). Hence,
\bb
\tn \left( \Omega_n\, \big|\, \omega_3^{\otimes n} \right) \geq \Tr\left[ X_3(\delta)^{\otimes n} \Omega_n\right] .
\label{pretty_eq2}
\ee
To further lower bound the right-hand side, we need to have a closer look at the spectral structure of $X_3(\delta)^{\otimes n}$. Let us distinguish the two cases $0<\delta< 1/3$ and $1/3\leq \delta\leq 1$. In the former case, according to~\eqref{X_3_delta_operator_norm} the largest eigenvalue of this operator, $\frac{3^n(1-\delta)^n}{(2-3\delta)^n}$, corresponds to the eigenspace with projector $(P_3-\Phi_3)^{\otimes n}$. The smallest eigenvalue, instead, is easily seen to be $-\frac{3^n(1-\delta)^{n-1}\delta}{(2-3\delta)^n}$. We deduce that
\bb
X_3(\delta)^{\otimes n} &\geq \frac{3^n(1-\delta)^n}{(2-3\delta)^n} (P_3-\Phi_3)^{\otimes n} - \frac{3^n(1-\delta)^{n-1}\delta}{(2-3\delta)^n} \left( \id - (P_3-\Phi_3)^{\otimes n} \right) \\&= \frac{3^n(1-\delta)^{n-1}}{(2-3\delta)^n} \left( (P_3-\Phi_3)^{\otimes n} - \delta \id \right) .
\label{spectral_structure_X_3}
\ee
The above operator inequality is valid for $0<\delta< 1/3$. When $1/3\leq \delta\leq 1$, instead, an analogous reasoning yields more simply
\bb
X_3(\delta)^{\otimes n} \geq 2^n (P_3-\Phi_3)^{\otimes n} - 2^{n-1} \left( \id - (P_3-\Phi_3)^{\otimes n} \right) = 2^{n-1} \left( 3(P_3-\Phi_3)^{\otimes n} - \id\right) .
\label{stupid_spectral_structure_X_3}
\ee
Before completing the proof, let us observe that by virtue of~\eqref{pretty_eq0} we have
\bb
\Tr \left[ (P_3-\Phi_3)^{\otimes n} \Omega_n \right] = 1 + \Tr \left[(P_3-\Phi_3)^{\otimes n} \left( \Omega_n - \omega_3^{\otimes n}\right)\right] \geq 1-\epsilon_n\, .
\label{pretty_eq3}
\ee
Continuing from~\eqref{pretty_eq2}, when $0<\delta< 1/3$ we infer
\bb
\tn \left( \Omega_n\, \big|\, \omega_3^{\otimes n} \right) &\textgeq{(iv)} \frac{3^n(1-\delta)^{n-1}}{(2-3\delta)^n} \Tr\left[ \left( (P_3-\Phi_3)^{\otimes n} - \delta \id \right) \Omega_n \right] \\
&\textgeq{(v)} \frac{3^n(1-\delta)^{n-1}}{(2-3\delta)^n} \left( 1 - \delta - \epsilon_n \right) \\
&= \left( \frac{3(1-\delta)}{2-3\delta}\right)^n \left( 1- \frac{\epsilon_n}{1-\delta}\right) .
\ee
Here, the inequality in~(iv) is obtained by plugging~\eqref{spectral_structure_X_3} into~\eqref{pretty_eq2}, while that in~(v) comes from~\eqref{pretty_eq3}. Remembering~\eqref{pretty_eq1}, we obtain that
\bb
2^{\floor{Rn}} \geq \frac12 \left( \frac{3(1-\delta)}{2-3\delta}\right)^n \left( 1- \frac{\epsilon_n}{1-\delta}\right) .
\ee
Computing the logarithm of both sides, dividing by $n$, and using~\eqref{pretty_eq0} to take the limit for $n\to\infty$ yields $R\geq \log_2 \frac{3(1-\delta)}{2-3\delta}$. The reasoning for the case where $1/3\leq\delta\leq 1$ is entirely analogous, leading to $R\geq 1$. Since the rate $R$ was completely arbitrary, the proof is complete.
\end{proof}

\subsection{Implications for coherence theory}

The resource theory of quantum coherence~\cite{Baumgraz2014,coherence-review} is concerned with the study of the operational manipulation of the resource represented by superposition in a fixed basis $\{\ket{i}\}_i$ of a Hilbert space --- very much analogously to the way that entanglement has been studied in this work. Indeed, one can define the asymptotic quantities such as coherence cost $C_c$ and distillable coherence $C_d$ in a manner completely equivalent to the entanglement-based definitions of our work~\cite{Winter2016}. It is, in particular, known that the manipulation of coherence is asymptotically reversible under the class of maximally incoherent operations (MIO)~\cite{Winter2016,one-shot-c-dilution}, which are the coherence theory equivalent of non-entangling transformations: they map any incoherent (diagonal) state into another incoherent state. For such maps, it then holds that $C_{d,\mathrm{MIO}} (\omega) = C_{c,\mathrm{MIO}} (\omega)$ for any state.

A strong parallel between the theories of coherence and entanglement was noticed when the latter is restricted to the so-called maximally correlated states~\cite{Winter2016}, of the form
\begin{equation}\begin{aligned}
    \rho = \sum_{i,j} \rho_{ij} \ketbraa{ii}{jj}
\end{aligned}\end{equation}
for some bases of $\HH_A$ and $\HH_B$ and coefficients $\rho_{ij}$. Such states are separable if and only if they are diagonal in the given basis, and their operational properties --- including the values of entanglement cost and distillable entanglement under LOCC --- exactly match the values of coherence cost and distillable coherence of the corresponding single-party state
\begin{equation}\begin{aligned}
    \wt{\rho} = \sum_{i,j} \rho_{ij} \ketbraa{i}{j}
\end{aligned}\end{equation}
when coherence manipulation is considered under the class of incoherent operations (IO)~\cite{Baumgraz2014,Winter2016}.

Our result, however, shows a major difference between the theories of entanglement for maximally correlated states and quantum coherence, and breaks the quantitative equivalence when the transformations under non-entangling operations NE and maximally incoherent operations MIO are considered. Crucially, our counterexample $\omega_3$ is a maximally correlated state. This implies that
\begin{equation}\begin{aligned}
    E_{c,\sepp} (\omega_3) > E_{d,\sepp} (\omega_3) = C_{d,\mathrm{MIO}} (\wt{\omega}_3) = C_{c,\mathrm{MIO}} (\wt{\omega}_3),
\end{aligned}\end{equation}
where the first equality is due to the fact that both $E_{d,\sepp}(\rho)$ and $C_{d,\mathrm{MIO}}(\wt{\rho})$ are given by the relative entropy of coherence of $\wt{\rho}$~\cite{Winter2016,one-shot-c-dilution,Rains2001} for any maximally correlated state. Intuitively, any MIO transformation can be seen to give rise to a transformation which is non-entangling, but only when restricted to the maximally correlated subspace; our result shows that such maps cannot always be extended to a transformation which is non-entangling for any input state, and so NE operations are weaker at manipulating maximally correlated states than MIO operations are at manipulating their corresponding single-party systems.

\subsection{More open questions}

Our results motivate a number of extensions and follow-up results that would strengthen the understanding of general entanglement manipulation. We have already remarked several of them throughout this Supplementary Information; let us collect them here and discuss other open questions.

First, although our bound $E_{c,\, \sepp} (\rho) \geq E_N^\tau(\rho)$ is good enough to establish the irreversibility of entanglement manipulation, one could ask whether a tighter computable bound can be obtained. For instance, does the regularised quantity $\tl_{\!\K}(\rho)$ in Theorem~\ref{main_tool_thm} admit a single-letter expression? Even more ambitiously, one could ask whether an exact expression for the entanglement cost $E_{c,\, \sepp}$ itself can be established. Such questions are interesting not only from an axiomatic perspective, but also because any such result would provide improved bounds on $E_{c,\, \locc}$.

Additionally, as mentioned in Supplementary Note~\ref{main_results_note}, several quantum states which have previously been used as examples of irreversibility under smaller sets of operations are actually reversible under $\sepp$ and $\pptp$. An understanding of what exactly makes a state such as $\omega_3$ irreversibile under \emph{all} non-entangling transformations, and --- more generally --- a complete characterisation of all irreversible states could help shed light on the stronger type of irreversibility that we have revealed in this work.

Another result uncovered in Supplementary Note~\ref{generation_note} of our work is that, although Brand\~ao and Plenio's framework~\cite{BrandaoPlenio1, BrandaoPlenio2} suggests that entanglement can be reversibly manipulated while generating only small (asymptotically vanishing) amounts of the generalised robustness $R^g_\SEP$, requiring that the standard robustness $R^s_\SEP$ also be small completely breaks reversibility. It would be very interesting to understand exactly why such a `phase transition' in the task of entanglement dilution occurs, and whether one can tighten our and Brand\~ao and Plenio's results with other choices of monotones.

Also, the phenomenon of catalysis~\cite{Plenio-catalysis} is a remarkable feature that can significantly enhance feasible entanglement transformations: the fact that a transformation $\rho \to \omega$ is impossible does not necessarily mean that $\rho \otimes \tau \to \omega \otimes \tau$ cannot be accomplished with the same set of allowed processes. One could then ask about transformation rates where, instead of requiring that $\Lambda(\rho^{\otimes n}) \to \omega^{\otimes m}$, one asks that $\Lambda(\rho^{\otimes n} \otimes \tau) \to \omega^{\otimes m} \otimes \tau$ for some state $\tau$, and define the distillable entanglement $E_d$ and entanglement cost $E_c$ correspondingly. Additionally, the operations $\sepp$ and $\pptp$ are rather curious in that they are not closed under tensor product, in the sense that $\Lambda, \Lambda' \in \sepp$ does not mean that $\Lambda \otimes \Lambda' \in \sepp$ when acting on a larger system. Such properties leave open the possibility of significantly different behaviour when catalysis is employed, something that was already remarked in~\cite{BrandaoPlenio1, BrandaoPlenio2}.

Finally, a fundamentally important question is: what is it about entanglement that makes it irreversible? As discussed in the main text, several examples of quantum resource theories have been shown to be reversible when all resource--non-generating operations (counterparts to non-entangling or PPT-preserving maps) are allowed~\cite{RT-review}. Although there is no \emph{a priori} reason to expect reversibility to be a generic property of quantum resources, one can note that the framework and main results of Brand\~{a}o and Plenio 
may be adapted to more general convex resource theories~\cite{Brandao-Gour}, 
suggesting~\cite{berta_gap} that reversibility can 
be guaranteed when \emph{asymptotically} resource--non-generating transformations are allowed and the generated resource is quantified with a generalised robustness--type measure $R^g_\K$. Our question then reduces to: in what types of resources can we enforce strict resource non-generation and still maintain reversibility? Is entanglement truly unique in its general irreversibility?

\section{Subtleties related to infinite dimension} \label{infinite_dim_note}

Throughout this section we will elaborate on some issues that arise specifically in dealing with infinite-dimensional spaces.

\subsection{On the definition of partial transpose}

The definition of the partial transposition~\cite{PeresPPT} requires some further care when one deals with infinite-dimensional spaces, the main reason being that this operation does not preserve the space of trace class operators. We first define it on the dense subspace $\T(\HH_A)\otimes \T(\HH_B)\subseteq \T(\HH_{AB})$ of finite linear combinations of simple tensors by the expression $\Gamma(X_A\otimes Y_B) = (X_A\otimes Y_B)^\Gamma \coloneqq X_A \otimes Y_B^\intercal$, already encountered in~\eqref{PT_simple_tensors}, extended by linearity.

Now, one observes that the linear map $\Gamma:\T(\HH_A)\otimes \T(\HH_B) \to \B(\HH_{AB})$ we just constructed has operator norm
\begin{equation*}
    \left\|\Gamma\right\|_{1\to\infty} \coloneqq \sup_{0\neq Z\in \T(\HH_A)\otimes \T(\HH_B)} \frac{\left\| Z^\Gamma \right\|_\infty}{\|Z\|_1} = 1\, ,
\end{equation*}
where $\|Z\|_\infty \coloneqq \sup_{\ket{\psi}}\left\| Z\ket{\psi}\right\|$. To see this, it suffices to write the singular value decomposition of any $Z\in \T(\HH_A)\otimes \T(\HH_B)\subseteq \T\left( \HH_{AB} \right)$ as $Z = \sum_{i=0}^\infty \lambda_i \ketbraa{\Psi_i}{\Phi_i}$, where $\|Z\|_1 = \sum_{i=0}^\infty |\lambda_i|$, and then notice that $\left\|\ketbraa{\Psi}{\Phi}^\Gamma\right\|_\infty \leq 1$ for all pairs of pure states $\ket{\Psi},\ket{\Phi}\in \HH_{AB}$, so that
\begin{equation*}
    \left\|Z^\Gamma \right\|_\infty \leq \sum_{i=0}^\infty |\lambda_i| \left\| \ketbraa{\Psi}{\Phi}^\Gamma \right\|_\infty \leq 1\, .
\end{equation*}
Finally, one may use the continuous extension theorem to lift $\Gamma$ to a continuous map $\Gamma:\T(\HH_{AB}) \to \B(\HH_{AB})$. It still holds that $\left\|\Gamma\right\|_{1\to\infty}=1$.

\subsection{Topological properties of the cone of PPT operators}

An important property of the two cones of separable and of PPT operators is that they are closed with respect to the topology induced on $\T(\HH_{AB})$ by its pre-dual, the Banach space of compact operators on $\HH_{AB}$. We will refer to this topology as the \emph{weak*-topology}. The fact that $\SEP_{AB}$ is weak*-closed has been established in Ref.~\cite[Lemma~25]{taming-PRA}. An analogous statement for $\PPT_{\!AB}$ can be proved even more directly (cf.~\cite[Lemma~13]{achievability}).

\begin{lemma} \label{weak*_closed_PPT_lemma}
The cone $\PPT_{\!AB}\subseteq \T(\HH_{AB})$ defined by~\eqref{PPT} is weak*-closed, i.e.\ closed with respect to the topology induced on $\T(\HH_{AB})$ by its pre-dual, the Banach space of compact operators on $\HH_{AB}$. 
\end{lemma}

\begin{proof}
Note that $\PPT_{\!AB} = \Tp(\HH_{AB})\cap \left(\Gamma(\Tp(\HH_{AB})) \cap \T(\HH_{AB}) \right)$, where the intersection with $\T(\HH_{AB})$ in the second factor reminds us of the fact that $\Gamma(\Tp(\HH_{AB}))\subseteq \B(\HH_{AB})$ contains operators that are not of trace class. Since $\Tp(\HH_{AB})$ is well known to be weak*-closed, it suffices to show that so is $\Gamma(\Tp(\HH_{AB})) \cap \T(\HH_{AB})$ as well. Pick local orthonormal bases $\{\ket{n}_A\}_{n\in \N}$ and $\{\ket{m}_B\}_{m\in \N}$. Let us say that an operator $Y\in \T(\HH_{AB})$ has a finite expansion if $\braket{n,m|Y|n',m'}\neq 0$ only for a finite number of quadruples $(n,m,n',m')\in \N^4$. If this is the case, it is simple to verify that also the partial transpose $Y^\Gamma$ has a finite expansion; in particular, both $Y$ and $Y^\Gamma$ are compact operators.

We now claim that $X\in \T(\HH_{AB})$ satisfies that $X\in \Gamma(\Tp(\HH_{AB}))$ if and only if $\Tr\left[X Y^\Gamma\right]\geq 0$ for all $Y\geq 0$ with a finite expansion. To see why, note that a necessary and sufficient condition for $X^\Gamma\geq 0$ is that all (finite) principal minors of $X^\Gamma$ are positive. In particular, $X^\Gamma\geq 0$ if and only if $\Tr\left[ X^\Gamma Y\right] = \Tr\left[ X Y^\Gamma\right] \geq 0$ for all $Y\geq 0$ with a finite expansion. Since any such $Y^\Gamma$ is compact, the functionals $X \mapsto \Tr\left[ X Y^\Gamma\right]$ are weak*-continuous; we deduce immediately that $\Gamma(\Tp(\HH_{AB}))\cap \T(\HH_{AB})$ is weak*-closed. This concludes the proof.
\end{proof}


\end{document}